\documentclass[11pt]{article}

\usepackage{arxivbasic}

\usepackage{graphicx}
\usepackage[title]{appendix}  %% More flexible appendix formatting

\usepackage{amsmath, amsthm, amssymb} 
\allowdisplaybreaks % allow align equations to be split across page breaks

\newtheorem{theorem}{Theorem}

\newtheorem{corollary}{Corollary} % above says to follow Theorem counter. This gives Corollary it's own counter.
\newtheorem{lemma}{Lemma}  % or should it be ...[theorem]{Lemma}?

\theoremstyle{definition}
\newtheorem{definition}{Definition}
\newtheorem{example}{Example}[section]

\usepackage{dsfont} % For a double-stroke 1 to write indicator functions (aka the characteristic function for a set)
\newcommand{\1}[1]{\mathds{1}_{[#1]}}

\usepackage{enumitem}

\usepackage[font=small,labelfont=bf]{caption}

\usepackage[normalem]{ulem}

\usepackage[mathscr]{eucal}

\usepackage{microtype}
 \DisableLigatures{encoding = *, family = * }

\usepackage{hyperref}
\hypersetup{colorlinks=true, citecolor=black, linkcolor=black, urlcolor=black}

\usepackage[authoryear]{natbib} 
\bibpunct{(}{)}{;}{a}{}{,}
\usepackage{doi} \usepackage{doi} % Changes DOI to doi: and makes dois into clickable links in the References

\begin{document}

\title{Generalizations of the `Linear Chain Trick':\\ Incorporating more flexible dwell time distributions into mean field ODE models.}

\author{Paul J. Hurtado\\
	University of Nevada, Reno \\
	ORCID: 0000-0002-8499-5986 \\
	\texttt{phurtado@unr.edu}
	\And
	Adam S. Kirosingh \\
	Stanford University \\
	ORCID: 0000-0003-0500-9269
}

\date{Last Updated: \today}

\maketitle

% THIS NEEDS TO BE 250 WORDS OR LESS
\begin{abstract} 
	Mathematical modelers have long known of a ``rule of thumb" referred to as the Linear Chain Trick (LCT; aka the Gamma Chain Trick): a technique used to construct mean field ODE models from continuous-time stochastic state transition models where the time an individual spends in a given state (i.e., the dwell time) is Erlang distributed (i.e., gamma distributed with integer shape parameter). Despite the LCT's widespread use, we lack general theory to facilitate the easy application of this technique, especially for complex models. This has forced modelers to choose between constructing ODE models using heuristics with oversimplified dwell time assumptions, using time consuming derivations from first principles, or to instead use non-ODE models (like integro-differential equations or delay differential equations) which can be cumbersome to derive and analyze. Here, we provide analytical results that enable modelers to more efficiently construct ODE models using the LCT or related extensions. Specifically, we 1) provide novel extensions of the LCT to various scenarios found in applications; 2) provide formulations of the LCT and it's extensions that bypass the need to derive ODEs from integral or stochastic model equations; and 3) introduce a novel Generalized Linear Chain Trick (GLCT) framework that extends the LCT to a much broader family of distributions, including the flexible \textit{phase-type} distributions which can approximate distributions on $\mathbb{R}^+$ and be fit to data. These results give modelers more flexibility to incorporate appropriate dwell time assumptions into mean field ODEs, including conditional dwell time distributions, and these results help clarify connections between individual-level stochastic model assumptions and the structure of corresponding mean field ODEs. 

\end{abstract}

\clearpage
\tableofcontents
\clearpage

\section{Introduction} \label{Intro} 

Many scientific applications involve systems that can be framed as continuous time state transition models \citep[e.g., see][]{Strogatz2014}, and these are often modeled using mean field ordinary differential equations (ODE) of the form $$\frac{d\mathbf{x}}{dt}=f(\mathbf{x},\theta,t),$$ where $\mathbf{x}(t) \in \mathbb{R}^n$, parameters $\theta\in\mathbb{R}^p$, and $f:\; \mathbb{R}^n \mapsto \mathbb{R}^n$ is smooth. The abundance of such applications, and the accessibility of analytical and computational tools for analyzing ODE models, have made ODEs one of the most popular modeling frameworks in scientific applications.

Despite their widespread use, one shortcoming of ODE models is their inflexibility when it comes to specifying probability distributions that describe the duration of time spent in a given state. The basic options available for assuming specific dwell time distributions within an ODE framework can really be considered as a single option: the $1^\text{st}$  event time distribution for a (nonhomogeneous) Poisson process, which includes the exponential distribution as a special case. 

To illustrate this, consider the following SIR model of infectious disease transmission by \citet{Kermack1927},  

\begin{subequations}\label{eq:SIR}\begin{align}
	\frac{d}{dt}S(t) =&\; -\lambda(t)\,S(t) \label{eq:S}\\
	\frac{d}{dt}I(t) =&\; \lambda(t)\,S(t) - \gamma\,I(t) \label{eq:I} \\
	\frac{d}{dt}R(t) =&\; \gamma\,I(t) \label{eq:R} 
	\end{align} \end{subequations} where $S(t)$, $I(t)$, and $R(t)$ correspond to the number of susceptible, infected, and recovered individuals in a closed population at time $t$ and $\lambda(t)\equiv \beta\,I(t)$ is the per-capita infection rate (also called the \textit{force of infection} by \citet {AndersonMay1992} and others). This model can be thought of as the mean field model for some underlying stochastic state transition model where a large but finite number of individuals transition from state S to I to R (see \citet{Kermack1927} for a derivation, and see \citet{Armbruster2017}, \citet{Banks2013}, and references therein for examples of the convergence of stochastic models to mean field ODEs). 

Although multiple stochastic models can yield the same mean field deterministic model, it is common to consider a stochastic model based on Poisson processes. For the SIR model above, for example, a stochastic analog would assume that, over the time interval $[t,t+\Delta t]$ (for very small $\Delta t$), each individual in S or I at time $t$ is assumed to transition from S to I with probability $\lambda(t)\,\Delta t$, or from I to R with probability $\gamma\,\Delta t$, respectively. Taking $\Delta t \to 0$ yields the desired continuous time stochastic model. Here, the linear rate of transitions from I to R ($\gamma\,I(t)$) arises from assuming the dwell time for an individual in the infected state (I) follows an exponential distribution with rate $\gamma$ (i.e., the $1^\text{st}$ event time distribution for a homogeneous Poisson process with rate $\gamma$). Similarly, assuming the time spent in state S follows the $1^\text{st}$ event time distribution under a nonhomogeneous (also called inhomogeneous) Poisson process with rate $\lambda(t)$ yields a time-varying per capita transition rate $\lambda(t)$. This association of a mean field ODE with a specific underlying stochastic model provides very valuable intuition in an applied context. For example, it allows modelers to ascribe application-specific (e.g., biological) interpretations to parameters and thus estimate parameter values (e.g., for $\gamma$ above, the mean time spent infectious is $1/\gamma$), and it provides intuition and a clear mathematical foundation from which to construct and evaluate mean field ODE models based on individual-level, stochastic assumptions.

To construct models using other dwell time distributions, a standard approach is to formulate a continuous time stochastic model and from it derive mean field \textit{distributed delay equations}, typically represented as integro-differential equations (IDEs) or sometimes integral equations (IEs) \citep[e.g., see][]{Kermack1927, Hethcote1980, Feng2007, Feng2016}. Readers unfamiliar with IEs and IDEs are referred to \citet{Burton2005} or similar texts. IEs and IDEs have proven to be quite useful models in biology, e.g., they have been used to model chemical kinetics \citep{Roussel1996}, gene expression \citep{Smolen2000,Takashima2011,Guan2018}, physiological processes such as glucose-insulin regulation \citep[][and references therein]{Makroglou2006}, cell proliferation and differentiation \citep{Ozbay2008,Clapp2015,Yates2017}, cancer biology and treatment \citep{Piotrowska2018,Krzyzanski2018,CamaraDeSouza2018}, pathogen and immune response dynamics \citep{Fenton2006}, infectious disease transmission \citep{Anderson1980,Lloyd2001a,Lloyd2001b,Feng2000,Wearing2005,Lloyd2009,Feng2007,Ciaravino2018}, and population dynamics  \citep{MacDonald1978,Blythe1984, Metz1986, Boese1989,Nisbet1989, Cushing1994, Wolkowicz1997, Gyllenberg2007, Wang2016, Lin2018, Robertson2018}. See also \citet{Campbell2009} and the applications reviewed therein. 

However, while distributed delay equations are very flexible, in that they can incorporate arbitrary dwell time distributions, they also can be more challenging to derive, to analyze mathematically, and to simulate \citep{Cushing1994, Burton2005}. Thus, many modelers face a trade-off between building appropriate dwell time distributions into their mean field models (i.e., opting for an IE or IDE model) and constructing parsimonious models that are more easily analyzed both mathematically and computationally (i.e., opting for an ODE model).  For example, the following system of integral equations generalizes the SIR example above by incorporating an arbitrary distribution for the duration of infectiousness (i.e., the dwell time in state I):  

\begin{subequations}\label{eq:SIR2}\begin{align}
	S(t) =&\; S(0)(1-F_S(t)) \label{eq:S2}\\
	I(t) =&\; I(0)(1-F_I(t)) + \int_0^t \beta\,I(u)\,S(u)\,(1-F_I(u))\,du  \label{eq:I2} \\
	R(t) =&\; N - S(t) - I(t) \label{eq:R2} 
	\end{align} \end{subequations} where $N=S(0)+I(0)+R(0)$, $1-F_S(t)=\exp\big(-\int_0^t\beta\,I(u)\,du\big)$ is the survival function for the distribution of time spent in susceptible state S (i.e. the 1$^\text{st}$  event time under a Poisson process with rate $\lambda(t)=\beta\,I(t)$), and $1-F_I(t)=\exp\big(-\gamma\,t\big)$ is the survival function for the time spent in the infected state I \citep[related models can be found in, e.g.,][]{Feng2000, Ma2006, Krylova2013, Champredon2018}. A different choice of the CDF $F_I$ allows us to generalize the SIR model to other dwell time distributions that describe the time individuals spend in the infected state. Integral equations like those above can also be differentiated (assuming the integrands are differentiable) and represented as integrodifferential equations \citep[e.g., as in][]{Hethcote1980}.

There have been some efforts in the past to identify which categories of integral and integro-differential equations can be reduced to systems of ODEs \citep[e.g.,][and references therein]{MacDonald1989, Metz1991, Ponosov2002, Jacquez2002, Burton2005, Goltser2013, Diekmann2017}, but in practice the most well known case is the reduction of IEs and IDEs that assume Erlang\footnote{Erlang distributions are Gamma distributions with integer-valued shape parameters.} distributed dwell times. This is done using what has become known as the Linear Chain Trick \citep[LCT, also referred to as the Gamma Chain Trick;][]{MacDonald1978ch2, Smith2010} which dates at least back to \citet{Fargue1973} and earlier work by Theodore Vogel \citep[e.g.,][according to \citet{CamaraDeSouza2018}]{Vogel1961, Vogel1965}.  However, for more complex models that exceed the level of complexity that can be handled by existing ``rules of thumb" like the LCT, the current approach is to derive mean field ODEs from mean field integral equations that might themselves first need to be derived from system-specific stochastic state transition models \citep[e.g.,][and see the Appendix for an example.]{Kermack1927, Feng2007, Banks2013,Feng2016}. Unfortunately, modelers often avoid these extra (often laborious) steps in practice by assuming (sometimes only implicitly) very simplistic dwell time distributions based on Poisson process 1$^\text{st}$  event times as in the SIR example above.  

In light of the widespread use of ODE models, these challenges and trade-offs underscore a need for a more rigorous theoretical foundation to more effectively and more efficiently construct mean field ODE models that include more flexible dwell time distribution assumptions \citep{Wearing2005, Feng2016, Robertson2018}. The goal of this paper is to address these needs by 1) providing a theoretical foundation for constructing the desired system of ODEs directly from ``first principles" (i.e., stochastic model assumptions), without the need to derive ODEs from intermediate IDEs or explicit stochastic models, and by 2) providing similar analytical results for novel extensions of the LCT which allow more flexible dwell time distributions, and conditional relationships among dwell time distributions, to be incorporated into ODE models. We also aim to clarify how underlying (often implicit) stochastic model assumptions are reflected in the structure of corresponding mean field ODE model equations.  

The remainder of this paper is organized as follows. An intuitive description of the Linear Chain Trick (LCT) is given in \S \ref{sec:intuition} as a foundation for the extensions that follow. In \S \ref{sec:defs} we review key notation and properties of Poisson processes and certain probability distributions needed for the results that follow. In \S \ref{sec:pcode} we detail the association between Poisson process intensity functions and per capita rates in mean field ODEs, and in \S \ref{sec:weakmem} we introduce what we call the \textit{weak memorylessness} property of (nonhomogeneous) Poisson process $1^{st}$ event time distributions. In \S \ref{sec:simpleX} and \S \ref{sec:base} we give a formal statement of the LCT and in \S \ref{sec:extendedlct} a generalization that allows time-varying rates in the underlying Poisson processes. We then provide similar generalizations for more complex cases: In \S \ref{sec:tomulti} we provide results for multiple ways to implement transitions from one state to multiple states (which arise from different stochastic model assumptions and lead to different systems of mean field ODEs), and we address dwell times that obey Erlang mixture distributions. In \S\ref{sec:intermediate} we provide results that detail how the choice to ``reset the clock" (or not) following a sub-state transition is reflected in the corresponding mean field ODEs. Lastly, in \S\ref{sec:glct} we present a Generalized Linear Chain Trick (GLCT) which details how to construct mean field ODEs from first principles based on assuming a very flexible family of dwell time distributions that include the \textit{phase-type} distributions, i.e., hitting time distributions for certain families of continuous time Markov chains \citep[][]{Reinecke2012a,Horvath2016}. Tools for fitting phase-type distributions to data, or using them to approximate other distributions, are mentioned in the Discussion section \S\ref{sec:discussion} and the appendices, which also include additional information on deriving mean field integral equations from continuous time stochastic models.  

\subsection{Intuitive description of the Linear Chain Trick}\label{sec:intuition}

To begin, an intuitive understanding of the Linear Chain Trick (LCT) based on some basic properties of Poisson processes, is helpful for drawing connections between underlying stochastic model assumptions and the structure of their corresponding mean field ODEs. Here we consider a very basic case: the mean field ODE model for a stochastic process in which particles in state X remain there for an Erlang($r,k$) distributed amount of time before exiting to some other state (see Figure \ref{fig:simpleX} and \S\ref{sec:simpleX}). 

In short, the LCT exploits a natural stage structure within state X imposed by assuming an Erlang distributed dwell time with rate $r>0$ and shape parameter $k>0$ (i.e., a gamma($r,k$) distribution with integer shape $k$). Recall that an Erlang($r,k$) distribution models the time until the $k^\text{th}$ event under a homogeneous Poisson process with rate $r$. In that context, each event is preceded by a length of time that is exponentially distributed with rate $r$, and thus the time to the $k^\text{th}$ event is the sum of $k$ independent and identically distributed exponential random variables (i.e., the sum of $k$  $iid$ exponential random variables with rate $r$ is Erlang($r,k$) distributed). Particles in state X at a given time can therefore be classified by which event they are awaiting, i.e., each particle is in exactly one of $k$ sub-states of X$=$X$_1\cup\cdots\cup$X$_k$ where a particle is in state X$_i$ if it is waiting for the $i^\text{th}$ event to occur. The dwell time distribution for each sub-state X$_i$ is exponential with rate $r$, and particles leave the last state X$_k$ (and thus X) upon the occurrence of the $k^\text{th}$ event.

This sub-state partition is useful to impose on X because we may then exploit the fact that the mean field equations corresponding to these sub-state transitions are systems of linear (or nearly linear) ODEs.  Specifically, if we let $x_i(t)$ denote the expected number of particles at time $t$ in state X$_i$, then the mean field equations for this scenario are given by 

\begin{equation} \begin{split} \frac{d}{dt}x_1(t) =& -r\,x_0(t), \\ \frac{d}{dt}x_i(t) =& r\,x_{i-1}(t)\,-\,r\,x_i(t) \quad\text{ for } i=2,\ldots,k    \end{split} \end{equation} where the total amount in X at time $t$ is $x(t)=\sum_{i=1}^k x_i(t)$, $x_1(0)=x_0$ and $x_i(0)=0$ for $i=2,\ldots,k$.

As we show below, a Poisson process based perspective allows us to generalize the LCT in two main ways: First, we can extend the basic LCT to other more complex cases where we ultimately partition a focal state X in a similar fashion, including sub-state transitions with conditional dwell time distributions (see \S\ref{sec:tomulti}). Second, this reduction of states to sub-states with exponential dwell time distributions (i.e., dwell times distributed as $1^\text{st}$ event times under homogeneous Poisson processes) can also be extended to $1^\text{st}$ event time distributions under a nonhomogeneous Poisson processes with time varying rate $r(t)$, allowing for time-varying dwell time distributions to be used in extensionss of LCT.

\section{Model Framework}\label{sec:defs}

The context in which we consider applications of the Linear Chain Trick (LCT) is the derivation of continuous time mean field model equations for stochastic state transition models with a distributed dwell time in a focal state, X. Such mean field models might otherwise be modeled as integral equations (IEs) or integro-differential equations (IDEs), and we seek to identify generalizations of the LCT that allow us to replace such mean field integral equations with equivalent systems of 1$^{st}$ order ODEs. To do this, we first introduce some notation and review key properties of the Erlang family of gamma distributions, and their time-varying counterparts, $k^\text{th}$ event time distributions under nonhomogeneous Poisson processes.

\subsection{Distributions \& notation}\label{sec:defns}

Below we will extend the LCT from Erlang($r,k$) distributions (i.e., $k^\text{th}$ event time distributions under homogeneous Poisson processes with rate $r$) to event time distributions under nonhomogeneous Poisson processes with time varying rate $r(t)$, and related distributions like the minimum of multiple Erlang random variables. In this section we will first review properties of event time distributions under homogeneous Poisson processes, i.e., Erlang distributions, then analogous properties of event time distributions under nonhomogeneous Poisson processes.

Gamma distributions can be parameterized\footnote{They can also be parameterized in terms of their mean and variance (see Appendix \ref{sec:approxgamma}), or with a \textit{shape} and \textit{scale} parameters, where the \textit{scale} parameter is the inverse of the \textit{rate}.} by two strictly positive quantities: \textit{rate} $r$ and \textit{shape} $k$ (sometimes denoted $\alpha$ and $\beta$, respectively). The Erlang family of distributions can also be thought of as the a subfamily of gamma distributions with integer-valued shape parameters $k\in\mathbb{Z_+}$, or equivalently as the distributions resulting from the sum of $k$ $iid$ exponential distributions. That is, if a random variable $T=\sum_{i=1}^k T_i$, where all $T_i$ are independent exponential distributions with rate $r$, then $T$ is Erlang($r,k$) distributed.  Since the inter-event times under a homogeneous Poisson process are exponentially distributed, the time to the $k^\text{th}$ event is thus Erlang($r,k$). This construction is foundational to a proper intuitive understanding of the LCT and its extensions below. 

If random variable $T$ is gamma$(r,k)$ distributed, then its mean $\mu$, variance $\sigma^2$, and coefficient of variation $c_v$ are given by 

\begin{equation}\label{eq:meanvar}\mu=\frac{k}{r} \text{, } \sigma^2=\frac{k}{r^2} \text{, and } c_v=\frac{1}{\sqrt{k}}.\end{equation} Note that by solving \eqref{eq:meanvar}, one can parameterize a gamma distributed random variable by writing the rate $r$ and shape $k$ in terms of a target mean $\mu$ and variance $\sigma^2$ as  

\begin{equation} \label{eq:meanvarrk} r = \frac{\mu}{\sigma^2} \text{, and } k = \frac{\mu^2}{\sigma^2} = r\,\mu. \end{equation} However, to ensure this gamma distribution is also Erlang (i.e., to ensure the shape parameter $k$ is an integer) one must adjust the assumed variance up or down by rounding the value of $k$ in eq. \eqref{eq:meanvarrk} down or up, respectively, to the nearest integer (see Appendix \ref{sec:approxgamma} for details, and alternatives).

The Erlang density function ($g$), CDF ($G$), and survival\footnote{A useful interpretation of survival functions, which is used below, is that they give the expected proportion remaining after a give amount time.} function ($S=1-G$; also called the \textit{complementary CDF}) are given by 

\begin{subequations}\label{eq:gamma}\begin{align}
	g^k_r(t) =&\; r\,\frac{(r\,t)^{k-1}}{(k-1)!}e^{-rt} \label{eq:dgamma}\\
	G^k_r(t) =&\; 1 - \sum_{j=1}^{k} \frac{(r\,t)^{j-1}}{(j-1)!}e^{-r\,t} = 1- \sum_{j=1}^{k} \frac1r\,g_r^j(t) \label{eq:pgamma} \\
	S^k_r(t) =&\; 1-G^k_r(t) =  \sum_{j=1}^{k} \frac1r\,g_r^j(t). \label{eq:sgamma} 
	\end{align}\end{subequations}  

The results below use (and generalize) the following property of Erlang distributions, detailed in Lemma \ref{lem:smith711} \citep[eqs. 7.11 in][restated here without proof]{Smith2010}, which is the linchpin of the LCT.

\begin{lemma}\label{lem:smith711} The Erlang distribution density functions $g_r^j(t)$, with rate $r$ and shape $j$, satisfy 
	\begin{subequations}\label{eq:smith711}\begin{align}
		\frac{d}{dt}g^1_{r}(t) =& -r g^1_{r}(t),\text{ where }g^1_{r}(0)=r, \\
		\frac{d}{dt}g^j_{r}(t) =& r[g^{j-1}_{r}(t)-g^{j}_{r}(t)],\text{ where }g^j_{r}(0)=0 \text{ for }j\geq 2.
		\end{align}\end{subequations}
\end{lemma} 

Since homogeneous Poisson processes are a special case of nonhomogeneous Poisson processes\footnote{... despite the implied exclusivity of the adjective \textit{nonhomogeneous}.} from here on we will use ``Poisson process" or ``Poisson process with rate $r(t)$" to refer to cases that apply to both homogeneous (i.e., $r(t)=r$ constant) and nonhomogeneous Poisson processes. The event time distributions under these more general Poisson processes have the following properties. 

The $k^\text{th}$ event time distribution under a Poisson process with rate $r(t)$, starting from some time $\tau < t$ has a density function ($h_{r}^k$), survival function ($\mathcal{S}_{r}^k$), and CDF ($H_{r}^k\equiv 1-\mathcal{S}_{r}^k$) given by

\begin{subequations}\label{eq:ktheventdist}\begin{align}
	h^k_{r}(t,\tau) =&\; r(t)\,\frac{m(t,\tau)^{k-1}}{(k-1)!}\,e^{-m(t,\tau)} \quad \text{ and} \label{eq:dktheventdist}\\
	%H^k_{r}(t,\tau) =&\; 1 - \sum_{j=1}^{k} \frac{h^j_{r}(t,\tau)}{r(t)} \label{eq:pktheventdist} \\
	\mathcal{S}^k_{r}(t,\tau) =&\;\sum_{j=1}^{k} \frac{h^j_{r}(t,\tau)}{r(t)} \label{eq:sktheventdist} 
	\end{align}\end{subequations} where 

\begin{equation}\label{eq:m} m(t,\tau) \equiv \int_\tau^t r(s)\,ds\end{equation} and $\frac{d}{dt}m(t,\tau)=r(t)$.

For an arbitrary survival function starting at time $\tau$ (i.e., over the period $[\tau,t]$ where $t\geq\tau$) we will use the notation $S(t,\tau)$. In some instances, we also use the notation $S(t)\equiv S(t,0)$. 

Lastly, in the context of state transitions models, it is common to assume that, upon leaving a given state (e.g., state X) at time $t$, individuals are distributed across multiple recipient states according to a generalized Bernoulli distribution (also known as the \textit{categorical distribution} or the multinomial distribution with $1$ trials) defined on the integers 1 through $k$ where the probability of a particle entering the $j^\text{th}$ of $k$ recipient states ($j\in{1,\ldots,k}$) is $p_j(t)$ and $\sum_{j=1}^k p_j(t)=1$.

\section{Results}\label{sec:results}

The results below focus on one or more states, within a potentially larger state transition model, for which we would like to assume a particular dwell time distribution and derive a corresponding system of mean field ODEs using the LCT or a generalization of the LCT. In particular, the results below describe how to construct those mean field ODEs directly from stochastic model assumptions without needing to derive them from equivalent mean field integral equations (which themselves may need to be derived from an explicit continuous-time stochastic model). 

\subsection{Preliminaries}
Before presenting extensions of the LCT, we first illustrate in \S \ref{sec:pcode} how mean field ODEs (for a given stochastic continuous-time state transition model) include terms that reflect underlying Poisson process rates using a simple generalization of the exponential decay equation $\frac{d}{dt}x(t)=-r\,x(t)$ where each particle is assumed to exit state X after an exponentially  distributed amount of time (i.e., after the 1$^\text{st}$  even under a Poisson process with constant rate $r$). We extend this model by (1) incorporating an influx rate ($\mathcal{I}(t)$) into state X, and (2) allowing a time varying rate $r(t)$ for the underlying Poisson process. In \S \ref{sec:weakmem}, we highlight a key property of these Poisson process 1$^\text{st}$  event time distributions that we refer to as a \textit{weak memorylessness property} since it is a generalization of the well known memorylessness property of the exponential and geometric distributions. 

\subsubsection{Per capita transition rates in ODEs reflect underlying Poisson process rates} \label{sec:pcode}
To build upon the intuition spelled out above in \S \ref{sec:intuition}, consider the basic exponential decay equation as a mean field model for a stochastic model where particles are assumed to leave state X following an exponentially distributed dwell time. Now assume instead that particles exit X following the 1$^{st}$ event time under Nonhomogeneous Poisson processes with rate $r(t)$ (recall the 1$^\text{st}$  event time distribution is exponential if $r(t)=r$ is constant), and that there is an additional influx rate $\mathcal{I}(t)$ into state X. As illustrated by the corresponding mean field equations given below, the rate function $r(t)$ can be viewed as either the intensity function\footnote{That is, the probability of a given individual exiting state X during a brief time period [$t,t+\Delta t$] is approximately $r(t)\Delta t$.} for the Poisson process governing when individuals leave state X, or as the (mean field) per-capita rate of loss from state X as shown in eq. \eqref{eq:nhpp}.

\begin{example}[Equivalence between Poisson process rates \& per capita rates in mean field ODEs]\label{ex:nhpp}
	Consider the scenario described above. The survival function for the dwell time distribution for a particle entering X at time $\tau$ is $S(t,\tau)=\exp(-\int_{\tau}^{t} r(u)\,du)$, and it follows from the Law of Large Numbers that the expected proportion of such particles remaining in X at time $t>\tau$ is given by $S(t,\tau)$.  Let $x(t)$ be the total amount in state X at time $t$, $x(0)=x_0$, and that $\mathcal{I}(t)$ and $r(t)$ are integrable, non-negative functions of $t$.  Then the corresponding mean field integral equation for this scenario is
	
	\begin{equation} 
	x(t) =\; x_0\,S(t,0) + \int^t_{0} \mathcal{I}(\tau)\, S(t,\tau) d\tau \label{eq:Xnhpp}
	\end{equation} and equation \eqref{eq:Xnhpp} above is equivalent to \begin{equation} \frac{d}{dt}{x}(t) =\; \mathcal{I}(t) - r(t)\, x(t), \; \text{ with } x(0)=x_0. \label{eq:nhpp} \end{equation}
	
\end{example} 

\begin{proof} Using the Leibniz rule for integrals to differentiate \eqref{eq:Xnhpp}, and using Lemma \ref{lem:smith711}, yields 
	
	\begin{equation}
	\begin{split}
	\frac{d}{dt}{x}(t) =&\; x_0\frac{d}{dt}S(t,\tau) + \frac{d}{dt}\int^t_{0} \mathcal{I}(\tau)\,S(t,\tau)d\tau\\
	%=&\; -r(t)\,x_0\,e^{-\int_{0}^{t}r(u)\,du}  +\mathcal{I}(t) + \int^t_{0}\mathcal{I}(\tau) \frac{d}{dt} e^{-\int_{\tau}^{t}r(u)\,du} d\tau \\
	=&\; -r(t)\,x_0\,e^{-\int_{0}^{t}r(u)\,du} +\mathcal{I}(t)  -r(t) \int^t_{0}\mathcal{I}(\tau) e^{-\int_{\tau}^{t}r(u)\,du} d\tau \\
	=&\; \mathcal{I}(t) - r(t)\bigg[x_0\,e^{-\int_{0}^{t}r(u)\,du}  + \int^t_{0} \mathcal{I}(\tau)e^{-\int_{\tau}^{t}r(u)\,du} d\tau \bigg] \\
	=&\; \mathcal{I}(t)   - r(t) x(t).
	\end{split}
	\end{equation}  
	%\qed 
\end{proof}

The intuition behind the LCT relies in part on the memorylessness property of the exponential distribution. For example, when particles accumulate in a state with an exponentially distributed dwell time distribution, then at any given time all particles currently in that state have \textit{iid} exponentially distributed amounts of time left before they leave that state regardless of the duration of time already spent in that state, thus the memorylessness property of the exponential distribution imparts a Markov property (i.e., the remaining time duration depends only on the current state, not the history of time spent in that state) which permits a mean field ODE. As detailed in the next section, there is an analogous Markov property imparted by the more general weak memorylessness property of (nonhomogeneous) Poisson process 1$^\text{st}$ event time distributions, which we use to extend the LCT.

\subsubsection{Weak memoryless property of Poisson process 1$^\text{st}$ event time distributions} \label{sec:weakmem}
The familiar memorylessness property of exponential and geometric distributions can, in a sense, be generalized to (nonhomogeneous) Poisson process $1^\text{st}$ event time distributions. Recall that if an exponentially distributed (rate $r$) random variable $T$ represents the time until some event, then if the event has not occurred by time $s$ the remaining duration of time until the event occurs is also exponential with rate $r$. The analogous \textit{weak memorylessness} property of nonhomogeneous Poisson process $1^\text{st}$  event time distributions is detailed in the following definition.

\begin{definition}[Weak memorylessness property of Poisson process 1$^\text{st}$ event times]\label{Th:wm} Assume $T$ is a (possibly nonhomogeneous) Poisson process $1^\text{st}$  event time starting at time $\tau$, which has CDF $H_r^1(t,\tau)=1-\exp(-m(t,\tau))$ (see eqs. \eqref{eq:ktheventdist} and \eqref{eq:m}). If the event has not occurred by time $s>\tau$ the distribution of the remaining time $T_s \equiv T-s\;|\;T>s$ follows a shifted but otherwise identical Poisson process 1$^\text{st}$ event time distribution with CDF $P(T_s \leq t)=H_r^1(t+s,s)$. 
\end{definition} 

\clearpage
\begin{proof} The CDF of $T_s$ (for $t>\tau$) is given by
	
	\begin{equation} \begin{split}
	P(T_s \leq t) =&\; P(T - s \leq t \;|\; T>s) =\; \frac{P(s < T \leq s+ t)}{P(s<T)} \\
	=&\; \frac{H_r^1(t+s,\tau) - H_r^1(s,\tau)}{1-H_r^1(s,\tau)} =\; 1-\frac{1-H_r^1(t+s,\tau)}{1-H_r^1(s,\tau)}\\
	=&\; 1-\frac{e^{-m(t+s,\tau)}}{e^{-m(s,\tau)}} %\\
	=\; 1-e^{-m(t+s,s)} =\; H_r^1(t+s,s). \\
	\end{split} \end{equation} If $r(t)=r$ is a positive constant we recover the memorylessness property of the exponential distribution.
	%%\qed 
\end{proof} 

That is, Poisson process 1$^\text{st}$ event time distributions are memoryless up to a time shift in their rate functions. Viewed another way, in the context of multiple particles entering a given state X at different times and leaving according to independent Poisson process 1$^\text{st}$ event times with identical rates $r(t)$ (i.e., $t$ is absolute time, not time since entry into X), then for all particles in state X at a given time the distribution of time remaining in state X is (1) independent of how much time each particle has already spent in X and (2) follows \textit{iid} Poisson process 1$^\text{st}$ event time distributions with rate $r(t)$.

\subsection{Simple case of the LCT}\label{sec:simpleX}

\begin{figure}[htb]
	\hspace{1.6cm} (a) \hspace*{4.2cm} (b)\\
	\centerline{\includegraphics[width=0.28\textwidth]{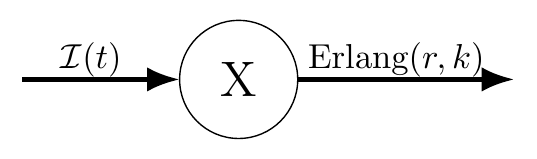} \vline \includegraphics[width=0.484\textwidth]{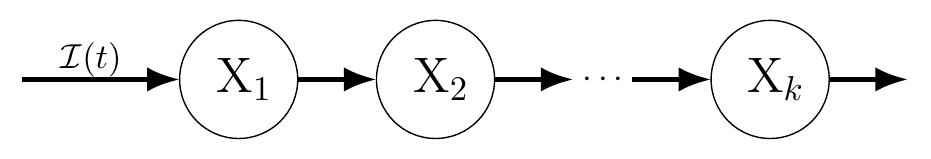}}
	\caption{Example diagram where state X has an Erlang($r,k$) distributed dwell time, represented either as (a) a single state and corresponding integral equation, or (b) as a set of $k$ sub-states each with exponential dwell time distributions whose mean field equations can be represented as either integral equations or a system of ODEs (see Theorem \ref{Th:simple}). Rate $\mathcal{I}(t)$ is an integrable non-negative function describing the mean field influx rate into state X.}
	\label{fig:simpleX}
\end{figure}

To illustrate how the LCT follows from Lemma \ref{lem:smith711}, consider the following simple case of the LCT as illustrated in Figure \ref{fig:simpleX}, where a higher dimensional model includes a state transition into, then out of, a focal state X. Assume the time spent in that state ($T_X$) follows an Erlang($r,k$) distribution (i.e., $T_X \sim $ Erlang($r,k$)). Then the LCT provides a system of ODEs equivalent to the mean field integral equations for this process as discussed in \S \ref{sec:intuition} and as detailed in the following theorem: \\

\begin{theorem}[Simple LCT]\label{Th:simple}
	Consider a continuous time state transition model with inflow rate $\mathcal{I}(t)$ (an integrable non-negative function of $t$) into state X which has an Erlang($r,k$) distributed dwell time (with survival function $S_r^k$ from eq. \eqref{eq:sgamma}). Let $x(t)$ be the amount in state X at time $t$ and assume $x(0)=x_0$.
	
	The mean field integral equation for this scenario is (see Fig. \ref{fig:simpleX}a)
	
	\begin{equation} 
	x(t) =\; x_0S_r^k(t) + \int^t_{0} \mathcal{I}(s)\,S_r^k(t-s) ds. \label{eq:X1}\\
	\end{equation} State X can be partitioned into $k$ sub-states X$_i$, $i=1,\ldots,k$, where particles in X$_i$ are those awaiting the $i^\text{th}$ event as the next event under a homogeneous Poisson process with rate $r$. Let $x_i(t)$ be the amount in X$_i$ at time $t$. Equation \eqref{eq:X1} above is equivalent to the mean field ODEs (see Fig. \ref{fig:simpleX}b) 
	
	\begin{subequations}\label{eq:chain1}\begin{align} 
		\frac{d}{dt}{x_1}(t) =&\; \mathcal{I}(t) -r\, x_1(t) \label{eq:chain1a} \\
		\frac{d}{dt}{x_j}(t) =&\; r\, x_{j-1}(t) - r\, x_j(t), \quad j=2,\ldots,k  \label{eq:chain1b} 
		\end{align}\end{subequations} with initial conditions $x_1(0)=x_0$, $x_j(0)=0$ for $j\geq2$. Here $x(t) = \sum_{j=1}^{k} x_j(t)$ and 
	
	\begin{equation}\label{eq:Xj1}
	x_j(t)= x_0\,\frac1r\,g_r^j(t) + \int^t_{0} \mathcal{I}(s) \frac1r\,g_r^j(t-s)ds. % substituted in terms of g(t-s)
	\end{equation}
\end{theorem} % END THEOREM 1

\begin{proof}
	
	Substituting eq. \eqref{eq:sgamma} into eq. \eqref{eq:X1} and then substituting eq. \eqref{eq:Xj1} yields
	
	\begin{equation}\label{eq:X1reduce}\begin{split}
	x(t) =&\;  x_0\,S_r^k(t) + \int^t_{0} \mathcal{I}(s)\,S_r^k(t-s) \,ds \\
	=&\;  x_0\,\sum_{j=1}^{k} \frac1r\,g_r^j(t) + \int^t_{0}\mathcal{I}(s)\,\sum_{j=1}^{k} \frac1r\,g_r^j(t-s) \,ds \\
	=&  \sum_{j=1}^{k} \left(x_0\, \frac1r\,g_r^j(t) + \int^t_{0} \mathcal{I}(s)\;\frac1r\,g_r^j(t-s) \,ds \right) 
	= \sum_{j=1}^{k} x_j(t).
	\end{split}
	\end{equation} Differentiating equations \eqref{eq:Xj1} (for $j=1,\ldots,k$) yields equations \eqref{eq:chain1} as follows.
	
	For $j=1$, equation \eqref{eq:Xj1} reduces to 
	
	\begin{equation}\label{eq:x1integral1}
	x_1(t) = x_0 e^{-r\,t} + \int^t_{0} \mathcal{I}(s) e^{-r(t-s)}ds.	
	\end{equation} 
	
	Differentiating $x_1(t)$ using the Leibniz integral rule, and then substituting \eqref{eq:x1integral1} yields 
	
	\begin{equation}
	\frac{d}{dt}{x_1}(t) =\; - r x_0 e^{-r\,t} - r\int^t_{0} \mathcal{I}(s) e^{-r(t-s)}ds + \mathcal{I}(t)  \;=\; \mathcal{I}(t) - r x_1(t).
	\end{equation}
	
	Similarly, for $j\geq2$, Lemma \ref{lem:smith711} yields   
	
	\begin{equation} \begin{split}
	\frac{d}{dt}{x_j}(t)=&\; x_0\,\frac1r\,\frac{d}{dt}g_r^j(t) + \int^t_{0} \mathcal{I}(s) \frac{d}{dt} \left(\frac1r\,g_r^j(t-s)\right)\,ds \\
	=&\; x_0\,\left(g_r^{j-1}(t)-g_r^j(t)\right) + \int^t_{0} \mathcal{I}(s) \left(g_r^{j-1}(t-s)-g_r^j(t-s)\right)\,ds \\
	=&\; r\,\bigg( \frac{x_0}r\,g_r^{j-1}(t) + \int^t_{0} \mathcal{I}(s) \frac1r g_r^{j-1}(t-s)\,ds \bigg) -\; r\,\bigg ( \frac{x_0}r\,g_r^{j}(t)\\
	& \;  + \int^t_{0} \mathcal{I}(s) \frac1r g_r^j(t-s)\,ds\bigg)
	=\; r\,x_{j-1}(t) - r\,x_j(t).
	\end{split}
	\end{equation}
	%\qed 
\end{proof}

Note the dwell time distributions for sub-states X$_j$ with $j\geq 1$ are exponential with rate $r$ (i.e., $T_{X_j}\sim$ Erlang($r,1$)). To see why, consider each particle in state $X$ to be following independent homogeneous Poisson processes (rate $r$), waiting for the $k^\text{th}$ event to occur. Then let $\chi_i$(t) (where $1\leq i \leq k$) be the expected number of particles in state X (at time $t$) that have not reached the $i^\text{th}$ event. Then 

\begin{equation}\label{eq:chii}
\chi_i(t) = x_0\,S_r^i(t) + \int_0^t \mathcal{I}(s)\,x_0\,S_r^i(t-s)\,ds
\end{equation} and by eq. \eqref{eq:gamma} we see from eqs. \eqref{eq:x1integral1} and eq. \eqref{eq:chii} that $x_j(t) = \chi_j(t) - \chi_{j-1}(t)$. That is, particles in state X$_j$ are those for which the $(j-1)^\text{th}$ event has occurred, but not the $j^\text{th}$ event. Thus, by properties of Poisson processes the dwell time in state X$_j$ is exponential with rate $r$.\\

Next, we consider a more general statement of Theorem \ref{Th:simple} that better formalizes the standard LCT as used in practice.

\subsection{Standard LCT}\label{sec:base}

The following Theorem and Corollary together provide a formal statement of the standard Linear Chain Trick (LCT). Here we have extended the basic case in the previous section (see Theorem \ref{Th:simple} and compare Figures \ref{fig:simpleX} and \ref{fig:generalXY}) to explicitly include that particles leaving X enter state Y and remain in Y according to an arbitrary distribution with survival function $S$, where $S(t,\tau)$ is the expected proportion remaining at time $t$ that entered at time $\tau<t$. We also assume non-negative, integrable input rates $\mathcal{I}_X(t)$ and $\mathcal{I}_Y(t)$ to X and Y, respectively, to account for movement into these two focal states from other states in the system.

\begin{theorem}[Standard LCT]\label{Th:base}
	Consider a continuous time dynamical system model of mass transitioning among various states, with inflow rate $\mathcal{I}_X(t)$ to a state X and an Erlang($r,k$) distributed delay before entering state Y. Let $x(t)$ and $y(t)$ be the amount in each state, respectively, at time $t$. Further assume an inflow rate $\mathcal{I}_Y(t)$ into state Y from other non-X states, and that the underlying stochastic model assumes that the duration of time spent in state Y is determined by survival function $S(t,\tau)$. Assume $\mathcal{I}_i(t)$  are integrable non-negative functions of $t$, and assume non-negative initial conditions $x(0)=x_0$ and $y(0)=y_0$. 
	
	The mean field integral equations for this scenario are 
	
	\begin{subequations}\label{eq:XY2} \begin{align} 
		x(t) =&\; x_0\,S_r^k(t) + \int^t_{0} \mathcal{I}_X(s)\,S_r^k(t-s) ds \label{eq:X2}\\
		\begin{split} \label{eq:Y2} y(t) =&\; y_0S(t,0) + \int^t_{0} \bigg(\mathcal{I}_Y(\tau) + x_0\,g_r^k(\tau) \\ & \qquad \qquad \qquad +\int^\tau_0 \mathcal{I}_X(s)\,g^k_r(\tau-s)ds\bigg)S(t,\tau)d\tau. \end{split} 
		\end{align} \end{subequations}
	
	Equations \eqref{eq:XY2} are equivalent to 
	
	\begin{subequations}\label{eq:chain2}\begin{align} 
		\frac{d}{dt}{x_1}(t) =&\; \mathcal{I}_X(t) -r x_1(t) \label{eq:chain2a} \\
		\frac{d}{dt}{x_j}(t) =&\; r x_{j-1}(t) - r x_j(t), \quad j=2,\ldots,k  \label{eq:chain2b} \\
		%\frac{d}{dt}{y}(t) =&\; \mathcal{I}_Y(t) + r x_{k} - \mu(t) y(t)
		y(t) =&\; y_0S(t,0) + \int^t_{0} \underbrace{\left(\mathcal{I}_Y(\tau) + r\,x_k(\tau)\right)}_\text{Net input rate at time $\tau$}S(t,\tau)d\tau \label{eq:chain2c} 
		\end{align}\end{subequations}
	
	where $x(t) = \sum_{j=1}^{k} x_j(t)$ with initial conditions $x_1(0)=x_0$, $x_j(0)=0$ for $j\geq2$ and 
	
	\begin{equation}\label{eq:Xj2}
	x_j(t)= x_0\,\frac1r\,g_r^j(t) + \int^t_{0} \mathcal{I}_X(s) \frac1r\,g_r^j(t-s)ds. % substituted in terms of g(t-s)
	\end{equation}
	
\end{theorem}

\begin{figure}[htb] 
	\hspace*{4em} (a) \\
	\hspace*{6em} \includegraphics[height=7em]{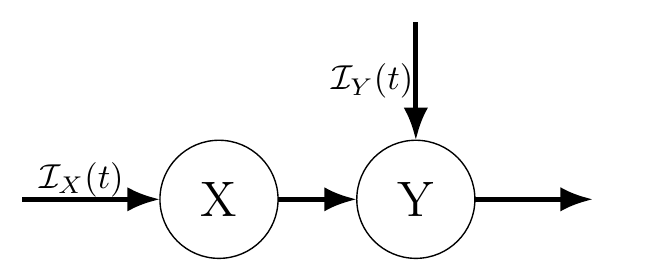}\\
	\hspace*{4em} (b) \\
	\hspace*{6em} \includegraphics[height=7em]{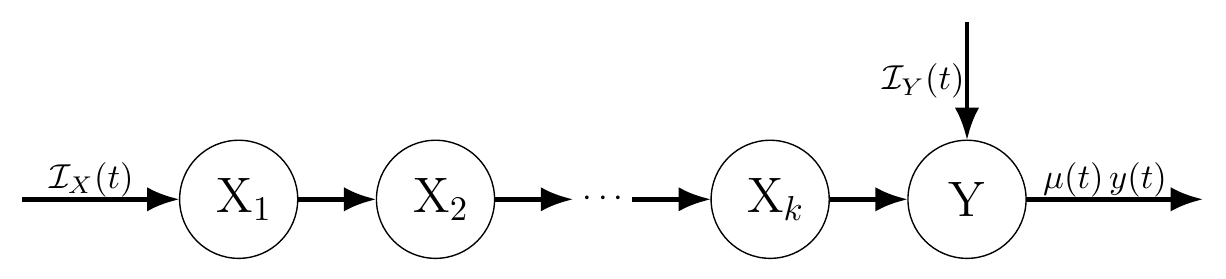}\\
	\hspace*{4em} (c) \\
	\hspace*{6em} \includegraphics[height=7em]{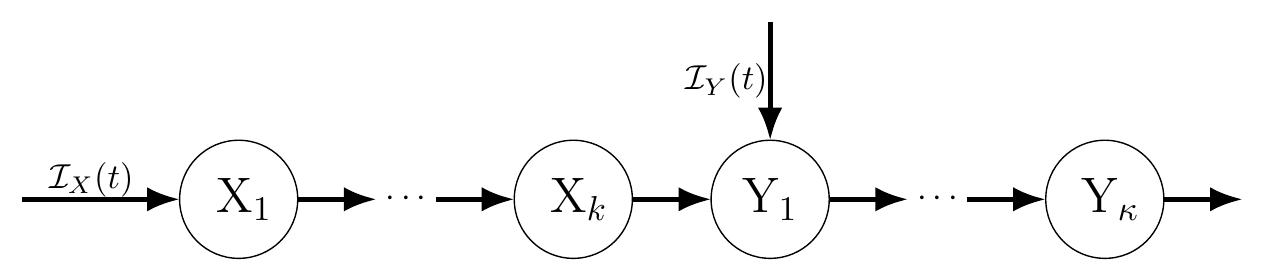}\\
	\caption{\textbf{(Standard LCT results)} This generic case assumes that the dwell times in state X (see panel \textbf{a}) are Erlang($r,k$) distributed with inflow rates $\mathcal{I}_X(t)\geq0$ into state X and $\mathcal{I}_Y(t)\geq0$ into state Y. Panels \textbf{b} and \textbf{c} show sub-states resulting from applying the LCT and Corollary \ref{Th:baseCorr} assuming either (\textbf{b}) dwell times in state Y are determined by per-capita rate function $\mu(t)$, or (\textbf{c}) dwell times in Y follow and Erlang distribution with shape parameter $\kappa$. }
	\label{fig:generalXY} 
\end{figure}

\begin{proof}
	
	Equations \eqref{eq:chain2a}, \eqref{eq:chain2b} and \eqref{eq:Xj2} follow from Theorem \ref{Th:simple}. Equation \eqref{eq:chain2c} follows from substituting \eqref{eq:Xj2} into \eqref{eq:Y2}. The definition of $x_j$ and initial condition $x(0)=x_0$ together imply $x_1(0)=x_0$ and $x_j(0)=0$ for the remaining $j\geq2$.
	%\qed 
\end{proof}

\clearpage
\begin{corollary}\label{Th:baseCorr} Integral equations like eq. \eqref{eq:chain2c} can be represented by equivalent systems of ODEs depending on the assumed Y dwell time distribution (i.e., $S(t,\tau)$), for example: 
	\begin{enumerate}[resume]
		\item If particles leave state Y following the 1$^\text{st}$  event time distribution under a nonhomogeneous Poisson process with rate $\mu(t)$ (i.e., if the per-capita rate of loss from Y is $\mu(t)$), then by Theorem \ref{Th:base}, with $\mathcal{I}(t)=\mathcal{I}_Y(t)+rx_k(t)$, it follows that $S(t,\tau)=\exp(-\int_{\tau}^{t}\mu(u)\,du)$ and 
		
		\begin{equation}\label{eq:Y2nonhomo}
		\frac{d}{dt}{y}(t) =\; \mathcal{I}_Y(t) + r x_{k}(t) - \mu(t) y(t).\end{equation}
		
		\item If particles leave Y after an Erlang($\mu,\kappa$) delay, then $S(t,\tau)=S_\mu^\kappa(t-\tau)$ and according to Theorem \ref{Th:base}, with $\mathcal{I}(t)=\mathcal{I}_Y(t)+rx_k(t)$, it follows that $y=\sum_{i=1}^\kappa y_i$ and 
		
		\begin{subequations}\label{eq:Y2chain}\begin{align}
			\frac{d}{dt}{y_1}(t) =&\; \mathcal{I}_Y(t) + r x_{k}(t) -\mu\, y_1(t)  \\
			\frac{d}{dt}{y_i}(t) =&\; \mu\, y_{i-1}(t) - \mu\, y_i(t), \quad i=2,\ldots,\kappa.
			\end{align}\end{subequations}
		
		\item As implied by parts 1 and 2 above, if the per-capita loss rate $\mu(t)=\mu$ is constant or time spent in Y is otherwise exponentially distributed, $S(t,\tau)=e^{-\mu\,(t-\tau)}$, then 
		
		\begin{equation}\label{eq:Y2exp}
		\frac{d}{dt}{y}(t) =\; \mathcal{I}_Y(t) + r x_{k}(t) - \mu\, y(t).\end{equation}
		
		\item Any of the more general cases considered in the sections below.
		
	\end{enumerate}
\end{corollary}

\begin{example} \label{ex:simple} To illustrate how the Standard LCT (Theorem \ref{Th:base} and Corollary \ref{Th:baseCorr}) is used to construct a system of mean field ODEs (with or) without the intermediate steps involving mean field integral equations, consider a large number of particles that begin (at time $t=0$) in state W and then each transitions to state X after an exponentially distributed amount of time (with rate $a$). Particles remain in state X according to a Erlang($r,k$) distributed delay before entering state Y. They then go to state Z after an exponentially distributed time delay with rate $\mu$. The mean field model of such a system can be stated as follows (see Appendix \ref{a:example1} for a derivation of eqs. \eqref{eq:example1}), 
	
	\begin{subequations}\label{eq:example1}\begin{align}
		\frac{d}{dt}w(t)\;=&\; -a\,w(t)  \label{eq:example1w} \\
		x(t)=&\int_0^t \overbrace{a\,w(s)}^{\mathcal{I}(s)}\,\overbrace{S_r^k(t-s)}^{\text{proportion remaining}}\,ds \label{eq:example1x}\\
		y(t)=&\int_0^t\underbrace{\left(\int_0^\tau a\,w(s)\,g_r^k(\tau-s)\,ds\right)}_\text{Net input rate at time $\tau$} S_\mu^1(t-\tau)\,d\tau  \label{eq:example1y}\\
		\frac{d}{dt}z(t)=&\;\mu\,y(t) \label{eq:example1z}
		\end{align}
	\end{subequations} where the state variables $w$, $x$, $y$, and $z$ correspond to the amount in each of the corresponding states, and we assume the initial conditions $w(0)=w_0>0$ and $x(0)=y(0)=z(0)=0$.
	
	Applying Theorem \ref{Th:base} to eqs. \eqref{eq:example1}, or using the results of Theorem \ref{Th:base} directly, given the assumptions spelled out above, yields the equivalent system of mean field ODEs.
	
	\begin{subequations}\label{eq:example1ode}\begin{align}
		\frac{d}{dt}w(t)\;=&\; -a\,w(t) \label{eq:example1odew} \\
		\frac{d}{dt}{x_1}(t) =&\; a\,w(t) - r\,x_1(t)  \label{eq:example1odex1}  \\
		\frac{d}{dt}{x_j}(t) =&\; r\,x_{j-1}(t) - r\,x_j(t), \qquad \text{ for } j=2,\ldots,k \label{eq:example1odexj}\\
		\frac{d}{dt}{y}(t)   =&\; r\,x_{k}(t) - \mu\,y(t)  \label{eq:example1odey} \\
		\frac{d}{dt}z(t)=&\;\mu\,y(t)  \label{eq:example1odez}
		\end{align}
	\end{subequations}
	where $x(t)=\sum_{j=1}^kx_j(t)$. 
	
\end{example}

\begin{example}\label{ex:sirlct} To illustrate how the Standard LCT can be applied to a system of mean field ODEs to substitute an implicit exponential dwell time distribution with an Erlang distribution, consider the SIR example discussed in the Introduction \citep[eqs. \eqref{eq:SIR} and \eqref{eq:SIR2}, see also][]{Anderson1980, Lloyd2001a,Lloyd2001b}. Assume the dwell time distribution for the infected state I is Erlang (still with mean $1/\gamma$) with variance\footnote{Here the variance is assumed to have been chosen so that the resulting shape parameter is integer valued. See Appendix \ref{sec:approxgamma} for related details.} $\sigma^2$, i.e., by eqs. \eqref{eq:meanvarrk}, Erlang with a rate $\gamma k$ and shape $k=\sigma^2/\gamma^2$. 
	
	By Theorem \ref{Th:base} and Corollary \ref{Th:baseCorr},  with $\mathcal{I}_I(t)=\lambda(t)\,S(t)$, the corresponding mean field ODEs are 
	
	\begin{subequations}\label{eq:SIRlct}\begin{align}
		\frac{d}{dt}S(t) =&\; -\lambda(t)\,S(t) \label{eq:Slct}\\
		\frac{d}{dt}I_1(t) =&\; \lambda(t)\,S(t) - \gamma k\,I_1(t) \label{eq:I1lct} \\
		\frac{d}{dt}I_j(t) =&\; \gamma k\,I_{j-1}(t) - \gamma k\,I_j(t), \quad \text{ for } j=2,\ldots,k \label{eq:Iilct} \\
		\frac{d}{dt}R(t) =&\; \gamma k\,I_k(t) \label{eq:Rlct} 
		\end{align} \end{subequations} where $S(t)$, $I(t)=\sum_{j=1}^k I_j(t)$, and $R(t)$ correspond to the number of susceptible, infected, and recovered individuals at time $t$. Notice that if $\sigma^2=\gamma^2$ (i.e. if shape $k=1$), the dwell time in infected state I is exponentially distributed with rate $\gamma$, $I(t)=I_1(t)$, and eqs. \eqref{eq:SIRlct} reduce to eqs. \eqref{eq:SIR}. 
	
	This example nicely illustrates how using Theorem \ref{Th:base} to relax an exponential dwell time assumption implicit in a system of mean field ODEs is much more straightforward than constructing them after first deriving the integral equations, like eqs. \eqref{eq:SIR2}, and then differentiating them using Lemma \ref{lem:smith711}. In the sections below, we present similar theorems intended to be used for constructing mean field ODEs directly from stochastic model assumptions.
	
\end{example}

\subsection{Extended LCT for Poisson process $k^\text{th}$ event time distributed dwell times}\label{sec:extendedlct}

Assuming an Erlang($r,k$) distributed dwell time in a given state as in the Standard LCT tacitly assumes that each particle remains in state X until the $k^\text{th}$ event under a homogeneous Poisson process with rate $r$.  Here we generalize the Standard LCT by assuming the dwell time in X follows the more general $k^\text{th}$ event time distribution under a Poisson process with rate $r(t)$.

First, observe the following Lemma, which is based on recognizing that eqs. \eqref{eq:smith711} in Lemma \ref{lem:smith711} are more practical when written in terms of $\frac1r g^j_r(t)$ (see the proof of Theorem \ref{Th:simple}), i.e, for $j=2,\ldots,k$ 

\begin{subequations} \label{eq:smith711divbyr} \begin{align}
	\frac{d}{dt}\bigg[\frac1r g^1_{r}(t)\bigg] =&\; -r \bigg[\frac1r g^1_{r}(t)\bigg], \\
	\frac{d}{dt}\bigg[\frac1r g^j_{r}(t)\bigg] =&\; r\bigg[\frac1r g^{j-1}_{r}(t)-\frac1r g^{j}_{r}(t)\bigg], 
	\end{align}\end{subequations} where $\frac1r g^1_{r}(0)=1$ and $\frac1r g^j_{r}(0)=0$.

\begin{lemma} \label{L:dh} A similar relationship to eqs. \eqref{eq:smith711divbyr} above (i.e., to Lemma \ref{lem:smith711}) holds true for the Poisson process $j^\text{th}$ event time distribution density functions $h_{r}^j$ given by eq. \eqref{eq:dktheventdist}. Specifically, 
	
	\begin{subequations} \label{eq:smith711NHPP} \begin{align}
		\frac{d}{dt}\bigg[\frac1{r(t)} h^1_{r}(t,\tau)\bigg] =&\; -r(t) \bigg[\frac1{r(t)} h^1_{r}(t,\tau)\bigg],    \label{eq:smith711NHPPa}\\
		\frac{d}{dt}\bigg[\frac1{r(t)} h^j_{r}(t,\tau)\bigg] =&\; r(t)\bigg[\frac1{r(t)} h^{j-1}_{r}(t,\tau)-\frac1{r(t)} h^{j}_{r}(t,\tau)\bigg],   \label{eq:smith711NHPPb}
		\end{align}\end{subequations} where $\frac1{r(\tau)}h^1_{r}(\tau,\tau)=1$ and $\frac1{r(\tau)} h^j_{r}(\tau,\tau)=0$ for $j\geq 2$. Note that, if for some $t$ $r(t)=0$, this relationship can be written in terms of \begin{equation}\label{eq:u}u_r^k(t,\tau)\ \equiv  \frac{m(t,\tau)^{k-1}}{(k-1)!}\,e^{-m(t,\tau)},\end{equation} as shown in the proof below, where $h_r^k(t,\tau)=r(t)\,u_r^k(t,\tau)$, $u^1_{r}(\tau,\tau)=1$, and $u^j_{r}(\tau,\tau)=0$ for $j\geq2$.
\end{lemma}

\begin{proof} For $j=1$, \begin{equation}\begin{split}
	%	\frac{d}{dt}\bigg[\frac1{r(t)} h^1_{r}(t,\tau)\bigg]=&\; 
	\frac{d}{dt}\bigg[u^1_{r}(t,\tau)\bigg] =&\; \frac{d}{dt}e^{-m(t,\tau)} 
	=\; -r(t)\,e^{-m(t,\tau)} \\
	=&\; -r(t)\,u^1_{r}(t,\tau). % = -r(t) \bigg[\frac1{r(t)} h^1_{r}(t,\tau)\bigg].
	\end{split}\end{equation}
	Likewise, for $j\geq2$, we have 
	
	\begin{equation}\begin{split}
	%\frac{d}{dt}\bigg[\frac1{r(t)} h^j_{r}(t,\tau)\bigg]=&\; 	
	\frac{d}{dt}\bigg[u^j_{r}(t,\tau)\bigg] =&\; \frac{d}{dt} \frac{m(t,\tau)^{k-1}}{(k-1)!}\,e^{-m(t,\tau)} \\
	=&\; r(t)\, \frac{m(t,\tau)^{k-2}}{(k-2)!}\,e^{-m(t,\tau)} - r(t)\, \frac{m(t,\tau)^{k-1}}{(k-1)!}\,e^{-m(t,\tau)}  \\
	=&\; r(t)\bigg[u^{j-1}_{r}(t,\tau)-u^{j}_{r}(t,\tau)\bigg]. %\\
	%=&\; r(t)\bigg[\frac1{r(t)} h^{j-1}_{r}(t,\tau)-\frac1{r(t)} h^{j}_{r}(t,\tau)\bigg].
	\end{split}\end{equation}
	%\qed 
\end{proof}

The above lemma allows us to generalize Erlang-based results like Theorem \ref{Th:base} to their time-varying counterparts, i.e., Poisson process $k^\text{th}$ event time distributions with a time-dependent (or state-dependent) rate $r(t)$, as in the following generalization of the Standard LCT (Theorem \ref{Th:base}).

\begin{theorem}[Extended LCT for dwell times distributed as Poisson process $k^\text{th}$ event times]\label{Th:baseNHPP}
	Consider the Standard LCT in Theorem \ref{Th:base} but where the dwell time distribution is a Poisson process $k^\text{th}$ event time distribution with rate $r(t)$. Denote the survival function for the distribution of time spent in Y as $S_Y$. The corresponding mean field integral equations, written in terms of $h_r^j$ and $\mathcal{S}_r^j$ from eqs. \eqref{eq:ktheventdist}, are
	
	\begin{subequations}\label{eq:XY2NHPP} \begin{align} 
		x(t) =&\; x_0\,\mathcal{S}_r^k(t,0) + \int^t_{0} \mathcal{I}_X(s)\,\mathcal{S}_r^k(t,s) ds \label{eq:X2NHPP}\\
		\begin{split} \label{eq:Y2NHPP} y(t) =&\; y_0S_Y(t,0) + \int^t_{0} \bigg(\mathcal{I}_Y(\tau) + x_0\,h_r^k(\tau,0) \\ & \qquad \qquad \qquad +\int^\tau_0 \mathcal{I}_X(s)\,h^k_r(\tau,s)ds\bigg)S_Y(t,\tau)d\tau. \end{split}
		\end{align} \end{subequations}
	The above eqs. \eqref{eq:XY2NHPP} are equivalent to 
	
	\begin{subequations}\label{eq:chain2NHPP}\begin{align} 
		\frac{d}{dt}{x_1}(t) =&\; \mathcal{I}_X(t) -r(t)\, x_1(t) \label{eq:chain2aNHPP} \\
		\frac{d}{dt}{x_j}(t) =&\; r(t)\, x_{j-1}(t) - r(t)\, x_j(t), \quad j=2,\ldots,k  \label{eq:chain2bNHPP} \\
		y(t) =&\; y_0\,S_Y(t,0) + \int^t_{0} \left(\mathcal{I}_Y(\tau) + r(\tau)\,x_k(\tau)\right)S_Y(t,\tau)d\tau \label{eq:chain2cNHPP} 
		\end{align}\end{subequations}
	where $x(t) = \sum_{j=1}^{k} x_j(t)$ with initial conditions $x_1(0)=x_0$, $x_j(0)=0$ for $j\geq2$ and 
	
	\begin{equation}\label{eq:Xj2NHPP}
	x_j(t)= x_0\,\frac1{r(t)}\,h_r^j(t,0) + \int^t_{0} \mathcal{I}_X(s) \frac1{r(t)}\,h_r^j(t,s)ds. % substituted in terms of g(t-s)
	\end{equation}
	
	As in previous cases, the $y(t)$ equation \eqref{eq:chain2cNHPP} may be further reduced to ODEs, e.g., according to Corollary \ref{Th:baseCorr}.
	
\end{theorem} % END THEOREM

\begin{proof} Substituting eq. \eqref{eq:sktheventdist} into eq. \eqref{eq:X2NHPP} and substituting eq. \eqref{eq:Xj2NHPP} yields $x(t) = \sum_{j=1}^k x_j(t)$. Differentiating eq. \eqref{eq:Xj2NHPP} with $j=1$ using the Liebniz integration rule as well as eq. \eqref{eq:smith711NHPPa} from Lemma \ref{L:dh} yields eq. \eqref{eq:chain2aNHPP}. Likewise, for $j\geq2$, differentiation of eq. \eqref{eq:Xj2NHPP} and Lemma \ref{L:dh} yields
	
	\begin{equation} \begin{split}
	\frac{d}{dt}x_j(t) =&\; x(0)\,r(t)\,\bigg[ \frac1{r(t)} h^{j-1}_{r}(t,0) - \frac1{r(t)} h^{j}_{r}(t,0) \bigg] \\
	& + \int_0^t  \mathcal{I}_X(s)\,r(t)\bigg[ \frac1{r(t)} h^{j-1}_{r}(t,\tau) - \frac1{r(t)} h^{j}_{r}(t,\tau) \bigg] ds \\
	=&\; r(t) \big(x_{j-1}(t) - x_j(t)\big).
	\end{split}
	\end{equation}
	Eq. \eqref{eq:chain2cNHPP} follows from substituting \eqref{eq:Xj2NHPP} into \eqref{eq:Y2NHPP}. The definition of $x_j$ and initial condition $x(0)=x_0$ together imply $x_1(0)=x_0$ and $x_j(0)=0$ for the remaining $j\geq2$.
	%\qed 
\end{proof}

% Note that Theorems \ref{Th:simple} and \ref{Th:base} are special cases of Theorem \ref{Th:baseNHPP}.

Having generalized the Standard LCT (Lemma \ref{lem:smith711} and Theorem \ref{Th:base}) to include Poisson process $k^\text{th}$ event time distributed dwell times (compare Lemmas \ref{lem:smith711} and \ref{L:dh}, and compare eqs. \eqref{eq:chain2} in Theorem \ref{Th:base} to eqs. \eqref{eq:chain2NHPP} in Theorem \ref{Th:baseNHPP}), we may now address more complex assumptions about the underlying stochastic state transition model.

% % % % % % % % % % % % % % % % % % % % % % % % % % % % % % % % % % % % %
\subsection{Transitions to multiple states}\label{sec:tomulti}

\begin{figure}[htb]
	\centerline{\includegraphics[trim=0 15 0 15,clip,width=0.35\textwidth]{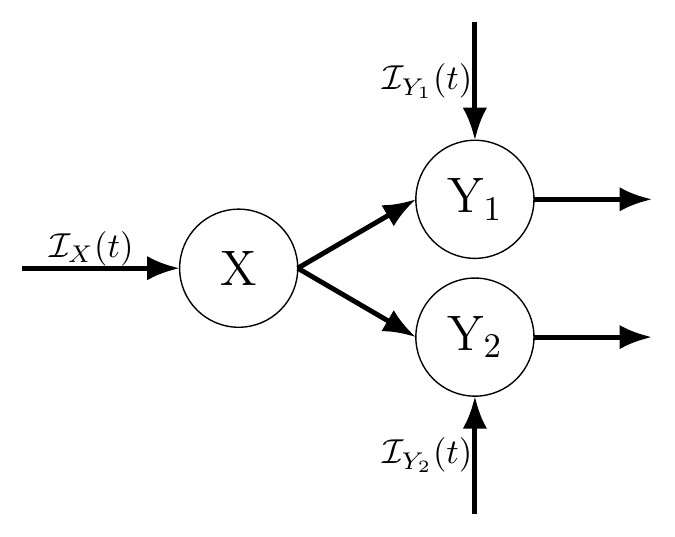}}
	\caption{Example diagram of transitions out of a given state (X) and into multiple states (Y$_1$ and Y$_2$). Different assumptions about (1) the dwell times in X, and (2) rules governing the subsequent transitions to Y$_1$ and/or Y$_2$ will lead to different sub-state partitions of X, and thus different mean field equations, as detailed in \S\ref{sec:tomulti}. Fortunately, different scenarios often encountered in applications can be reduced to ODEs by applying the results in \S\ref{sec:tomulti} as detailed in Theorems \ref{Th:tomultibasic}, \ref{Th:splitfirst}, and \ref{Th:comp} and as illustrated in Figs. \ref{fig:tomultibasic}, \ref{fig:tomultimix}, \ref{fig:mixture}, and \ref{fig:tomultimin}. }
	\label{fig:tomulti}
\end{figure}

Modeling the transition from one state to multiple states following a distributed delay (as illustrated in Fig. \ref{fig:tomulti}) can be done under different sets of assumptions about the underlying stochastic processes, particularly with respect to the rules governing how individuals are distributed across multiple recipient states and how those rules depend on the dwell time distribution(s) for individuals in that state. Importantly, those different sets of assumptions can yield very different mean field models \citep[e.g., see][]{Feng2016} and so care must be taken to make those assumptions appropriately for a given 	application. While modelers have some flexibility to choose appropriate assumptions, in practice modelers sometimes unintentionally make inappropriate assumptions, especially when constructing ODE models using ``rules of thumb" instead of deriving them from first principles. In this section we present results aimed at helping guide (a) the process of picking appropriate dwell time distribution assumptions, and (b) directly constructing corresponding systems of ODEs without deriving them from explicit stochastic models or intermediate integral equations. 

First, in \S \ref{sec:tomultibasic}, we consider the extension of Theorem \ref{Th:baseNHPP} where upon leaving X particles are distributed across $m\geq1$ recipient states according to a generalized Bernoulli distribution with (potentially time varying) probabilities/proportions $p_{j}(t)$, $j=1,\ldots,m$. Here the outcome of which state a particle transitions to is independent of the time spent in the first state.

Second, in \S \ref{sec:tomultisub-states} and \S \ref{sec:mixture}, particles entering the first state (X) do not all follow the same dwell time distribution in X. Instead, upon entering X they are distributed across $n\geq2$ sub-states of X, X$_i$, according to a generalized Bernoulli distribution, and each sub-state X$_i$ has a dwell time given by a Poisson process $k_i^\text{th}$ event time distribution with rate $r_i(t)$. That is, the X dwell time is a finite mixture of Poisson process event time distributions. Particles transition out of X into $m$ subsequent states Y$_j$ according to the probabilities/proportions $p_{ij}(t)$, the probability of going to Y$_j$ from X$_i$, $i=1,\ldots,n$ and $j=1,\ldots,m$. Here the determination of which recipient state Y$_\ell$ a particle transitions to depends on which sub-state of X the particle was assigned to upon entering X (see Fig. \ref{fig:tomultimix}). 

Third, in \S \ref{sec:tomultimin}, the outcome of which recipient state a particle transitions to upon leaving X is not independent of the time spent in the first state (as in \S \ref{sec:tomultibasic}), nor is it pre-determined upon entry into X (as in \S \ref{sec:tomultisub-states} and \S \ref{sec:mixture}). This result is obtained using yet another novel extension of Lemma \ref{lem:smith711} in which the dwell time in state X is the minimum of $n\geq2$ independent Poisson process event time distributions.

Each of these cases represents some of the different underlying stochastic model assumptions that can be made to construct a mean field ODE model for the scenario depicted in Fig. \ref{fig:tomulti}.

Lastly (\S \ref{sec:equivcompmulti}), we describe an equivalence between 1) the more complex case addressed in \S \ref{sec:tomultimin} assuming a dwell time that obeys the minimum of Poisson process $1^\text{st}$ event times, before being distributed across $m$ recipient states, and 2) the conceptually simpler case in \S \ref{sec:tomultibasic} where the dwell time follows a single Poisson process $1^\text{st}$ event time distribution before being distributed among $m$ recipient states. This is key to understanding the scope of the Generalized Linear Chain Trick in \S \ref{sec:glct}.

\subsubsection{Transition to multiple states independent of the X dwell time distribution}\label{sec:tomultibasic}

Here we extend the case in the previous section and assume that, upon leaving state X, particles can transition to one of $m$ states (call them $Y_i$, $i=1,\ldots,m$), and that a particle leaving X at time $t$ enters state $Y_i$ with probability $p_i(t)$, where $\sum_{i=1}^m p_i(t)=1$ (i.e., particles are distributed across all Y$_i$ following a generalized Bernoulli distribution with parameter vector $\mathbf{p}(t)=(p_1(t),\ldots,p_m(t))$). See Fig. \ref{fig:tomultibasic} for a simple example with constant $\mathbf{p}$ and $m=2$. An important assumption in this case is that the determination about which state a particle goes to after leaving X is made once it leaves X, and thus the state it transitions to is determined independent of the dwell time in X. Examples from the literature include Model II in \citet{Feng2016}, where infected individuals (state X) either recovered (Y$_0$) or died (Y$_1$) after an Erlang distributed time delay. 

\begin{theorem}[Extended LCT with proportional output to multiple states]\label{Th:tomultibasic} Consider the case addressed by Theorem \ref{Th:baseNHPP}, and further assume particles go to one of $m$ states (call them Y$_j$) with $p_j(t)$ being the probability of going to Y$_j$. Let $S_j$ be the survival functions for the dwell times in Y$_j$. 
	
	The mean field integral equations for this case, with $x(0)=x_0$ and $y_j(0)=y_{j0}$, are 
	
	\begin{subequations}\label{eq:delaysplit} \begin{align} 
		x(t) =& x_0\,\mathcal{S}_{r}^{k}(t,0) + \int_0^t\mathcal{I}(s)\,\mathcal{S}_{r}^{k}(t,s)\,ds \label{eq:delaysplitx} \\
		\begin{split}  \label{eq:delaysplityj}  y_j(t) =& y_j(0)S_j(t,0) + \int_0^t\bigg(\mathcal{I}_j(\tau)+p_j(t)\bigg( x_0\,h_{r}^{k}(\tau,0)   \\ & \qquad \qquad \qquad \quad + \int_0^\tau\mathcal{I}(s)h_{r}^{k}(\tau,s)\,ds \bigg)\bigg)S_j(t,\tau) d\tau  \end{split}
		\end{align}\end{subequations} These integral equations are equivalent to the following system of equations: 
	
	\begin{subequations}\label{eq:delaysplitODEs}\begin{align} 
		\frac{d}{dt}{x_1}(t) =&\; \mathcal{I}(t) -r(t)\, x_1(t)  \label{eq:delaysplitODEsa}\\
		\frac{d}{dt}{x_i}(t) =&\; r(t)\, x_{i-1}(t) - r(t)\, x_i(t), \quad i=2,\ldots,k  \label{eq:delaysplitODEsb} \\
		y_j(t) =&\; y_j(0)S_j(t,0) + \int_0^t\bigg(\mathcal{I}_j(\tau)+p_j(t)\;r(t)\,x_k(\tau)\bigg)S_j(t,\tau) d\tau  \label{eq:delaysplitODEsc}
		%\frac{d}{dt}{y_j}(t) =&\; p_j\,r x_{k}(t) - \mu_j(t) y_j(t), \quad j=2,\ldots,n   \label{eq:delaysplitODEsc} 
		\end{align}\end{subequations}
	where $x(t)=\sum_{i=1}^k x_i(t)$, $x_1(0)=x_0$, $x_i(0)=0$ for $i\geq 2$, and 
	
	\begin{equation}
	x_i(t)= x_0 \frac1{r(t)} h_{r}^{k}(t,0) + \int_0^t\mathcal{I}(s) \frac1{r(t)} h_{r}^{k}(t,s)\,ds. \label{eq:delaysplitxi} \\
	\end{equation}
	The $m$ integral equations \eqref{eq:delaysplitODEsc} may be further reduced to a system of ODEs, e.g., via Corollary \ref{Th:baseCorr}.
	
\end{theorem}

\begin{proof} Equations \eqref{eq:delaysplitODEsa}, \eqref{eq:delaysplitODEsb} and \eqref{eq:delaysplitxi} follow from Theorem \ref{Th:baseNHPP}. Eq. \eqref{eq:delaysplitODEsc} follows from substitution of eq. \eqref{eq:delaysplitxi} into \eqref{eq:delaysplityj}. The derivation of eq. \eqref{eq:delaysplityj} is similar to the derivation in Appendix \ref{a:example1} but accounts for the expected proportion entering each Y$_j$ at time $t$ being equal to $p_j(t)$. 	
	%\qed
\end{proof}

\begin{figure}[htb]
	\centerline{\includegraphics[trim=0 15 0 15,clip,width=0.7\textwidth]{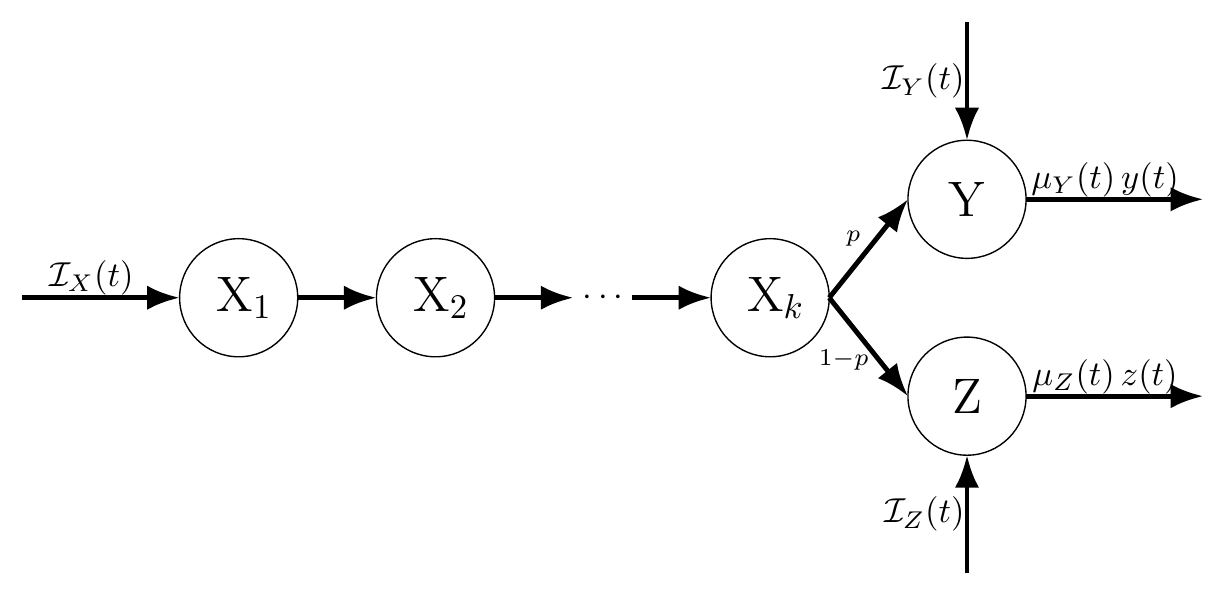}}
	\caption{The special case of Fig. \ref{fig:tomulti} under the assumptions of Theorem \ref{Th:tomultibasic} (Extended LCT with proportional outputs to multiple states; see \S \ref{sec:tomultibasic}). Specifically, this case assumes that (1) the dwell time distribution for X is Erlang($r,k$), and (2) upon exiting X particles are distributed to multiple recipient states, here Y and Z, with probabilities $p$ and $1-p$, respectively. }
	\label{fig:tomultibasic}
\end{figure}

\begin{example} Consider the example shown in Figure \ref{fig:tomultibasic}, where the dwell time distribution for X is Erlang($r,k$) and the dwell times in Y and Z follow 1$^\text{st}$  event times under nonhomogeneous Poisson processes with respective rates $\mu_Y(t)$ and $\mu_Z(t)$. The corresponding mean field ODEs, given by Theorem \ref{Th:tomultibasic}, are 
	
	\begin{subequations}\label{eq:gammaproportionODE}\begin{align} 
		\frac{d}{dt}{x_1}(t) =&\; \mathcal{I}(t) -r\, x_1(t)  \label{eq:gammaproportionODEx1}\\
		\frac{d}{dt}{x_i}(t) =&\; r\, x_{i-1}(t) - r\, x_i(t), \quad i=2,\ldots,k  \label{eq:gammaproportionODExi} \\
		\frac{d}{dt}{y}(t) =&\; \mathcal{I}_Y(t) + p\,r\, x_{k} - \mu_Y(t) y(t)   \label{eq:gammaproportionODEy} \\
		\frac{d}{dt}{z}(t) =&\; \mathcal{I}_Z(t) + (1-p)\,r\, x_{k}(t) - \mu_Z(t) z(t).  \label{eq:gammaproportionODEz} 
		\end{align}\end{subequations}
	
\end{example}

% % % % % % % % % % % % % % % % % % % % % % % % % % % % % % % % % % % % %

\subsubsection{Transition from sub-states of X with differing dwell time distributions and differing output distributions across states Y$_j$}\label{sec:tomultisub-states}

We next consider the case where particles in a given state X can be treated as belonging to a heterogeneous population, where each remains in that state according to one of $N$ possible dwell time distributions, the $i^\text{th}$ of these being the $k_i^\text{th}$ event time distribution under a Poisson process with rate $r_i(t)$). Each particle is assigned one of these $N$ dwell time distributions (i.e., it is assigned to sub-state X$_i$) upon entry into X according to a generalized Bernoulli distribution with a (potentially time varying) probability vector $\rho(t)=(\rho_1(t),\ldots,\rho_N(t))$. In contrast to the previous case, here the outcome of which recipient state a particle transitions to is not necessarily independent of the dwell time distribution. 

Note that the above assumptions imply that the dwell time distribution for state X is a finite mixture of event time distributions under $N$ independent Poisson processes.  If a random variable $T$ is a mixture of Erlang distributions, or more generally a mixture of $N$ independent Poisson process event time distributions, then the corresponding density function ($f$) and survival function ($\Phi$) are

\begin{subequations}\label{eq:mixturedist}\begin{align}
	f_\theta(t,\tau) =&\; \sum_{i=1}^N \rho_i(\tau)\, h_{r_i}^{k_i}(t,\tau) \label{eq:mixturef} \\
	\Phi_\theta(t,\tau) =&\; \sum_{i=1}^N \rho_i(\tau)\, \mathcal{S}_{r_i}^{k_i}(t,\tau) =\; \sum_{i=1}^N \rho_i(\tau) \sum_{j=1}^{k_i} \frac1{r_{i}(t)}\,h_{r_{i}}^j(t,\tau) \label{eq:mixtureS} 
	\end{align}
\end{subequations}
where the (potentially time varying) parameter vector $\theta(t)=$($\rho_1$, $r_1(t)$, $k_1$, $\ldots$, $\rho_N$, $r_N(t)$, $k_N$) is the potentially time varying parameter vector for the $N$ distributions that constitute the mixture distribution, with $\sum_{i=1}^N \rho_i(t) =1$. Note that if all $r_i(t)=r_i$ are constant, this becomes a mixture of independent Erlang distributions, or if additionally all $k_i=1$, a mixture of independent exponentials.

\begin{theorem}[Extended LCT for dwell times given by mixtures of Poisson process event time distributions and outputs to multiple states] \label{Th:splitfirst}
	Consider a continuous time state transition model with inflow rate $\mathcal{I}(t)$ into state X. Assume that the duration of time spent in state X follows a finite mixture of $N$ independent Poisson process event time distributions. That is, X can be partitioned into $N$ sub-states X$_i$, $i=1,\ldots,N$, each with dwell time distributions given by a Poisson process $k_i^\text{th}$ event time distributions with rates $r_i(t)$. Suppose the inflow to state X at time $t$ is distributed among this partition according to a generalized Bernoulli distribution with probabilities $\rho_i(t)$, where $\sum_{i=1}^{N} \rho_i(t) = 1$, so that the input rate to X$_i$ is $\rho_i(t)\mathcal{I}(t)$. Assume that particles leaving sub-state X$_i$ then transition to state Y$_\ell$ with probability $p_{i\ell}(t)$, $\ell=1,\ldots,m$, where the duration of time spent in state Y$_\ell$ follows a delay distribution give by survival function $S_j$. Then we can partition each X$_i$ into X$_{ij}$, $j=1,\ldots,k_i$, according to Theorem \ref{Th:baseNHPP} and let $x(t)$, $x_i(t)$, $x_{ij}(t)$, and $y_\ell(t)$ be the amounts in states X, X$_i$, X$_{ij}$, and Y$_\ell$ at time $t$, respectively. Assume non-negative initial conditions $x(0)=x_0$, $x_i(0)=\rho_i(0)x_0$, $x_{i1}(0)=\rho_i(0)\,x_0$, $x_{ij}(0)=0$ for $j\geq2$, and $y_\ell(0)\geq0$. 
	
	The mean field integral equations for this scenario are 
	
	\begin{subequations}\label{eq:splitXYmix} \begin{align} 
		x(t) =&\; x_0\,\Phi_{\theta}(t,0) + \int^t_{0} \mathcal{I}(s)\,\Phi_{\theta}(t,s) ds  \label{eq:splitXYmixX}\\
		\begin{split} y_\ell(t) =&\; y_\ell(0) S_\ell(t,0)  + \int^t_{0} \bigg(\mathcal{I}_\ell(\tau) +  \sum_{i=1}^N p_{ij}(\tau) \bigg( x_0\, \rho_i(\tau) \,h_{r_i}^{k_i}(\tau,0) \\
		& \qquad \qquad \qquad \quad  +\int^\tau_0 \rho_i(s)\, \mathcal{I}(s)\,h_{r_i}^{k_i}(\tau, s) ds\bigg)\bigg)  S_\ell (t,\tau)d\tau. \label{eq:splitYmix}  
		\end{split}
		\end{align} \end{subequations}
	
	The above system of equations \eqref{eq:splitXYmix} are equivalent to 
	
	\begin{subequations}\label{eq:splitmixchain}\begin{align} 
		\frac{d}{dt}{x_{i1}}(t) =&\;  \rho_i(t)\, \mathcal{I}(t) - r_i(t)\,x_{i1}(t),  \quad i=1,\ldots,N  \label{eq:splitmixchaina} \\
		\frac{d}{dt}{x_{ij}}(t) =&\; r_i(t) \big(x_{i,j-1}(t) - x_{ij}(t)\big), \quad i=1,\ldots,N; \; j=2,\ldots,k_i  \label{eq:splitmixchainb} \\
		%y_i(t) =&\; y_i(0)S_i(t,0) + \int^t_{0} \left(\mathcal{I}_i(\tau) + r_i\,x_{i,k_i}(\tau)\right)S_i(t,\tau)d\tau. \label{eq:splitmixchainc}
		\begin{split} \label{eq:splitmixchainc} y_\ell(t) =&\; y_\ell(0) S_\ell(t,0) + \int^t_{0} \bigg(\mathcal{I}_\ell(\tau) \\ & \qquad \qquad +  \sum_{i=1}^N r_i(t)\,x_{ik_i}(\tau)\,p_{i\ell}(\tau) \bigg)  S_\ell (t,\tau)d\tau \end{split}
		\end{align}\end{subequations} with initial conditions $x_{i1}(0)=\rho_i(0)\,x_0$, $x_{ij}(0)=0$ for $j\geq2$, where $x(t) = \sum_{i=1}^{N} x_i(t)$, and $x_i(t) = \sum_{j=1}^{k_i} x_{ij}(t)$. The amounts in each X$_i$, and in sub-states X$_{ij}$, are given by 
	
	\begin{align} 
	x_i(t) =&\; \rho_i(0)\,x_0\,\mathcal{S}_{r_i}^{k_i}(t,0) + \int^t_{0} \rho_i(s)\,\mathcal{I}(s)\,\mathcal{S}_{r_i}^{k_i}(t,s) ds \label{eq:splitXmix}\\
	x_{ij}(t) =&\; \rho_i(0)\,x_0 \frac{h_{r_i}^j(t,0)}{r_i(t)} + \int^t_{0} \rho_i(s)\mathcal{I}(s)\;\frac{h_{r_i}^j(t,s)}{r_i(t)} \,ds. \label{eq:splitXjmix} % The ith chain, jth state in that chain.
	\end{align} The $m$ integral equations \eqref{eq:splitmixchainc} for $y_\ell(t)$ may be reduced to ODEs, e.g., via Corollary \ref{Th:baseCorr}.
	
\end{theorem}

\begin{proof} Substituting eq. \eqref{eq:mixtureS} into eq. \eqref{eq:splitXYmixX} and then substituting eq. \eqref{eq:splitXmix} yields $x(t) = \sum_{i=1}^{N} x_{i}(t)$. Applying Theorem \ref{Th:baseNHPP} to each X$_i$ (i.e., to each eq. \eqref{eq:splitXmix}) then yields eqs. \eqref{eq:splitXjmix}, \eqref{eq:splitmixchaina} and \eqref{eq:splitmixchainb}. (Alternatively, one could prove this directly by differentiating eqs. \eqref{eq:splitXjmix} using eqs. \eqref{eq:smith711NHPP} from Lemma \ref{L:dh}).  The $y_\ell(t)$ equations \eqref{eq:splitmixchainc} are obtained from \eqref{eq:splitYmix} by substitution of eqs. \eqref{eq:splitXjmix}.
	%\qed 
\end{proof}

\begin{example}\label{ex:tomultimix} 
	
	\begin{figure}[!htb]
		\centerline{\includegraphics[trim=0 15 0 15,clip,width=0.7\textwidth]{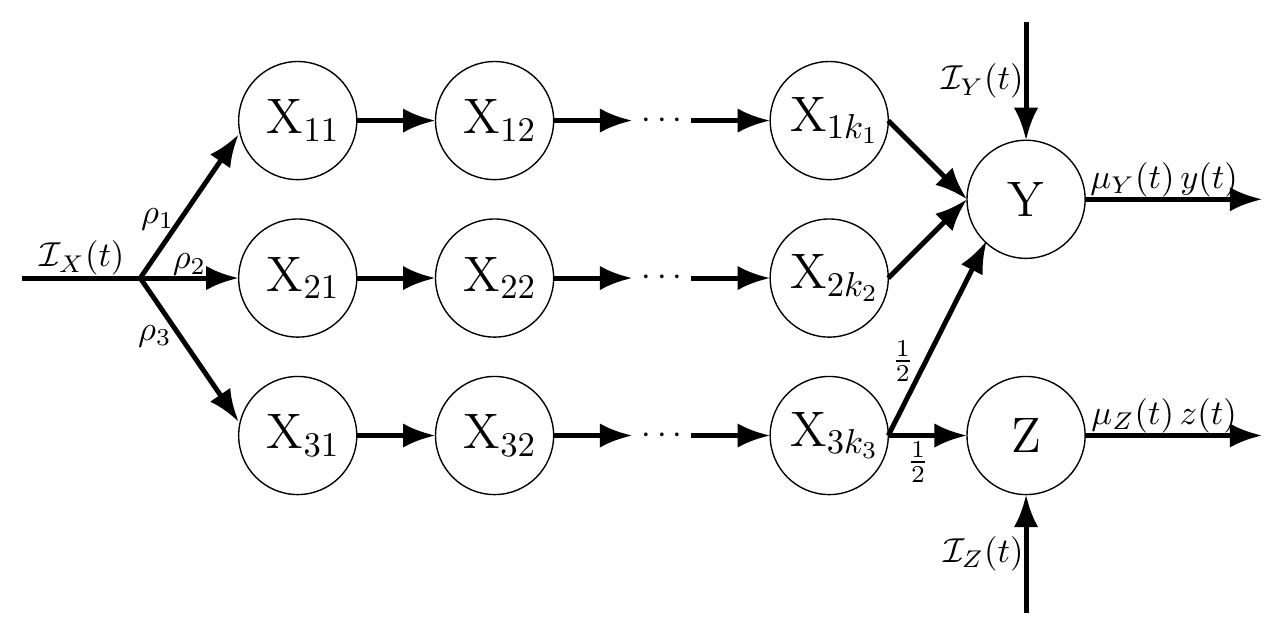}}
		\caption{The sub-state diagram (cf. Fig. \ref{fig:tomulti}) resulting from the application of Theorem \ref{Th:splitfirst} (Extended LCT for finite mixtures of Poisson process event time distributions with output to multiple states) to Example \ref{ex:tomultimix}, where upon entering X particles have an Erlang($r_i,k_i$) dwell time in X with probability $\rho_i$, $i=1,\ldots,3$. Thus, the overall dwell time in X follows an Erlang mixture distribution (see \S\ref{sec:mixture}).}
		\label{fig:tomultimix}
	\end{figure}
	
	Suppose particles entering state X at rate $\mathcal{I}_X(t)$ enter sub-state X$_1$ with probability $\rho_1$, X$_2$ with probability $\rho_2$, and X$_3$ with probability $\rho_3=1-\rho_1-\rho_2$. Further assume particles in state X$_i$ remain there for an Erlang$(r_i,k_i)$ distributed amount of time, and that particles exiting X$_1$ and X$_2$ transition to Y with probability 1, while particles exiting X$_3$ transition either to state Y or Z with equal probability.  Assume particle may also enter states Y and Z from sources other than state X (at rates $\mathcal{I}_X(t)$ and $\mathcal{I}_X(t)$, respectively), and the dwell times in those two states follow the $1^\text{st}$ event times of independent nonhomogeneous Poisson processes with rates $\mu_Y(t)$ and $\mu_Z(t)$, respectively. Then Theorem \ref{Th:splitfirst} yields the following mean field system of ODEs (see Fig. \ref{fig:tomultimix}). 
	
	\begin{subequations}\begin{align} 
		\frac{d}{dt}{x_{i,1}}(t) =&\;  \rho_i\, \mathcal{I}(t) - r_i x_{i,1}(t), \quad i=1,\ldots,3,   \\
		\frac{d}{dt}{x_{i,j}}(t) =&\; r_i \big(x_{i,j-1}(t) - x_{ij}(t)\big), \quad  j=2,\ldots,k_i  \\
		\frac{d}{dt}{y}(t) =&\; r_1 \, x_{1,k_1}(t) + r_2\, x_{2,k_2}(t) + r_3 \frac12 x_{3,k_3}(t) - \mu_Y(t) y(t)  \\
		\frac{d}{dt}{z}(t) =&\; r_3 \frac12 x_{3,k_3}(t) - \mu_Z(t) z(t) .    
		\end{align}\end{subequations}
	
\end{example}

\subsubsection{Extended LCT for dwell times given by finite mixtures of Poisson process event time distributions}\label{sec:mixture}

It's worth noting here that it may be appropriate in some applied contexts to approximate a non-Erlang delay distribution with a mixture of Erlang distributions (see Appendix \ref{sec:approxgamma} for more details on making such approximations). The following corollary to Theorem \ref{Th:splitfirst} above (specifically, the $m=1$ case) details how assuming such a mixture distribution (or more generally, a finite mixture of independent nonhomogeneous Poisson process event times) would be reflected in the structure of the corresponding mean field ODEs (see Fig. \ref{fig:mixture}). 

\begin{figure}[thb]
	\centerline{\includegraphics[clip,width=0.7\textwidth]{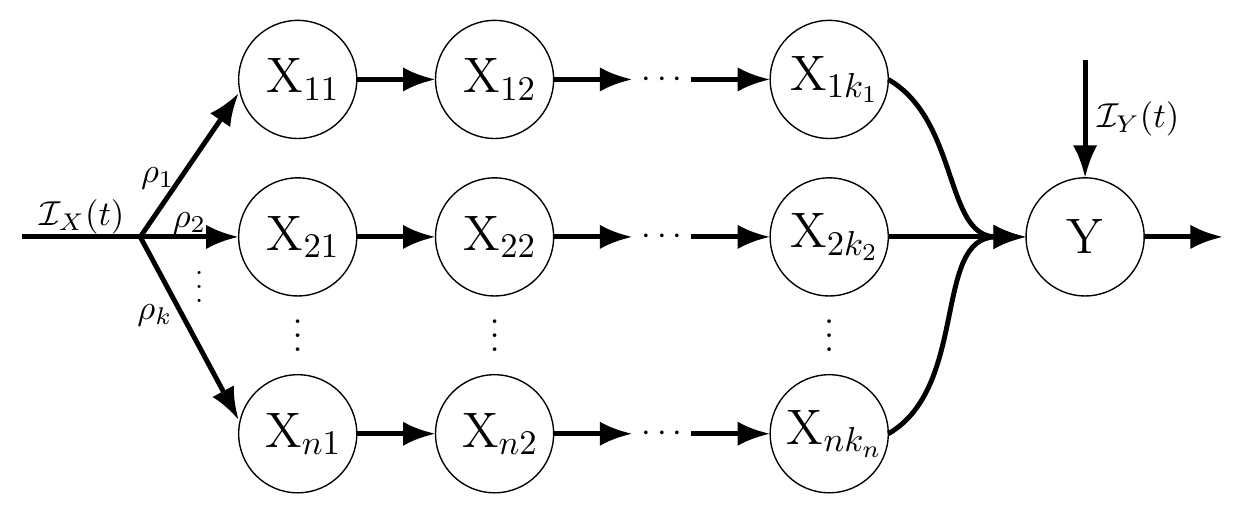}}
	\caption{Example sub-state diagram corresponding to the ODEs provided by Corollary \ref{Th:mixture} (Extended LCT for mixtures of Poisson process event time distributions with output to a single recipient state) where the dwell time in state X follows an Erlang mixture distribution (see section \ref{sec:defns}) as detailed in \S\ref{sec:mixture}.} 
	\label{fig:mixture}
\end{figure}

\begin{corollary}[Extended LCT for Poisson process event time mixture distributions]\label{Th:mixture}
	Consider the case addressed in Theorem \ref{Th:splitfirst} where the distribution of time spent in state X is a finite mixture of event time distributions under independent homogeneous or nonhomogeneous Poisson processes, and that upon leaving X particles enter a single state Y (c.f. Theorem \ref{Th:baseNHPP}). Then corresponding mean field equations are 
	
	\begin{subequations}\label{eq:mixchain}\begin{align} 
		\frac{d}{dt}{x_{i1}}(t) =&\;  \rho_i(t)\, \mathcal{I}_X(t) - r_i(t) x_1(t),  \quad i=1,\ldots,N  \label{eq:mixchaina} \\
		\frac{d}{dt}{x_{ij}}(t) =&\; r_i(t) \,x_{i,j-1}(t) - r_i(t) \,x_{ij}(t), \quad i=1,\ldots,N; \; j=2,\ldots,k_i  \label{eq:mixchainb} \\
		y(t) =&\; y_0\,S(t,0) + \int^t_{0} \left( \mathcal{I}_Y(\tau) + \sum_{i=1}^{N} r_i(\tau) \, x_{ik_i}(\tau) \right)S(t,\tau)d\tau.    \label{eq:mixchainc} 
		\end{align}\end{subequations}
	
	with initial conditions $x_{i1}(0)=\rho_i(t)\,x_0$, $x_{ij}(0)=0$ for $j\geq2$. Here $x(t) = \sum_{i=1}^N \sum_{j=1}^{k_i} x_{ij}(t)$ where $x_{ij}$ is the amount in the $j^\text{th}$ intermediate state in the $i^\text{th}$ linear chain.
\end{corollary} 

%%%%%%%%%%%%%%%%%%%%%%%%%%%%%%%%%%%%%%%%%%%%%%%%%%%%%%%%%%%%%%%%%%%%%%%%%%%%%%%%%%%%%%%%%%%%%%%%%%%%%%%%%%%%%%%%%
%\clearpage
\subsubsection{Transition to multiple states following ``competing" Poisson processes}\label{sec:tomultimin}

We now consider the case where $T$, the time a particle spends in a given state X, follows the distribution given by $T=\min_i T_i$, the minimum of $n\geq2$ independent random variables $T_i$, where $T_i$ has either an  Erlang($r_i,k_i$) distribution or, more generally, Poisson process $k_i^\text{th}$ event time distributions with rates $r_i(t)$. Upon leaving state X, particles have the possibility of transitioning to any of $m$ recipient states $Y_\ell$, $\ell=1,...,m$, where the probability of transitioning to state Y$_\ell$ depends on which of the $n$ random variables $T_i$ was the minimum. That is, if a particle leaves X at time $T=T_i=t$, then the probability of entering state Y$_\ell$ is $p_{i\ell(t)}$. 

The distribution associated with $T$ is not itself an Erlang distribution or a Poisson process event time distribution, however its survival function is the product\footnote{It is generally true that the survival function for a minimum of multiple independent random variables is the product of their survival functions.} of such survival functions, i.e., \begin{equation}\label{eq:scomp}\mathscr{S}(t,\tau)\equiv \prod_{i=1}^{n}\mathcal{S}_{r_i}^{k_i}(t,\tau).\end{equation} As detailed below, we can further generalize the recursion relation in Lemma \ref{lem:smith711} for the distributions just described above, which can then be used to produce a mean field system of ODEs based on appropriately partitioning X into sub-states.

Before considering this case in general, it is helpful to first describe the sub-states of X imposed by assuming the dwell time distribution described above, particularly the case where the distribution for each $T_i$ is based on $1^\text{st}$ event times (i.e., all $k_i=1$). Recall that the minimum of $n$ exponential random variables (which we may think of as $1^\text{st}$ event times under a homogeneous Poisson process) is exponential with a rate that is the sum of the individual rates $r=\sum_{i=1}^n r_i$. More generally, it is true that the minimum of $n$ $1^\text{st}$ event times under independent Poisson processes with rates $r_i(t)$ is itself distributed as the $1^\text{st}$ event time under a single Poisson processes with rate $r(t)\equiv\sum_{i=1}^n r_i(t)$, i.e., in this case $\mathscr{S}(t,\tau)= \prod_{i=0}^{n}\mathcal{S}_{r_i}^{1}(t,\tau)=\mathcal{S}_{r}^{1}(t,\tau)$. Additionally, if particles leaving state X are then distributed across the recipient states Y$_\ell$ as described above, then this scenario is equivalent to the proportional outputs case described in Theorem \ref{Th:tomultibasic} with a dwell time that follows a Poisson process 1$^\text{st}$  event time distribution with rate $r(t)\equiv\sum_{i=1}^n r_i(t)$ and a probability vector $p_\ell = \sum_{i=1}^n p_{i\ell}(t)r_i(t)/r(t)$, since $P(T=T_i)=r_i(T)/r(T)$. (This mean field equivalence of these two cases is detailed in \S \ref{sec:equivcompmulti}.) Thus, the natural partitioning of X in this case is into sub-states with dwell times that follow \textit{iid} $1^\text{st}$ event time distributions with rate $r(t)\equiv \sum_{i=1}^{N} r_i(t)$.

We may now describe the mean field ODEs for the general case above using the following notation. To index the sub-states of X, consider the $i^\text{th}$ Poisson process and its $k_i^\text{th}$ event time distribution which defines the distribution of $T_i$. Let $a_i\leq k_i$ denote the event number a particle is awaiting under the $i^\text{th}$ Poisson process. Then we can describe the particle's progress through X according to its progress along each of these $n$ Poisson processes using the index vector $\alpha \in \mathcal{K}$, where \begin{equation}\label{eq:Kset1}\mathcal{K}=\{(a_1,a_2,\ldots,a_n)\;|\;a_j\in\{1,\ldots,k_j\}\}.\end{equation} We will also use the notation $\mathcal{K}_{i}\subset\mathcal{K}$ which are the subset of indices where $a_i=k_i$ (where we think of particles in these sub-states as being poised to reach the $k_i^\text{th}$ event related to the $i^\text{th}$ Poisson process, and thus poised to transition out of state X). 

To extend Lemma \ref{L:dh} for these distributions, let $m_i(t,\tau)=\exp\big(-\int_{\tau}^t r_i(s)ds\big)$ and define 

\begin{equation}
u(t,\tau,\alpha) \equiv \prod_{i=1}^{n} e^{-m_i(t,\tau)} \frac{m_i(t,\tau)^{a_i-1}}{(a_i-1)!}. \label{eq:ucomp}
\end{equation} Note that $\prod_{i=1}^{n} h_{r_i}^{a_i}(t,\tau) = u(t,\tau,\alpha)\prod_{i=1}^n r_i(t)$ (c.f. Lemma \ref{L:dh}) and  $u(\tau,\tau,\alpha)=1$ if $\alpha=(1,\ldots,1)$ and $u(\tau,\tau,\alpha)=0$ otherwise. Then applying eq. \eqref{eq:sktheventdist} to $\mathscr{S}(t,\tau)$ (i.e., eq. \eqref{eq:scomp}) it follows that the survival function for the distribution of time spent in X in this case (c.f. eqs. \ref{eq:u} and \eqref{eq:sktheventdist}) can be written 

\begin{equation} \label{eq:scompsum} \mathscr{S}(t,\tau) = \,\sum_{\alpha\in\mathcal{K}}u(t,\tau,\alpha).\end{equation}  We will also refer to the quantities $u$ and $\mathscr{S}$ with the $j^\text{th}$ element of each product removed using the notation 

\begin{subequations} \label{eq:compsubj} \begin{align}
	u_{\setminus j}(t,\tau,\alpha)\equiv &\; \prod_{i=1, i\neq j}^n e^{-m_i(t,\tau)} \frac{m_i(t,\tau)^{a_i-1}}{(a_i-1)!} \label{eq:compusubj} \\ 
	\mathscr{S}_{\setminus j}(t,\tau) \equiv  &\;  \sum_{\alpha\in\mathcal{K}_j} u_{\setminus j}(t,\tau,\alpha). \label{eq:compSsubj}
	\end{align} \end{subequations}  

This brings us to the following lemma, which generalizes Lemma \ref{lem:smith711} and Lemma \ref{L:dh} to distributions that are the minimum of $n$ different (independent) Poisson process event times. As with the above lemmas, Lemma \ref{L:df} will allow one to partition X into sub-states corresponding to each of the event indices in $\mathcal{K}$ describing the various stages of progress along each Poisson process prior to the first of them reaching the target event number.
\begin{lemma}\label{L:df} For $u$ as defined in eq. \eqref{eq:ucomp}, differentiation with respect to $t$ yields 
	
	\begin{equation} \label{eq:smith711fcomp}
	\frac{d}{dt}u(t,\tau,\alpha) = \; \sum_{j=1}^n r_j(t)\,u(t,\tau,\alpha_{j,-1})\,\1{a_j>1}(\alpha) - \sum_{j=1}^n r_j(t) u(t,\tau,\alpha)
	\end{equation} where the notation $\alpha_{j,-1}$ denotes the index vector generated by decrementing the $j^\text{th}$ element of $\alpha$, $a_j$ (assuming $a_j>1$; for example, $\alpha_{2,-1}=(a_1,a_2-1,\ldots,a_n)$), and the indicator function $\1{a_j>1}(\alpha)$ is 1 if $a_j>1$ and 0 otherwise. \\
\end{lemma}

\begin{proof} Using the definition of $u$ in eq. \eqref{eq:ucomp} above, it follows that  
	
	\begin{equation} \begin{split}
	\frac{d}{dt}&u(t,\tau,\alpha) =\; \frac{d}{dt}\prod_{i=1}^{n} e^{-m_i(t,\tau)} \frac{m_i(t,\tau)^{a_i-1}}{(a_i-1)!}\\
	&=\sum_{j=1}^n\bigg(\prod_{\substack{i=1\\ i\neq j}}^n e^{-m_i(t,\tau)} \frac{m_i(t,\tau)^{a_i-1}}{(a_i-1)!}\bigg)\bigg[-r_j(t)e^{-m_j(t,\tau)} \frac{m_j(t,\tau)^{a_j-1}}{(a_j-1)!}  \\
	&\qquad + \1{a_j>1}(\alpha)\,r_j(t)\,e^{-m_j(t,\tau)}\frac{m_j(t,\tau)^{a_j-2}}{(a_j-2)!}\bigg] \\
	&=\sum_{j=1}^n -r_j(t)\,\prod_{i=1}^{n} e^{-m_i(t,\tau)} \frac{m_i(t,\tau)^{a_i-1}}{(a_i-1)!} \; + \\
	&\quad \sum_{j=1}^n \1{a_j>1}(\alpha)\,r_j(t)\,e^{-m_j(t,\tau)}\frac{m_j(t,\tau)^{a_j-2}}{(a_j-2)!} \prod_{\substack{i=1\\i\neq j}}^n e^{-m_i(t,\tau)} \frac{m_i(t,\tau)^{a_i-1}}{(a_i-1)!} \\
	&= \sum_{j=1}^n r_j(t)\,u(t,\tau,\alpha_{j,-1})\,\1{a_j>1}(\alpha) - \sum_{j=1}^n r_j(t) u(t,\tau,\alpha).
	\end{split}
	\end{equation}
	%\qed 
\end{proof}

The next theorem details the LCT extension that follows from Lemma \ref{L:df}.

\begin{theorem}[Extended LCT for dwell times given by competing Poisson processes]\label{Th:comp}
	Consider a continuous time dynamical system model of mass transitioning among multiple states, with inflow rate $\mathcal{I}_X(t)$ to a state X. The distribution of time spent in state X (call it $T$) is the minimum of $n$ random variables, i.e., $T=\min_{i}(T_i)$, $i=1,\ldots,n$, where $T_i$ are either Erlang($r_i,k_i$) distributed or follow the more general (nonhomogeneous) Poisson process $k_i^\text{th}$ event time distribution with rate $r_i(t)$. Assume particles leaving X can enter one of $m$ states Y$_\ell$, $\ell=1,\ldots,m$. If a particle leaves X at time $T_i$ (i.e., $T_i$ occurred first, so $T=T_i$), and then the particle transitions into state $Y_\ell$ with probability $p_{i\ell}(T)$. Let $x(t)$, and $y_\ell(t)$ be the amount in each state, respectively, at time $t$, and assume non-negative initial conditions. %Further assume inflow rates $\mathcal{I}_i(t)$ to $Y_i$ are determined by integrable non-negative Nonhomogeneous Poisson processes.
	
	The mean field integral equations for this scenario, for $\ell=1,\ldots,m$ and $i=1,\ldots,n$, are 
	
	\begin{subequations}\label{eq:compIE} \begin{align} 
		x(t) =&\; x_0\,\mathscr{S}(t,0) + \int^t_{0} \mathcal{I}_X(s)\,\mathscr{S}(t,s)ds \label{eq:compIEx}\\
		\begin{split}	y_\ell(t) =&\; y_{\ell}(0)S_\ell(t,0) +  \int^t_{0} \bigg(\mathcal{I}_\ell(\tau) + \sum_{i=1}^n p_{i\ell} \bigg( x_0\,\mathscr{S}_{\setminus i}(\tau,0)\,h_{r_i}^{k_i}(\tau,0) \; + \\ 
		& \qquad \int_{0}^{\tau} \mathcal{I}_X(s) \mathscr{S}_{\setminus i}(\tau,s)\,h_{r_i}^{k_i}(\tau,s)  ds\bigg)\bigg)S_\ell(t,\tau)d\tau. \label{eq:compIEy}
		\end{split}
		\end{align} \end{subequations}
	
	Equations \eqref{eq:compIE} above are equivalent to 
	
	\begin{subequations}\label{eq:chain3}\begin{align} 
		\frac{d}{dt}x_{(1,\ldots,1)}(t) =&\; \mathcal{I}_X(t) - r(t)\,x_{(1,\ldots,1)}(t),  \label{eq:chain3a} \\
		\frac{d}{dt}x_\alpha(t) =&\;\sum_{i=1}^{n} r_i(t)\,x_{\alpha_{i,-1}}(t)\,\1{a_i>1}(\alpha) \;-\; r(t) \, x_{\alpha}(t) \label{eq:chain3b} \\
		\begin{split} \label{eq:chain3c} y_\ell(t) =&\; y_{\ell}(0) S_\ell(t,0)  \\ & \; + \int^t_{0} \bigg(\mathcal{I}_\ell(\tau) + \sum_{i=1}^n p_{i\ell}(\tau)\sum_{\alpha\in\mathcal{K}_i} r_i(t)\,x_\alpha(\tau) \bigg)S_\ell(t,\tau)d\tau   \end{split}
		\end{align}\end{subequations}
	for all $\alpha\in\mathcal{K}\setminus(1,\ldots,1)$,  $r(t)=\sum_{i=1}^{n}r_i(t)$, $x(t)=\sum_{\alpha\in\mathcal{K}} x_\alpha(t)$, and 
	
	\begin{equation} \label{eq:xalpha} x_\alpha(t) = x_0\,u(t,0,\alpha) + \int_0^t\mathcal{I}_X(s)\,u(t,s,\alpha)\,ds. 	\end{equation}
	
	The $y_\ell(t)$ equations \eqref{eq:chain3c} may be further reduced to a system of ODEs, e.g., via Corollary \ref{Th:baseCorr}.

\end{theorem} % END Th:comp

\begin{proof}	
	Substituting eq. \eqref{eq:scompsum} into eq. \eqref{eq:compIEx} yields 	
	
	\begin{equation}\label{eq:compODE}\begin{split}
	x(t) =&  x_0\,\sum_{\alpha\in\mathcal{K}} u(t,0,\alpha) + \int^t_{0} \mathcal{I}_X(s)\,\sum_{\alpha\in\mathcal{K}}u(t,s,\alpha)ds \\
	=& \sum_{\alpha\in\mathcal{K}} \bigg( x_0\, u(t,0,\alpha) + \int^t_{0} \mathcal{I}_X(s)\, u(t,s,\alpha)ds \bigg) 
	= \sum_{\alpha\in\mathcal{K}} x_\alpha(t).
	\end{split}
	\end{equation}
	
	Differentiating \eqref{eq:xalpha} yields equations eqs. \eqref{eq:chain3a} and \eqref{eq:chain3b} as follows. First, if $\alpha=(1,\ldots,1)$ then by Lemma \ref{L:df} 
	
	\begin{equation} \begin{split}
	\frac{d}{dt}{x_{(1,\ldots,1)}}(t) =&\; -x_0\,\sum_{i=1}^n r_i(t)\,u(t,0,\alpha) \\ & \quad - \sum_{i=1}^n r_i(t)\,\int^t_{0} \mathcal{I}_X(s)\, u(t,s,\alpha)ds + \mathcal{I}_X(t) \\
	=&\; \mathcal{I}_X(t) - \sum_{i=1}^n r_i(t)\, x_{(1,\ldots,1)}(t).
	\end{split}
	\end{equation} Next, if $\alpha$ has any $a_i>1$, differentiating eq. \eqref{eq:xalpha} and applying Lemma \ref{L:df} yields  
	
	\begin{equation} \begin{split}
	\frac{d}{dt}&{x_\alpha}(t) =\; x_0\, \frac{d}{dt}u(t,0,\alpha) + \int^t_{0} \mathcal{I}_X(s)\; \frac{d}{dt}u(t,s,\alpha) \,ds \\
	&=\; x_0\, \bigg( \sum_{i=1}^{n} r_i(t)u(t,0,\alpha_{i,-1})\,\1{a_i>1}(\alpha) - \sum_{i=1}^{n}r_i(t)u(t,\alpha) \bigg) \; + \\
	&\; \int^t_{0} \mathcal{I}_X(s)\bigg( \sum_{i=1}^{n} r_i(t)\,u(t,s,\alpha_{i,-1})\,\1{a_i>1}(\alpha) - \sum_{i=1}^{n}r_i(t)\,u(t,s,\alpha)\bigg)\,ds  \\
	&=\;  \sum_{i=1}^{n} r_i(t)\,x_{\alpha_{i,-1}}(t)\,\1{a_i>1}(\alpha) - \sum_{i=1}^{n} r_i(t)\,x_{\alpha}(t)
	\end{split}
	\end{equation}
	
	Note that, by the definitions of $x_\alpha$ and $u$ that initial condition $x(0)=x_0$ becomes $x_{(1,\ldots,1)}(0)=x_0$ and $x_\alpha(0)=0$ for the remaining $\alpha\in\mathcal{K}$.
	
	The $y_\ell(t)$ equations \eqref{eq:compIEy} become \eqref{eq:chain3c}, where  $\mathcal{K}_i=\{\alpha\;|\;\alpha\in\mathcal{K},\;a_i=k_i\}$, by substituting eqs. \eqref{eq:compsubj}, eq. \eqref{eq:xalpha},  and  $\mathscr{S}_{\setminus i}(t,\tau)\,h_{r_i}^{k_i}(t,\tau) =  \sum_{\alpha \in \mathcal{K}_i} r_i(t)\,u(t,\tau,\alpha)$, which yields 
	
	\begin{equation} \begin{split}
	& x_0\,\mathscr{S}_{\setminus i}(\tau,0)h_{r_i}^{k_i}(\tau,0) +  \int_{0}^{\tau} \mathcal{I}_X(s) \mathscr{S}_{\setminus i}(\tau,s)h_{r_i}^{k_i}(\tau,s)  ds \; = \\
	& \quad x_0 \sum_{\alpha\in\mathcal{K}_i} r_i(t)\,u(\tau,0,\alpha) +  \int_{0}^{\tau} \mathcal{I}_X(s) \sum_{\alpha\in\mathcal{K}_i} r_i(t)\,u(\tau,s,\alpha)  ds \; =  \\
	& \quad r_i(t) \sum_{\alpha\in\mathcal{K}_i} \bigg(x_0\, u(\tau,0,\alpha) +  \int_{0}^{\tau} \mathcal{I}_X(s) u(\tau,s,\alpha)  ds \bigg) =\;  \sum_{\alpha\in\mathcal{K}_i}r_i(t) \, x_\alpha(\tau).  \\
	\end{split}
	\end{equation} %\qed 
\end{proof}

\begin{figure}[!hbt]
	\begin{minipage}[c]{0.43\textwidth}
		\centerline{\includegraphics[trim=0 5 0 15,clip,width=\textwidth]{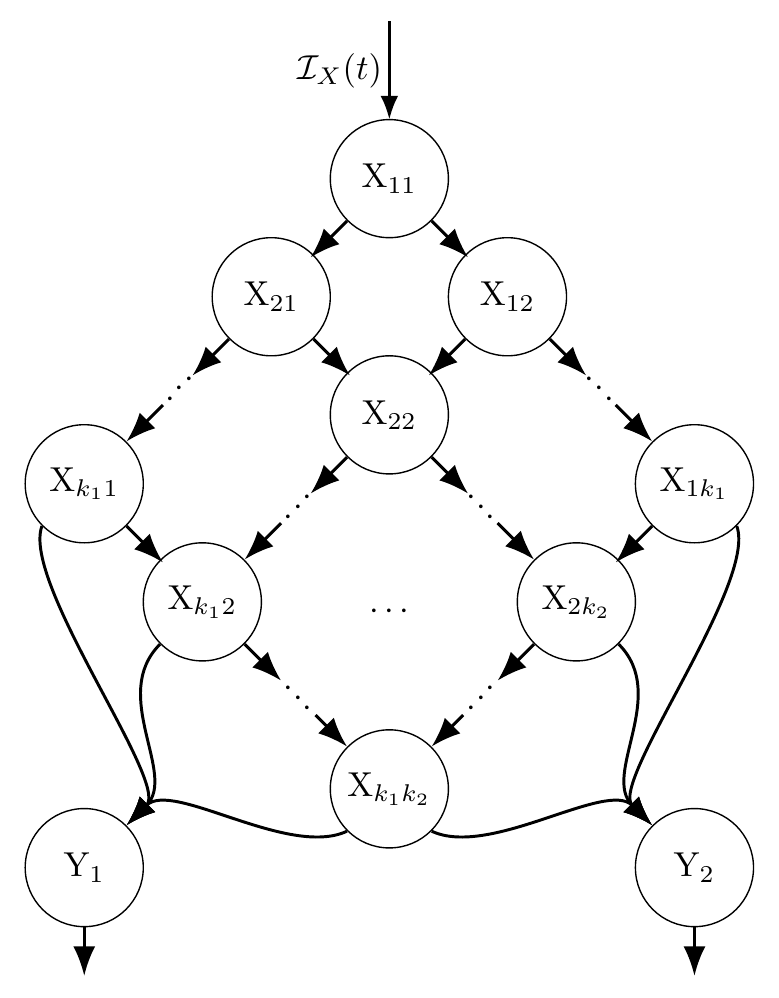}}
	\end{minipage}\hfill
	\begin{minipage}[c]{0.55\textwidth}
		\caption{The sub-state diagram (cf. Fig. \ref{fig:tomulti}) resulting from the application of Theorem \ref{Th:comp} (Extended LCT for dwell times given by competing Poisson processes) to the scenario detailed in Example \ref{ex:multimin}, where the X dwell time distribution is the minimum of two Erlang random variables $T_i\sim$Erlang($r_i,k_i$), $i=1,2$, which can be thought of as event time distribution under two homogeneous Poisson processes as detailed in the main text. We here assume that whichever of these occurs first determines whether particles leaving X transition to Y$_1$ or Y$_2$, respectively. } \label{fig:tomultimin}
	\end{minipage}
	
\end{figure}

\begin{example}\label{ex:multimin} Suppose $T=\min(T_1,T_2)$ where $T_1$ and $T_2$ are the $k_1^\text{th}$ and $k_2^\text{th}$ event time distributions under independent Poisson processes (call these PP1 and PP2) with rates $r_1(t)$ and $r_2(t)$, respectively (see Fig. \ref{fig:tomultimin}). Assume that, upon leaving X, particles transition to Y$_1$ if $T=T_1$ or to Y$_2$ if $T=T_2$. Then by Theorem \ref{Th:comp} above, we can partition X into sub-states defined by which event (under each Poisson process) particles are awaiting. Upon entry into X, all particles enter a sub-state we will denote X$_{1,1}$ where they each await the $1^\text{st}$ events under PP1 or PP2 (recall each particle has its own independent PP1 and PP2 processes governing its transition out of X, and these are \textit{iid} across particles). If the next event to occur for a given particle is from PP1, the particle transitions to X$_{2,1}$ where it awaits either event number 2 from PP1 or event number 1 from PP2 (hence the subscript notation X$_{2,1}$). Likewise, if PP2's first event occurs before PP1's first event, the particle would transition to X$_{1,2}$ where it would await event 1 under PP1, or event 2 under PP2. Particles would leave these two states to either X$_{2,2}$, Y$_1$, or Y$_2$ depending on which event occurs next.  Under these assumptions, and also assuming that $k_1=k_2=2$ and the dwell times in Y$_i$ are exponential with rate $\mu$, then the corresponding mean field equations (using $r(t) = r_1(t)+r_2(t)$) are
	
	\begin{subequations}\label{eq:comp2mf}\begin{align}
		\frac{dx_{11}}{dt} =&\; \mathcal{I}_X(t) - r(t)\,x_{11}(t) \label{eq:comp2mfa}\\
		\frac{dx_{21}}{dt} =&\; r_1(t)\,x_{11}(t) - r(t)\,x_{21}(t) \label{eq:comp2mfb} \\
		\frac{dx_{12}}{dt} =&\; r_2(t)\,x_{11}(t) - r(t)\,x_{12}(t) \label{eq:comp2mfc}\\
		\frac{dx_{22}}{dt} =&\; r_1(t)\,x_{12}(t) + r_2(t)\,x_{21}(t) - r(t)\,x_{12}(t) \label{eq:comp2mfd}\\
		\frac{dy_{1}}{dt} =&\; r_1(t)\,x_{22}(t) - \mu(t)\,y_{1}(t) \label{eq:comp2mfe}\\
		\frac{dy_{2}}{dt} =&\; r_2(t)\,x_{22}(t) - \mu(t)\,y_{2}(t). \label{eq:comp2mff}
		\end{align} \end{subequations}
	
	It's worth pointing out that, in this example, the dwell times for all such sub-states of X are all, in a sense, identically distributed (note the per capita loss rates are all $r(t)$ in eqs. \eqref{eq:comp2mfa}-\eqref{eq:comp2mfd}, and recall the weak memorylessness property of Poisson process $1^\text{st}$ event time distributions discussed in \S \ref{sec:weakmem}). That is, if particles enter one of these sub-states at time $\tau$, it and all other particles in that state at time $\tau$ have a remaining amount of time in that state that follows a 1$^\text{st}$  event time distributions under a Poisson process with rate $r(t)=r_1(t)+r_2(t)$. This is simply a slight generalization of the familiar fact that the minimum of $n$ independent exponentially distributed random variables (with respective rates $r_i$) is itself an exponential random variable (with rate $r\equiv \sum_{i=1}^n r_i$). 
	
	The next section clarifies how this observation about the X sub-state dwell time distributions generalizes to more than two competing Poisson processes, and below (in \S \ref{sec:glct}) we will see how this is a key component of the GLCT. 
	
\end{example}

\subsubsection{Mean field equivalence of proportional outputs \& competing Poisson processes} \label{sec:equivcompmulti}

The scenarios described in \S \ref{sec:tomultibasic} and \S\ref{sec:tomultimin}, which are based on different underlying stochastic assumptions, can lead to equivalent mean field equations when the assumed dwell times all follow $1^\text{st}$ event time distributions. This equivalence is detailed in the following theorem, and is an important aspect of the GLCT detailed in \S \ref{sec:glct}.

\begin{theorem}[Equivalence of proportional outputs \& competing Poisson processes]\label{Th:equivcompmulti} Consider the special case of Theorem \ref{Th:comp} (the Extended LCT for competing Poisson processes) where X has a dwell time given by $T=\min_i T_i$, where each $T_i$ is a Poisson process $1^\text{st}$ event time with rate $r_i(t)$, $i=1,\ldots,n$ and particles transition to Y$_\ell$ with probability $p_{i\ell}(T)$ when $T=T_i$. The corresponding mean field model is equivalent to the special case of Theorem \ref{Th:tomultibasic} (the Extended LCT for multiple outputs) where the X dwell time is a Poisson process 1$^\text{st}$  event time distribution with rate $r(t)=\sum_{i=1}^n r_i(t)$, and the transition probability vector for leaving X and entering state Y$_\ell$ is given by $p_\ell(t)=\sum_{i=1}^n p_{i\ell}(t)\,r_i(t)/r(t)$.
\end{theorem}

\begin{proof} First, in this case $\mathscr{S}(t,\tau) = \prod_{i=0}^{n}\mathcal{S}_{r_i}^{1}(t,\tau)=\mathcal{S}_{r}^{1}(t,\tau)$. Since all $k_i=1$, the probability that $T=T_i$ is $r_i(T)/r(T)$, thus the probability that a particle leaving X at $t$ goes to Y$_\ell$ is $p_\ell(t)=\sum_{i=1}^n \frac{r_i(t)}{r(t)}p_{i\ell}(t)$. Substituting the above equalities into the mean field eqs. \eqref{eq:chain3a} (where there's only one possible index in $\mathcal{K}=\{(1,1,\ldots,1)\}$) and \eqref{eq:chain3c} gives  
	
	\begin{subequations}\begin{align}
		\frac{d}{dt}x(t) =&\; \mathcal{I}_X(t) - r(t)\,x(t)\\
		y_j(t) =&\; y_{j}(0)S_j(t,0)\; +  \int^t_{0} \bigg(\mathcal{I}_j(\tau) + r(t)\,p_j(\tau)\,x(\tau)\bigg) S_j(t,\tau) d\tau 
		\end{align}\end{subequations} which are the mean field equations for the aformentioned special case of Theorem \ref{Th:tomultibasic}. 
	%\qed 
\end{proof}

As we will see in \S\ref{sec:glct}, this equivalence provides some flexibility in simplifying mean field ODEs based on these more complex assumptions about the underlying stochastic state transition models, and allows us to adhere to Poisson process $1^\text{st}$ event time distributions as the building blocks of these generalizations of the LCT.

\subsection{Modeling intermediate state transitions: Reset the clock, or not?}\label{sec:intermediate}

\begin{figure}[tbh]  \begin{minipage}[c]{0.4\textwidth}
		\centerline{\includegraphics[width=\textwidth]{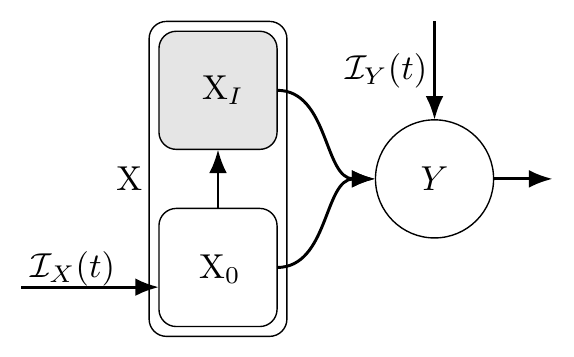}}
	\end{minipage}\hfill
	\begin{minipage}[c]{0.6\textwidth}
		\caption{Should the overall dwell time distribution for state X be ``reset" by the transition from base sub-state X$_0$ to intermediate sub-state X$_I$ (i.e., should the dwell time in state X$_I$ be independent of the time already spent in X$_0$?), or should the X$_I$ dwell time be conditioned on time already spent in X$_0$ so that the X$_0\to$X$_I$ transition does not alter the overall dwell time in state X? How do these different assumptions alter the structure of the corresponding mean field ODEs? We answer these question in \S \ref{sec:intermediate} where we describe how to apply the LCT in scenarios with intermediate states, assuming in \S\ref{sec:intermediatereset} that the dwell time distribution for X$_I$ is independent of the amount of time spent in X$_0$, and assuming in \S\ref{sec:intermediateneutral} that the overall dwell time for X is unaffected by transitions from X$_0$ to X$_I$.}\label{fig:intermediate}
	\end{minipage}
\end{figure}

In this section, we discuss how to apply extensions of the Linear Chain Trick in two similar but distinctly different scenarios where the transition to one or more intermediate sub-states either resets an individual's overall dwell time in state X by assuming the time spent in an intermediate sub-state X$_{I_i}$ is independent of time already spent in X$_0$ (see \S \ref{sec:intermediatereset}), or instead leaves the overall dwell time distribution for X unchanged by conditioning the time spent in intermediate state X$_{I_i}$ is conditioned on time already spent in X$_0$ (see \S \ref{sec:intermediateneutral} and Fig. \ref{fig:intermediate}). 

To illustrate these two cases considered below, consider the simple case illustrated in Fig. \ref{fig:intermediate} where a single intermediate sub-state $X_I$ is being modeled, and particles enter X into sub-state X$_0$ at rate $\mathcal{I}_X(t)$. Let X$=$X$_0\cup $X$_I$. Assume particles subsequently transition out of X$_0$ either to sub-state X$_I$ or they leave state X directly and enter state Y. Assume the distribution of time spent in X$_0$ (in both scenarios) is $T_*=$min($T_0,T_1$) where particles transition to X$_I$ if $T_1<T_0$ (i.e., if $T=T_1$) or to Y if $T_0<T_1$ (where each $T_i$ is the $k_i^\text{th}$ event time under Poisson processes with rates $r_0(t)$ and $r_1(t)$ (see \S\ref{sec:tomultimin} and \S\ref{sec:equivcompmulti}). The distribution of time spent in intermediate state X$_I$, which we'll denote as $T_I$, can either be assumed to be independent of time spent in X$_0$ (i.e., the transition to X$_I$ `resets the clock'; see \S\ref{sec:intermediatereset}) or in the second scenario it is conditional on time already spent in X$_0$, $T_*$, such that the total amount of time spent in X, $T_*+T_I$, is equivalent in distribution to $T_0$ (i.e., the transition to X$_I$ does not change the overall distribution of time spent in X; see \S\ref{sec:intermediateneutral}).

An example of these different assumptions leading to important differences in practice comes from \citet{Feng2016} where individuals infected with Ebola can either leave the infected state (X) directly (either to a recovery or death), or after first transitioning to an intermediate hospitalized state (X$_I$) which needs to be explicitly modeled in order to incorporate a quarantine effect into the rates of disease transmission (i.e., the force of infection should depend on the number of non-quarantined individuals, i.e., X$_0$). As shown in \citet{Feng2016}, the epidemic model output depends strongly upon whether or not it is assumed that moving into the hospitalized sub-state impacts the distribution of time spent in the infected state X. 

In the next two sections, we provide extensions of the LCT that detail the structure of mean field ODEs corresponding to the generalization of these two scenarios, extended to multiple possible intermediate states reached following the outcome of multiple competing Poisson processes, and multiple recipient states.

\subsubsection{Intermediate states that reset dwell time distributions}\label{sec:intermediatereset}

\begin{figure}[htb]
	\centerline{\includegraphics[width=.85\textwidth]{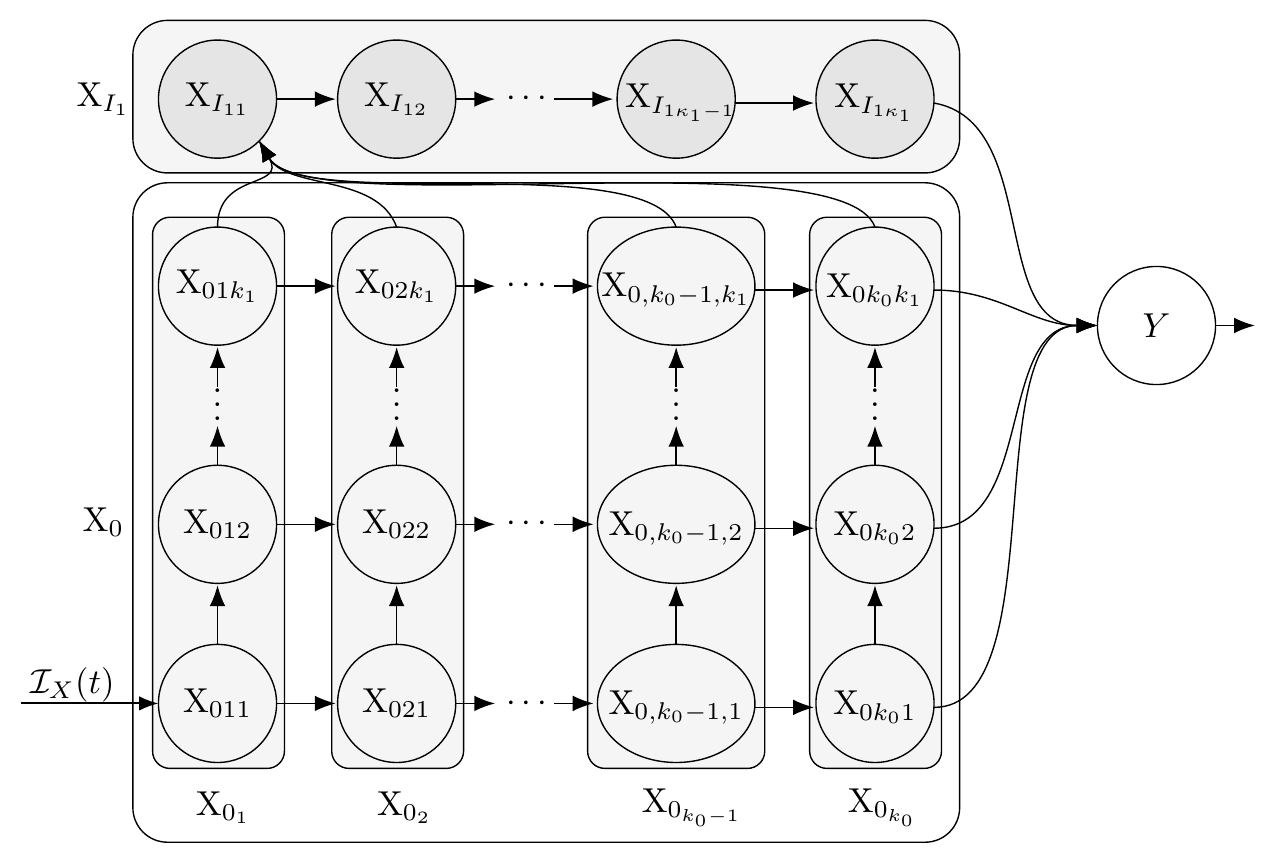}}
	\caption{The linear chain trick (LCT) extension given in Theorem \ref{Th:resettime} applied to the first scenario discussed in Figure \ref{fig:intermediate} in which the dwell time in intermediate state X$_I$ is independent of time already spent in X$_0$, causes the sub-state transitions within X$=$X$_0\cup$X$_{I_1}$ to alter the overall dwell time in state X. Here the dwell time distribution for X$_0$ is the minimum of two independent Erlang distributions. It is assumed that after transitioning to sub-state X$_{I_1}$ the remaining time spent in state X (i.e., the dwell time in state X$_{I_1}$) is Erlang($\varrho_1,\kappa_1$), i.e., independent of time already spent in X$_0$. Sub-states within X$_0$ are the cohorts of particles waiting to advance towards a transition to Y (advance right as events occur under the base Poisson process) or towards X$_{I_1}$ (advance up as events occur along the second Poisson process) as discussed in section \ref{sec:tomultimin}. Sub-states within X$_{I_1}$ represent the usual linear chain of $\kappa_1$ sub-states of X$_{I_1}$, with particles spending an exponentially distributed (rate $\varrho_1$) amount of time in each.  For a detailed treatment of the more general case, see \S\ref{sec:intermediatereset} and Theorem \ref{Th:resettime}. } \label{fig:intermediate1}
\end{figure}

First, we consider the case in which the time spent in the intermediate state X$_I$ is independent of the time already spent in X (i.e., in the base state X$_0$). Note this is arguably the more commonly encountered (implicit) assumption found in ODE models that aren't explicitly derived from a stochastic model and/or mean field integro-differential delay equations. 

The construction of mean field ODEs for this case is a straightforward application of Theorem \ref{Th:comp} from the previous section, combined with the extended LCT with output to multiple states (Theorem \ref{Th:tomultibasic}), as detailed in the following theorem. Here we have extended this scenario to include $M_X$ intermediate sub-states X$_{I_j}$ where the transition to those sub-states from base state X$_0$ is based on the outcome of $N$ competing Poisson process event time distributions ($T_i$), and upon leaving the intermediate states particles transition out of state X into one of $M_Y$ possible recipient states Y$_\ell$.

\begin{theorem}[Extended LCT with dwell time altering intermediate sub-state transitions]\label{Th:resettime} Suppose particles enter X at rate $\mathcal{I}_X(t)$ into a base sub-state X$_0$. Assume particles remain in X$_0$ according to a dwell time distribution given by $T$, the minimum of $N+1$ independent Poisson process $k_i^\text{th}$ event time distributions with rates $r_i(t)$, $i=0,\ldots,N$ (i.e., $T=\min_i(T_i)$). Particles leaving X$_0$ transition to one of $M_X\geq1$ intermediate sub-states X$_{I_i}$ or to one of $M_Y\geq1$ recipient states $Y_\ell$ according to which $T_i=T$.  If $T_0=T$ then the particle leaves X and the probability of transitioning to Y$_\ell$ is $p_{0\ell}(T)$, where $\sum_{\ell=1}^{M_Y} p_{0\ell}(T)=1$. If $T_i=T$ for $i\geq1$ then the particle transitions to X$_{I_j}$ with probability $p_{ij}(T)$, where $\sum_{j=1}^{M_X} p_{ij}(T)=1$. Particles in intermediate state $X_{I_j}$ remain there according to the $\kappa_i^\text{th}$ event times under a Poisson process with rate $\varrho_i(t)$, and then transition to state Y$_\ell$ with probability $q_{j\ell}(t)$, where (for fixed $t$) $\sum_{\ell=1}^{M_Y} q_{j\ell}(t)=1$, and they remain in Y$_\ell$ according to a dwell time with survival function $S_\ell(t,\tau)$. 
	
	In this case the corresponding mean field equations are  
	
	\begin{subequations} \label{eq:reset} \begin{align} 
		\frac{d}{dt}x_{0(1,\ldots,1)}(t) =&\; \mathcal{I}_X(t) - \sum_{i=0}^N r_i(t) \,x_{(1,\ldots,1)}(t) \\ 
		\frac{d}{dt}x_{0\alpha}(t) =&\; \sum_{i=0}^{N} r_i(t) \bigg( x_{0\alpha_{i,-1}}(t)\,\1{a_i>1}(\alpha) - x_{0\alpha}(t)\bigg)   \\
		\frac{d}{dt}x_{I_{j1}}(t)  =&\; \mathcal{I}_{X_{Ij}}(t) + p_{ij}(t)\bigg(\sum_{\alpha\in\mathcal{K}_i} r_i(t)\,x_{0\alpha}(t)\bigg) - \varrho_j(t)\,x_{I_{j1}}(t)  \\
		\frac{d}{dt}x_{I_{jk}}(t)  =&\; \varrho_j(t)\big(x_{I_{j,k-1}}(t) - x_{I_{jk}}(t)\big),  \;\qquad  k=2,\ldots,\kappa_j  \\
		\begin{split}
		y_\ell(t) =&\; y_{\ell}(0)\,S_\ell(t,0) + \int^t_{0} \bigg(\mathcal{I}_{Y_\ell}(\tau) + p_{0\ell}(\tau) \sum_{\alpha\in\mathcal{K}_0} r_0(\tau)\,x_{0\alpha}(\tau) \\
		& \qquad \qquad \quad \; \; + \sum_{j=1}^{M_X} \varrho_j(\tau)\,x_{I_{j\kappa_j}}(\tau)\, q_{j\ell}(\tau) \bigg) S_Y(t,\tau)\,d\tau. \label{eq:resety}
		\end{split}
		\end{align} \end{subequations} where $\mathcal{K}=\{(a_0,a_1,\ldots,a_N)\;|\;a_j\in\{1,\ldots,k_j\}\}$, $\alpha=(a_0,\ldots,a_N) \in \mathcal{K}\setminus(1,\ldots,1)$, $j=1,\ldots,N$, $\ell = 1,\ldots,M_Y$, the amount in base sub-state X$_0$ is $x_0(t)=\sum_{\alpha\in\mathcal{K}}x_{0\alpha}(t)$, and the amount in the $j^\text{th}$ intermediate state X$_{I_j}$ is $x_{Ij}(t)=\sum_{k=1}^{\kappa_j} x_{I_{jk}}(t)$ (see Theorem \ref{Th:comp} for notation).  Note that the $y(t)$ equation \eqref{eq:resety} may be further reduced to a system of ODEs, e.g, via Corollary \ref{Th:baseCorr}, and that more complicated distributions for dwell times in intermediate states X$_{I_i}$ (e.g., an Erlang mixture) could be similarly modeled according to other cases addressed in this manuscript. \\
\end{theorem}

\begin{proof} This result follows from applying Theorem \ref{Th:comp} to X$_0$ and treating the intermediate states X$_{I_j}$ as recipient states, then applying Theorem \ref{Th:tomultibasic} to each intermediate state to partition each X$_{I_j}$ into X$_{I_{jk}}$, $k=1,\ldots,\kappa_j$, yielding eqs. \eqref{eq:reset}. 	
	%\qed 
\end{proof} 

\begin{example}\label{ex:ressettime} To illustrate the application of Theorem \ref{Th:resettime}, consider the case in Fig. \ref{fig:intermediate} but with 1 intermediate state (i.e., $N=1$), with $T_0\sim$Erlang($r_0,k_0=2$), $T_1\sim$Erlang($\varrho_1,k_2=2$), $T_{I_1}\sim$Erlang($\varrho_1,\kappa_1=3$) and an exponential (rate $\mu$) dwell time in Y. Also assume the only inputs into X are into X$_0$ at rate $\mathcal{I}_X(t)$. Then the corresponding mean field ODEs are given by eqs. \eqref{eq:resettimeex} below, where  $x_0(t)=x_{0(1,1)}(t)+x_{0(2,1)}(t)+x_{0(1,2)}(t)+x_{0(2,2)}(t)$ and $x_{I_1}(t)=x_{I_{11}}(t)+x_{I_{12}}(t)+x_{I_{13}}(t)$.
	
	\begin{subequations} \label{eq:resettimeex} \begin{align} 
		\frac{d}{dt}x_{0(1,1)}(t) =&\; \mathcal{I}_X(t) - (r_0+r_1)\,x_{0(1,1)}(t)  \label{eq:resettimeexa} \\ 
		\frac{d}{dt}x_{0(2,1)}(t) =&\; r_0 x_{0(1,1)}(t) - (r_0+r_1) x_{0(2,1)}(t)  \label{eq:resettimeexb} \\ 
		\frac{d}{dt}x_{0(1,2)}(t) =&\; r_1 x_{0(1,1)}(t) - (r_0+r_1) x_{0(1,2)}(t)  \label{eq:resettimeexc} \\ 
		\frac{d}{dt}x_{0(2,2)}(t) =&\; r_0 x_{0(1,2)}(t) + r_1 x_{0(2,1)}(t) - (r_0+r_1) x_{0(2,2)}(t)  \label{eq:resettimeexd} \\ 
		\frac{d}{dt}x_{I_{11}}(t)  =&\; r_1\,x_{0(1,2)}(t) + r_1\,x_{0(2,2)}(t) - \varrho\,x_{I_{11}}(t) \label{eq:resettimeexe} \\ 
		\frac{d}{dt}x_{I_{12}}(t)  =&\;\varrho\,x_{I_{11}}(t) - \varrho\,x_{I_{12}}(t) \label{eq:resettimeexf} \\ 
		\frac{d}{dt}x_{I_{13}}(t)  =&\;\varrho\,x_{I_{12}}(t) - \varrho\,x_{I_{13}}(t) \label{eq:resettimeexg} \\ 
		\begin{split}\label{eq:resettimeexh}\frac{d}{dt}y(t) =&\; \mathcal{I}_Y(t) + r_0\,x_{0(2,1)}(t) + r_1\,x_{0(1,2)}(t)\\ & \qquad \;\;\, + r_0\,x_{0(2,2)}(t) + \varrho\,x_{I_{13}}(t) - r\,y(t).  \end{split}
		\end{align} \end{subequations} 
	
\end{example}

In the next section, we show how one can modify eqs. \eqref{eq:resettimeex} above to implement an alternative assumption: that the overall dwell time in state X is independent of any transitions to intermediate sub-states X$_{I_i}$, which is achieved by conditioning the intermediate sub-state dwell times on time already spent in X$_0$. 

%\clearpage
\subsubsection{Intermediate states that preserve dwell time distributions}\label{sec:intermediateneutral}

In this section we address how to construct mean field ODE models that incorporate `dwell time neutral' sub-state transitions, i.e., where the distribution of time spent in X is the same regardless of whether or not particles transition (within X) from some base sub-state X$_0$ to one or more intermediate sub-states X$_{I_j}$. This is done by conditioning the dwell time distributions in X$_{I_i}$ on time spent in X$_0$ in a way that leverages the weak memorylessness property discussed in \S\ref{sec:weakmem}. 

In applications, this case (in contrast to the previous case) is perhaps the more commonly desired assumption, since modelers often seek to partition states into sub-states where key characteristics (e.g., the overall dwell time distribution) remain unchanged, but where the different sub-states have functional differences elsewhere in the model. For example, consider an SIR type infectious disease model in which a goal is to incorporate reduced disease transmission from quarantined individuals, but where (in the absence of effective treatment) the transition to the quarantined state does not alter the overall distribution of the infectious period duration.

One approach to deriving such a model is to condition the dwell time distribution for an intermediate state X$_{I_i}$ on the time already spent in X$_0$ (as in \citet{Feng2016}).  We take a slightly different approach and exploit the weak memoryless property of Poisson process 1$^\text{st}$  event time distributions (see Theorem \ref{Th:wm} in \S \ref{sec:weakmem}, and the notation used in the previous section) to instead condition the dwell time distribution for intermediate states X$_{I_j}$ on how many of the $k_0$ events have already occurred when a particle transitions from X$_0$ to X$_{I_j}$ (rather than conditioning on the exact elapsed time spent in X$_0$). In this case, since each sub-state of X$_0$ has \textit{iid} dwell time distributions that are Poisson process 1$^\text{st}$  event times with rate $r(t)=\sum_{i=0}^N r_i(t)$, if $i$ of the $k_0$ events had occurred prior to the transition out of X$_0$, then the weak memoryless property of Poisson process 1$^\text{st}$  event time distributions implies that the remaining time spent in X$_{I_j}$ should follow a $(k_0-i)^\text{th}$ event time distribution under an Poisson process with rate $r_0(t)$, thus ensuring that the total time spent in X follows a $k_0^\text{th}$ event time distribution with rate $r_0(t)$. With this realization in hand, one can then apply Theorem \ref{Th:comp} and Theorem \ref{Th:tomultibasic} as in the previous section to obtain the desired mean field ODEs, as detailed in the following Theorem, and as illustrated in Fig. \ref{fig:intermediate2}.

\begin{figure}
	\centerline{\includegraphics[width=0.85\textwidth]{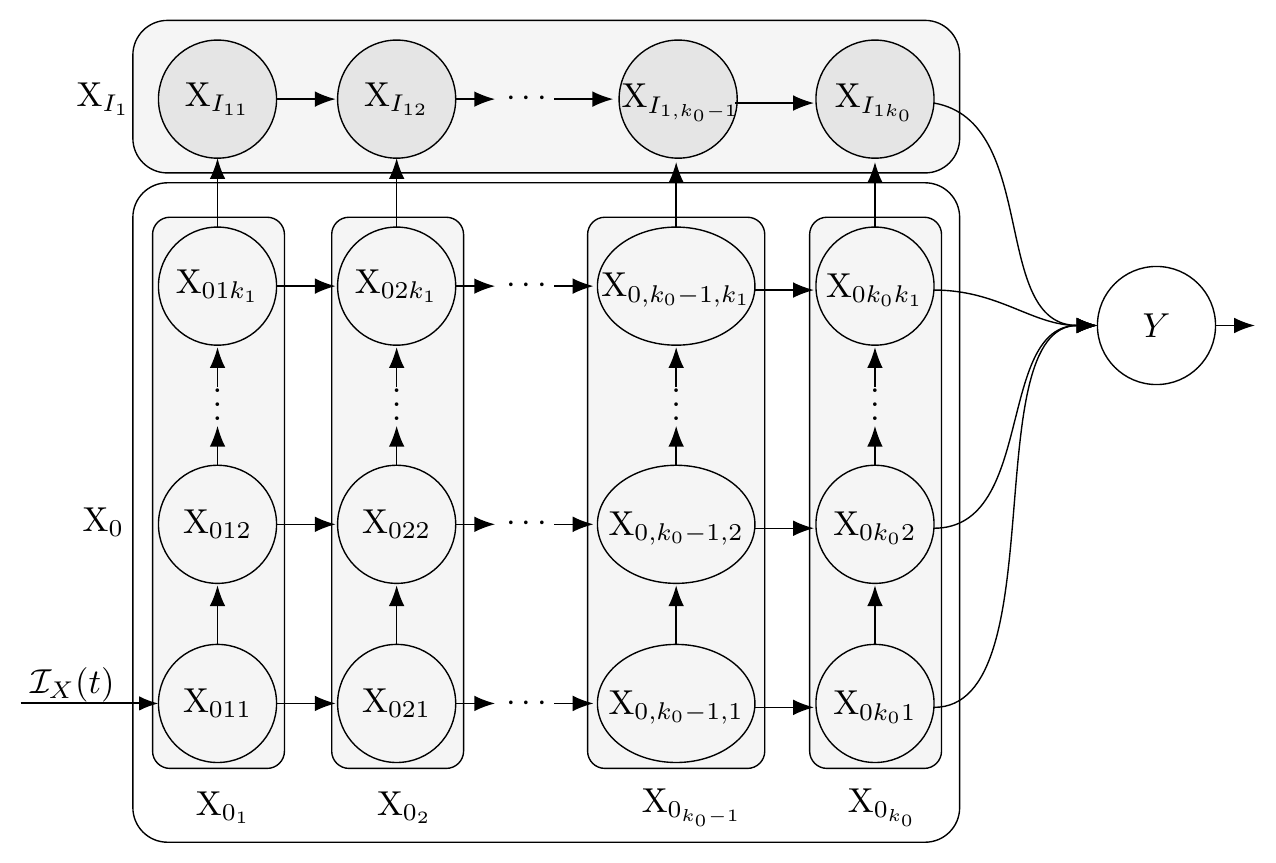}}
	\caption{In contrast to Fig. \ref{fig:intermediate1} and Theorem \ref{Th:resettime}, this example illustrates an application of Theorem \ref{Th:timepreserving} (Extended LCT with dwell time preserving intermediate states) which assumes an Erlang distributed dwell time X with a single (dwell time neutral) intermediate sub-state X$_I$. Compare the transitions from X$_0$ to X$_{I_1}$ in this sub-state diagram to the corresponding transitions out of X$_0$ in Fig. \ref{fig:intermediate1}. In this case, the overall dwell time in X is determined by $T_0\sim$Erlang($r_0,k_0$), and is independent of whether or not a transition within X (from X$_0$ to X$_{I_1}$) occurs as discussed in the main text in \S\ref{sec:intermediate}. Note that we have graphically arranged these sub-states as in Fig. \ref{fig:intermediate1}, so that events under the Poisson process that defines $T_0$ drive transitions to the right, and events under the Poisson process that defines $T_1$ drive vertical transitions.  }\label{fig:intermediate2}
\end{figure}

\begin{theorem}[Extended LCT with dwell time preserving intermediate states]\label{Th:timepreserving} Consider the mean field equations for a system of particles entering state X (into sub-state X$_0$) at rate $\mathcal{I}_X(t)$. As in the previous case, assume the time spent in X$_0$ follows the minimum of $N+1$ independent Poisson process $k_i^\text{th}$ event time distributions with respective rates $r_i(t)$, $i=0,\ldots,N$ (i.e., $T=\min_i(T_i)$). Particles leaving X$_0$ at time $T$ transition to recipient state Y$_\ell$ with probability $p_{0\ell}(T)$ if $T=T_0$, or if $T=T_i$ ($i=1,\ldots,N$) to the $j^\text{th}$ of $M_X$ intermediate sub-states, X$_{I_j}$, with probability $p_{ij}(T)$. If $T<T_0$, we may define a random variable $K\in\{0,\ldots,k_0-1\}$ indicating how many events had occurred under the Poisson process associated with $T_0$ at the time of the transition out of X$_0$ (at time $T$). In order to ensure the overall time spent in X follows a Poisson process $k_0^\text{th}$ event time distribution with rate $r_0(t)$, it follows that particles entering state, X$_{I_j}$ will remain there for a duration of time that is conditioned on $K=k$ such that the conditional dwell time for that particle in X$_{I_j}$ will be given by a Poisson process $(k_0-k)^\text{th}$ event time with rate $r_0(t)$. Finally, assume that particles leaving X via intermediate sub-state X$_{I_j}$ at time $t$ transition to Y$_\ell$ with probability $q_{j\ell}$, where they remain according to a dwell time determined by survival function $S_\ell(t,\tau)$.
	
	The corresponding mean field equations are  
	
	\begin{subequations}\label{eq:nintermediate}\begin{align} 
		&\frac{d}{dt}x_{0(1,\ldots,1)}(t) =\; \mathcal{I}_X(t) -\sum_{i=0}^{N}r_i(t)\,x_{0(1,\ldots,1)}(t)   \label{eq:nintermediatea} \\
		&\frac{d}{dt}x_{0\alpha}(t) =\; \sum_{i=0}^{N} r_i(t)\, x_{0\alpha_{i,-1}}(t)\,\1{a_i>1}(\alpha) - \sum_{i=0}^{N} r_i(t)\, x_{0\alpha}(t) \label{eq:nintermediateb} \\
		&\frac{d}{dt}x_{I_{jk}}(t) =\; r_0(t)\,\big( x_{I_{j,k-1}}(t)\,\1{k>1} - x_{I_{jk}}(t) \big) + \sum_{\alpha\in\mathcal{K}_{ij}} r_i(t)\,x_\alpha(t)\,p_{ij}(t)  \\
		\begin{split}	&y_\ell(t) =\; y_\ell(0)S_\ell(t,0) +  \int^t_{0} \bigg(\mathcal{I}_{Y_\ell}(\tau) + \sum_{\alpha\in\mathcal{K}_0} r_0(\tau)\,x_\alpha(\tau) \\
		& \qquad \qquad \qquad \qquad \quad + \sum_{j=1}^{M_X} r_0(\tau)\,x_{I_{jk_0}}(\tau)\,q_{j\ell}(\tau) \bigg) S_\ell(t,\tau)d\tau \end{split}  \label{eq:nintermediated}
		\end{align}\end{subequations}
	where $\mathcal{K}=\{(a_0,a_1,\ldots,a_N)\;|\;a_j\in\{1,\ldots,k_j\}\}$, $\alpha \in\mathcal{K}\setminus(1,\ldots,1)$, $j=1,\ldots,M_X$, $k=1,\ldots,k_0$, $\ell=1,\ldots,M_Y$, $\mathcal{K}_i\subset\mathcal{K}$ are the subset of indices where $a_i=k_i$, $\mathcal{K}_{ij}\subset\mathcal{K}_i$ are the subset of indices where $a_i=k_i$ and $a_0=j$, $x_0(t)=\sum_{\alpha\in\mathcal{K}} x_{0\alpha}(t)$, $x_{iI}(t)=\sum_{j=1}^{k_0} x_{iIj}(t)$, and $x(t)=x_0(t)+\sum_{i=1}^{n} x_{iI}(t)$. The $y_\ell(t)$ equations \eqref{eq:nintermediated} may be further reduced to a system of ODEs, e.g., via Corollary \ref{Th:baseCorr}.
	
\end{theorem}

\begin{proof}  The proof of Theorem \ref{Th:timepreserving} parallels the proof of Theorem \ref{Th:resettime}, but with the following modifications. First, each sub-state of X$_{I_j}$ (for all $j$) has the same dwell time distribution, namely, they are all $1^\text{st}$ event time distributions under a Poisson process with rate $r_0(t)$. Second, upon leaving X$_0$ where $T=T_i$ and $K(T)=k < k_0$ (i.e., when only $k<k_0$ events have occurred under the $0^\text{th}$ Poisson process; see the definition of $K$ in the text above) particles will enter (with probability $p_{ij}(T)$) the $j^\text{th}$ intermediate state X$_{I_j}$ by entering sub-state X$_{I_{jk}}$ which (due to the weak memorylessness property described in Theorem \ref{Th:wm}) ensures that, upon leaving X$_{I_j}$ particles will have spent a duration of time that follows the Poisson process $k_0^\text{th}$ event time distribution with rate $r_0(t)$.  
	%\qed 
\end{proof}

\begin{example}\label{ex:timepreserving} Consider Example \ref{ex:ressettime} in the previous section, but now instead assume that the transition to the intermediate state does not impact the overall time spent in state X as detailed above.  Then by Theorem \ref{Th:timepreserving} the corresponding mean field ODEs are given by eqs. \eqref{eq:timepreservingex} below (compare eqs. \eqref{eq:timepreservingexe}-\eqref{eq:timepreservingexg} to eqs. \eqref{eq:resettimeexe}-\eqref{eq:resettimeexh}).\\ 
	
	\begin{subequations} \label{eq:timepreservingex} \begin{align} 
		\frac{d}{dt}x_{0(1,1)}(t) =&\; \mathcal{I}_X(t) - (r_0+r_1)\,x_{0(1,1)}(t) \label{eq:timepreservingexa} \\ 
		\frac{d}{dt}x_{0(2,1)}(t) =&\; r_0 x_{0(1,1)}(t) - (r_0+r_1) x_{0(2,1)}(t) \label{eq:timepreservingexb} \\
		\frac{d}{dt}x_{0(1,2)}(t) =&\; r_1 x_{0(1,1)}(t) - (r_0+r_1) x_{0(1,2)}(t) \label{eq:timepreservingexc} \\
		\frac{d}{dt}x_{0(2,2)}(t) =&\; r_0 x_{0(1,2)}(t) + r_1 x_{0(2,1)}(t) - (r_0+r_1) x_{0(2,2)}(t) \label{eq:timepreservingexd} \\
		\frac{d}{dt}x_{I_{11}}(t)  =&\; r_1\,x_{0(1,2)}(t) - r_0\,x_{I_{11}}(t) \label{eq:timepreservingexe} \\
		\frac{d}{dt}x_{I_{12}}(t)  =&\; r_1\,x_{0(2,2)}(t) + r_0\,x_{I_{11}}(t) - r_0\,x_{I_{12}}(t) \label{eq:timepreservingexf} \\
		\begin{split}\label{eq:timepreservingexg} \frac{d}{dt}y(t) =&\; \mathcal{I}_Y(t) + r_0\,x_{0(2,1)}(t) + r_1\,x_{0(1,2)}(t) \\ & \qquad \; \; \, + r_0\,x_{0(2,2)}(t) + r_0\,x_{I_{12}}(t) - r\,y(t). \end{split} 
		\end{align} \end{subequations}

\end{example}

\subsection{Generalized Linear Chain Trick (GLCT)}\label{sec:glct}

In the preceding sections we have provided various extensions of the Linear Chain Trick (LCT) that describe how the structure of mean field ODE models reflects the assumptions that define corresponding continuous time stochastic state transition models. Each case above can be viewed as a special case of the following more general framework for constructing mean field ODEs, which we refer to as the Generalized Linear Chain Trick (GLCT). 

The cases we have addressed thus far share the following stochastic model assumptions, which constitute the major assumptions of the GLCT stated in Theorem \ref{Th:glct} below: \begin{enumerate}
	\item[A1.] A focal state (which we call state X) can be partitioned into a finite number of sub-states (e.g, X$_1,\ldots,$X$_n$), each with independent (across states and particles) dwell time distributions that are either exponentially distributed with rates $r_i$ or, more generally, are distributed as independent $1^\text{st}$ event times under nonhomogeneous Poisson processes with rates $r_i(t)$, $i=1,\ldots,n$. Recall the equivalence relation in \S \ref{sec:equivcompmulti}.
	\item[A2.] Inflow rates into the focal state can be described by non-negative, integrable inflow rates into each of these sub-states (e.g., $\mathcal{I}_{X_1}(t),\ldots,\mathcal{I}_{X_n}(t)$), some or all of which may be zero. This includes a single inflow rate $\mathcal{I}_X(t)$ and a vector of probabilities/proportions $\rho(t) = [\rho_1(t),\ldots,\rho_n(t)]^\text{\textit{T}}$ describing how incoming particles are distributed across sub-states X$_i$ (i.e., we let $\mathcal{I}_{X_i}(t)\equiv \rho_i(t)\,\mathcal{I}_X(t)$).
	\item[A3.] Particles that transition out of a sub-state X$_i$ at time $t$ transition into either a different sub-state X$_j$ with probability $p_{ij}(t)$, or enter one of a finite number of recipient states Y$_\ell$, $\ell=1,\ldots,m$, with probability $p_{i,n+\ell}$. That is, let $p_{ij}(t)$ denote the probability that a particle leaving state X$_i$ at time $t$ enters either X$_j$ if $j\leq n$ or Y$_{j-n}$ if $j>n$, where $i=1,\ldots,n$, $j=1,\ldots,n,n+1,\ldots,n+m$. 
	\item[A4.] Recipient states Y$_\ell$, $\ell=1,\ldots,m$, also have dwell time distributions defined by survival functions $S_{Y_\ell}(t,\tau)$ and integrable, non-negative inflow rates $\mathcal{I}_{Y_\ell}(t)$ that describe inputs from all other non-X sources.
\end{enumerate}

The GLCT (Theorem \ref{Th:glct}) below describes how to construct mean field ODEs for the broad class of state transition models that satisfy the above assumptions.%\footnote{What if any implicit assumptions are being made here? These should be made explicit!}

\begin{theorem}[Generalized Linear Chain Trick]\label{Th:glct} Consider a stochastic, continuous time state transition model of particles entering state X and transitioning to states Y$_\ell$, $\ell=1,\ldots,m$, according to the above assumptions A1-A4. Then the corresponding mean field model is given by the following system of equations. 
	
	\begin{subequations}\label{eq:glct}\begin{align}
		\frac{d}{dt}x_i(t) =&\; \mathcal{I}_{X_i}(t) + \sum_{j=1}^n p_{ji}(t)\,r_j(t)\,x_j(t) - r_i(t)\,x_i(t), \quad i=1,\ldots,n, \label{eq:glctx} \\
		\begin{split} \label{eq:glcty} y_\ell(t) =&\; y_\ell(0)S_{Y_\ell}(t,0) + \int_{0}^{t} \bigg ( \mathcal{I}_{Y_\ell}(\tau) + \\ & \qquad \qquad \qquad \qquad \quad \sum_{j=1}^{n} r_j(t)\,x_j(\tau)\,p_{j,n+\ell}(t) \bigg ) S_{Y_\ell}(t,\tau)\,d\tau \end{split} 
		\end{align} \end{subequations} where $x(t)=\sum_{i=1}^n x_i(t)$, and we assume non-negative initial conditions $x_i(0)=x_{i0}$, $y_\ell(0)=y_{\ell0}$. Note that the $y_\ell(t)$ equations might be reducible to ODEs, e.g., via Corollary \ref{Th:baseCorr} or other results presented above.
	
	Furthermore, eqs. \eqref{eq:glctx} may be written in vector form where $P_X(t)=(p_{ij}(t))$ ($i,j\in\{1,\ldots,n\}$) is the $n\times n$ matrix of (potentially time-varying) probabilities describing which transitions out of X$_i$ at time $t$ go to X$_j$ (likewise, one can define $P_Y(t)=(p_{ij}(t))$, $i\in\{1,\ldots,n\}$, $j\in\{n+1,\ldots,n+m\}$, which is the $n\times m$ matrix of probabilities describing which transitions from X$_i$ at time $t$ go to Y$_{j-n}$), $\mathbf{\mathcal{I}_X}(t) = [\mathcal{I}_{X_1}, \ldots, \mathcal{I}_{X_n}]^\text{T}$,  $R(t)=[r_1(t),\ldots,r_n(t)]^\text{T}$, and $\mathbf{x}(t)=[x_1(t),\ldots,x_n(t)]^\text{T}$ which yields 
	
	\begin{equation}
	\frac{d}{dt}\mathbf{x}(t) =\; \mathbf{\mathcal{I}_X}(t) + P_X(t)^\text{T}\,(R(t)\circ\mathbf{x}(t)) - R(t)\circ\mathbf{x}(t).\label{eq:GLCTX}
	\end{equation} where $\circ$ indicates the Hadamard (element-wise) product.% and $^\text{T}$ denotes matrix transposition.
\end{theorem}
\begin{proof} The proof of the theorem above follows directly from applying Theorem \ref{Th:tomultibasic} to each sub-state.
	%\qed 
\end{proof}

\begin{corollary}[LCT for phase-type distributions] If $R(t)=R$, $P_X(t)=P_X$, and $P_Y(t)=P_Y$ are all constant, then the X dwell time distribution follows the hitting time distribution for a Continuous Time Markov Chain (CTMC) with absorbing states Y$_\ell$ and an ($n+m$)$\times$($n+m$) transition probability matrix 
	
	\begin{equation}
	P=\begin{bmatrix} P_X & P_Y\\ 0 & \mathbf{I} \\ \end{bmatrix}.
	\end{equation} These CTMC hitting time distributions include the hypoexponential distribution, hyper-exponential and hyper-Erlang distributions, generalized Coxian distribution, and other continuous phase-type distributions \citep{Reinecke2012a,Horvath2016}.
\end{corollary}

\begin{figure}[!htb]  \begin{minipage}[c]{0.45\textwidth}
		\centerline{\includegraphics[width=\textwidth]{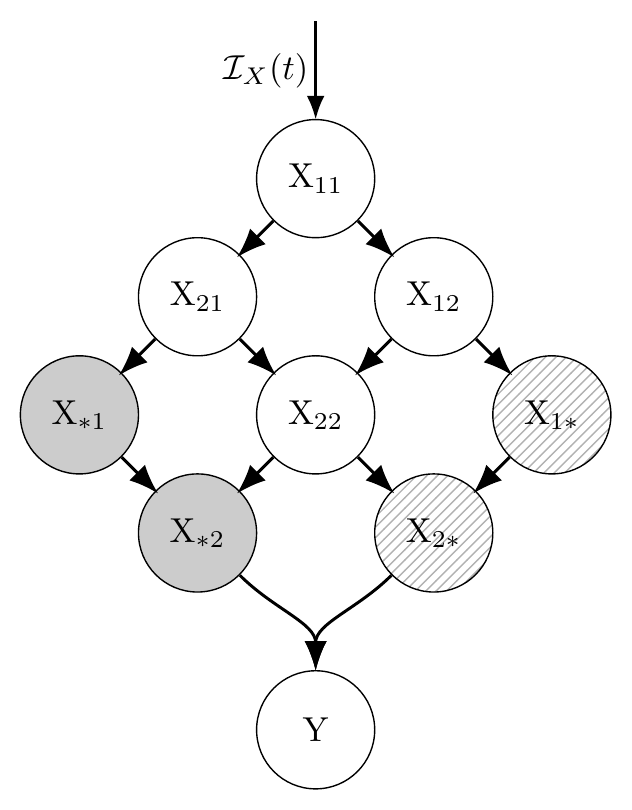}}
	\end{minipage}\hfill
	\begin{minipage}[c]{0.55\textwidth}
		\caption{The X sub-state structure for Example \ref{ex:tomultimax} 
			where X dwell time distribution follows the maximum of two Erlang random variables with rates $r_1$ and $r_2$, respectively, and shape parameters $k_i=2$. As in \S\ref{sec:tomultimin}, where the minimum is assumed instead of the maximum, X can be partitioned using indices based on organizing particles by which events they are awaiting under each Poisson process associated with each Erlang distribution. Upon reaching the target event (here, the 2$^\text{nd}$ event) under any given Poisson process, particles transition to sub-states with an asterisk in the corresponding index position (e.g., see figure). In general, these sub-states all have a dwell times given by the $1^\text{st}$ even time under a Poisson process, but with differing rates (see Example \ref{ex:tomultimax}): here they follow exponential distributions with either rate $r=r_1+r_2$ (white backgrounds), rate $r_2$ (gray backgrounds), or rate $r_1$ (lined backgrounds). \label{fig:tomultimax} } \end{minipage} 
\end{figure}

\begin{example}[Serial LCT \& hypoexponential distributions] Assume the dwell time in state X is given by the sum of independent (not identically distributed) Erlang distributions or, more generally, Poisson process $k_i^\text{th}$ event time distributions with rates $r_i(t)$, i.e., $T=\sum_i T_i$, $i=1,\ldots,N$ (note the special case where all $k_i=1$ and $r_i(t)=r_i$ are constant, which yields that $T$ follows a hypoexponential distribution). Let $n=\sum_i k_i$ and further assume particles go to $Y_\ell$ with probability $p_{\ell}$ upon leaving X, $\ell=1,\ldots,m$. Using the GLCT framework above, this corresponds to partitioning X into sub-states X$_j$, where $j=1,\ldots,n$, and   
	
	\begin{equation}
	R(t)=[r_1(t),r_1(t),\ldots,r_2(t),\ldots,r_n(t)]^\text{T} 
	\end{equation} where the first $k_1$ elements of $R(t)$ are $r_1(t)$, the next $k_2$ are $r_2(t)$, etc., and   \begin{equation}  
	P_X=\begin{bmatrix} 0 & 1 & 0 & \cdots & 0 & 0 \\
	0 & 0 & 1 & \cdots & 0 & 0  \\
	\vdots & \vdots & \ddots & \ddots & \vdots & \vdots  \\
	0 & 0 & 0 & \ddots & 1 & 0 \\
	0 & 0 & 0 & \cdots & 0 & 1 \\
	0 & 0 & 0 & \cdots & 0 & 0 \\
	\end{bmatrix}_{n\times n} \hspace{-1.9em}, \qquad 
	P_Y =\begin{bmatrix} 0 & 0 & \cdots & 0 \\
	\vdots & \vdots & \ddots & \vdots \\
	0 & 0 & \cdots & 0 \\
	p_1 & p_2 & \cdots & p_m 
	\end{bmatrix}_{n\times \ell}\hspace{-1.75em}.
	\end{equation} 
	
	\clearpage
	By the GLCT (Theorem \ref{Th:glct}), using $r_{(j)}(t)$ to denote the $j^\text{th}$ element of $R(t)$, the corresponding mean field equations are 
	
	\begin{subequations} \begin{align}
		\frac{d}{dt}x_1(t) =&\; \mathcal{I}_X(t) - r_1(t)\,x_1(t) \\
		\frac{d}{dt}x_j(t) =&\;  r_{(j-1)}(t)\,x_{j-1}(t) - r_{(j)}(t)\,x_j(t), \text{ for } j\geq 1,\\
		\begin{split} y_\ell(t) =&\; y_\ell(0)S_{Y_\ell}(t,0) + \int_{0}^{t} \bigg ( \mathcal{I}_{Y_\ell}(\tau) + \\ & \qquad \qquad \qquad \qquad \quad  \sum_{j=1}^{m} r_{(j)}(t)\,x_j(\tau)\,p_{j}(t) \bigg ) S_{Y_\ell}(t,\tau)\,d\tau. \end{split}
		\end{align}
	\end{subequations}   
\end{example}

\begin{example}[Dwell time given by the maximum of independent Erlang random variables] \label{ex:tomultimax}
	Lastly, we consider an example that illustrates how the GLCT can provide a conceptually simpler framework for deriving ODEs relative to derivation from mean field integral equations. Here we assume the X dwell time obeys the maximum of multiple Erlang distributions. 
	
	Recall in \S\ref{sec:tomultimin} we considered a dwell time given by the minimum of $N$ Erlang distributions. Here we instead consider the case where the dwell time distribution is given by the maximum of multiple Erlang distributions, $T=\max(T_1,T_2)$ where $T_i\sim$Erlang($r_i,2$). For simplicity, assume the dwell time in a single recipient state Y is exponential with rate $\mu$. We again partition X according to which events (under the two independent homogeneous Poisson processes associated with each of $T_1$ and $T_2$) particles are awaiting, and index those sub-states accordingly (see Fig. \ref{fig:tomultimax}). These sub-states are X$_{11}$, X$_{21}$, X$_{12}$, X$_{*1}$, X$_{22}$, X$_{1*}$, X$_{*2}$, and X$_{2*}$, where a `$*$' in the $i^\text{th}$ index position indicates that particles in that sub-state have already had the $i^\text{th}$ Poisson process reach the $k_i^\text{th}$ event (in this case, the $2^\text{nd}$ event). Each such sub-state has exponentially distributed dwell times, but rates for these dwell time distributions differ (unlike the cases in \S \ref{sec:tomultimin} where all sub-states had the same rate):  the Poisson process rates for sub-states X$_{11}$, X$_{21}$, X$_{12}$, and X$_{22}$ are $r=r_1+r_2$ (see Fig. \ref{fig:tomultimax} and compare to Theorem \ref{Th:comp} and Fig. \ref{fig:tomultimin}), but the rate for the states X$_{1*}$ and X$_{2*}$ (striped circles in Fig. \ref{fig:tomultimax}) are $r_1$ , and for X$_{*1}$ and X$_{*2}$ (shaded circles in Fig. \ref{fig:tomultimax})) are $r_2$.

	In the context of the GLCT, let $\mathbf{x}(t) = $[$x_{11}(t)$, $x_{21}(t)$, $x_{12}(t)$, $x_{*1}(t)$, $x_{22}(t)$, $x_{1*}(t)$, $x_{*2}(t)$, $x_{2*}(t)]^\text{T}$  then by the assumptions above $R(t)=$[$r$, $r$, $r$, $r_2$, $r$, $r_1$, $r_2$, $r_1]^\text{T}$, $\mathcal{I}_\mathbf{X}(t)=[\mathcal{I}_X(t),0,\ldots,0]^\text{T}$, and denoting  $p_1\equiv r_1/r$ and $p_2\equiv r_2/r$ (\textit{\`a la} Theorem \ref{Th:equivcompmulti} in \S \ref{sec:equivcompmulti})    
	
	\begin{equation} 
	\left[ P_X  \; \middle\vert \; P_Y \right] = \left[\begin{matrix} 
	0 & p_1 & p_2 & 0 & 0 & 0 & 0 & 0 \\
	0 & 0 & 0 & p_1 & p_2 & 0 & 0 & 0 \\
	0 & 0 & 0 & 0 & p_1 & p_2 & 0 & 0 \\
	0 & 0 & 0 & 0 & 0 & 0 & 1 & 0 \\
	0 & 0 & 0 & 0 & 0 & 0 & p_1 & p_2 \\
	0 & 0 & 0 & 0 & 0 & 0 & 0 & 1 \\
	0 & 0 & 0 & 0 & 0 & 0 & 0 & 0 \\
	0 & 0 & 0 & 0 & 0 & 0 & 0 & 0 \\   \end{matrix} \, \middle \vert \, \begin{matrix} 
	0 \\
	0 \\
	0 \\
	0 \\
	0 \\
	0 \\
	1 \\
	1 \\	
	\end{matrix} \right].
	\end{equation} 
	
	Then by the GLCT (Theorem \ref{Th:glct}), the corresponding mean field ODEs are 
	
	\begin{subequations}	\begin{align}
		\frac{d}{dt}x_{11}(t) =&\; \mathcal{I}_X(t) - r\,x_{11}(t)  \\ % CHECK THESE BY ENTERING THEM INTO THE GLCT FORMULA IN MAXIMA
		\frac{d}{dt}x_{21}(t) =&\; r_1\,x_{11}(t) - r\,x_{21}(t)  \\
		\frac{d}{dt}x_{12}(t) =&\; r_2\,x_{11}(t) - r\,x_{12}(t)  \\
		\frac{d}{dt}x_{*1}(t) =&\; r_1\,x_{21}(t) - r_2\,x_{*1}(t) \\
		\frac{d}{dt}x_{22}(t) =&\; r_2\,x_{21}(t) + r_1\,x_{12}(t) - r\,x_{22}(t)  \\
		\frac{d}{dt}x_{1*}(t) =&\; r_2\,x_{12}(t) - r_1\,x_{1*}(t)  \\
		\frac{d}{dt}x_{*2}(t) =&\; r_2\,x_{*1}(t) + r_1\,x_{22}(t) - r_2\,x_{*2}(t)    \\
		\frac{d}{dt}x_{2*}(t) =&\; r_1\,x_{1*}(t) + r_2\,x_{22}(t) - r_1\,x_{2*}(t)  \\
		\frac{d}{dt}y(t)      =&\; r_1\,x_{2*}(t) + r_2\,x_{*2}(t) - \mu\,y(t).
		\end{align} \end{subequations}
	
\end{example}

\section{Discussion}\label{sec:discussion}

The above results generalize the Linear Chain Trick (LCT), and detail how to construct mean field ODE models for a broad range of scenarios found in applications. Our hope is that these contributions improve the speed and efficiency of constructing mean field ODE models, increase the flexibility to make more appropriate dwell time assumptions, and help clarify (for both modelers and those reading the results of their work) how individual-level stochastic assumptions are reflected in the structure of mean field ODE model equations. We have provided multiple novel theorems that describe how to construct such ODEs directly from underlying stochastic model assumptions, without formally deriving them from an explicit stochastic model or from intermediate integral equations. The Erlang distribution recursion relation (Lemma \ref{lem:smith711}) that drives the LCT has been generalized to include the time-varying analogues of Erlang distributions, i.e., $k^\text{th}$ event time distributions under nonhomogeneous Poisson processes (Lemma \ref{L:dh}), and distributions that reflect ``competing Poisson process even times" defined as the minimum of a finite number of independent Poisson process event times (Lemma \ref{L:df}). These new lemmas, and our generalization of the memorylessness property of the exponential distribution (which we refer to as the \textit{weak memorylessness} property of nonhomogeneous Poisson process 1$^{st}$ event time distributions) together allow a much broader class of dwell time distributions to be incorporated into mean field ODE models, including the phase-type family of distributions and their time-varying analogues. We have also introduced a novel generalized linear chain trick (GLCT; Theorem \ref{Th:glct} in \S\ref{sec:glct}) which complements previous extensions of the LCT \citep[e.g.,][]{Jacquez2002,Diekmann2017} and allows one to construct mean field ODE models for a broad class of dwell time distributions and sub-state configurations (e.g., conditional dwell time distributions for intermediate sub-state transitions). The GLCT also provides a framework for considering other scenarios not specifically addressed by  the above results, as illustrated by example \ref{ex:tomultimax} which assumes the dwell time distribution follows the maximum of multiple Erlang distributions. 

These results not only provide a framework to incorporate more accurate dwell time distributions into ODE models, but also hopefully encourage more comparative studies, such as \citet{Feng2016}, that explore the dynamic and application-specific consequences of incorporating non-Erlang dwell time distributions, and conditional dwell time distributions, into ODE models. The flexible phase-type family of distributions can be thought of as the hitting-time distributions for Continuous Time Markov Chains, and includes mixtures of Erlang distributions (a.k.a. hyper-Erlang distributions), the minimum or maximum of multiple Erlang distributions, the hypoexponential distributions, generalized Coxian distributions, and others \citep{Reinecke2012a,Horvath2016}. While the phase-type distributions are currently mostly unknown to mathematical biologists, they have received some attention in other fields and modelers can take advantage of existing methods that have been developed to fit phase-type distributions to other distributions on $\mathbb{R}^+$ and to data \citep{Asmussen1996, jPhase, Osogami2006, Thummler2006, Reinecke2012b, Mapfit, BuTools2}. These results provide a flexible framework for approximating dwell time distributions, and incorporating those empirically or analytically derived dwell time distributions into ODE models. That increased flexibility augments our capacity to investigate the dynamic and application-specific consequences of incorporating non-exponential and non-Erlang dwell time distributions into ODE models. 

There are some additional considerations, and potential challenges to implementing these results in applications, that are worth addressing. First, the increase in the number of state variables may lead to both computational and analytical challenges, however we have a growing number of tools at our disposal for tackling high dimensional systems. Second, it is tempting to assume that the sub-states resulting from the above theorems correspond to some sort of sub-state structure in the actual system being modeled. This is not necessarily the case, and we should be cautious about interpreting these sub-states as evidence of, e.g., cryptic population structure. Third, some of the above theorems make a simplifying assumption that, upon entry into X, the initial distribution of particles is only into the first sub-state. This may not be the appropriate assumption to make in some applications, but it is fairly straight forward to modify these these initial condition assumptions within the context of the GLCT. Fourth, in certain applications it may be more appropriate to avoid mean field models all together, and instead analyze the stochastic model dynamics directly \citep[e.g., see ][and references therin]{Allen2010, Allen2017}. Lastly, the history of successful attempts to garner scientific insights from mean field ODE models (i.e., those that assume only exponential and Erlang dwell time distributions) seems to suggest that such distributional refinements are unnecessary. However, this is clearly not always the case, as evidenced by studies that compare the results of models using simpler versus more realistic dwell time distributions (either via the LCT or through the use of integral or integrodifferential equations), and as evidenced by the many instances in which modelers have abandoned ODEs and instead opted to use integral equations to model systems with non-Erlang dwell time distributions. At a minimum, these results will allow a more rigorous comparison of such detailed models and their simplified counterparts to determine if using the simpler model is in fact warranted, e.g., as in \citet{Feng2016} and \citet{Piotrowska2018}.

In closing, these results introduce novel extensions of the LCT, and provide a means for incorporating more flexible dwell time distributions into mean field ODE models directly from first principles, without a need to derive ODEs from stochastic models or intermediate mean field integral equations. The Generalized Linear Chain Trick (GLCT) provides both a conceptual framework for understanding how individual-level stochastic assumptions are reflected in the structure of mean field model equations, and a practical framework for incorporating exact, empirically derived, or approximated dwell time distributions into mean field ODE models.

~\\
\textbf{Acknowledgments:} The authors thank Michael H. Cortez, Jim Cushing, Marisa E. Eisenberg, Jace Gilbert, Zoe Haskell, Tomasz Kozubowski, Catalina Medina, Amy Robards, Deena R. Schmidt, Joe Tien, and Narae Wadsworth for conversations, comments, and suggestions that improved this manuscript. This work was conducted while PJH was supported by start-up funds provided by the University of Nevada, Reno (UNR) Office of Research and Innovation.

\begin{appendices} %% Equations were starting at (1), so this next line changes them to (A1), ...
	\numberwithin{equation}{section}
	%\numberwithin{table}{section}
	%\numberwithin{figure}{section}
	\renewcommand{\theequation}{\Alph{section}\arabic{equation}}\setcounter{equation}{0} 
	%\renewcommand{\thetable}{\Alph{section}\arabic{table}}\setcounter{table}{0} 
	%\renewcommand{\thefigure}{\Alph{section}\arabic{figure}}\setcounter{figure}{0} 
	
	%\clearpage
	
	\section{Deterministic Models as Mean Field Equations}\label{sec:meanfield}

	To give some intuition for how mean field equations arise from stochastic state transition models, we here give a brief description of the process of deriving deterministic mean field equations from stochastic first principles. 
	
	Intuition for incorporating Erlang-distributed delays (or, equivalently, Erlang distributed dwell times) in ODE models begins by considering a stochastic model of discrete particles (e.g., individual organisms in a population, molecules in a solution, etc.) that transition among a finite number of $n$ states in continuous time. Building upon this, we let the state variables of interest be the amount of particles in each state, and then derive from the individual-level stochastic model gives a model for how these counts change over time in the mean field limit. This set of counts in each state can be thought of as a \textit{state vector} in the \textit{state space} $\mathbb{N}^n\subset\mathbb{R}^n$, and our model describes the (stochastic) rules governing transitions from one state vector to the next (i.e., from a given state there is some probability distribution across the state space describing how the system will proceed). Mean field models essentially average that distribution, and thus describe the mean state transitions from any given state of the system. That is, for a given state in $\mathbb{R}^n$, we can think of a probability distribution that describes where the system would move from that point in $\mathbb{R}^n$, and find the mean transition direction in state space, which then defines a deterministic dynamical system on $\mathbb{R}^n$ which we refer to as a mean field model for the given stochastic process. 
	
	More formally, let the $n$ state variables $\widetilde{x_i}(t)\in\mathbb{N}$ be the numbers of particles in the $i^\text{th}$ state at time $t\geq 0$. Assume $t$ takes on discrete time values that are integer multiples of the time step size $\Delta t$. The goal in deriving a mean field model of this stochastic process is to find the expected value of $\widetilde{x_i}(t)$, which we'll denote as $x_i(t)\equiv E(\widetilde{x_i}(t))$ (or in deriving a differential equation model, the expected change $\widetilde{x_i}(t)\mapsto \widetilde{x_i}(t+\Delta t)$). To derive a continuous time mean field model, we do this for an arbitrary step size $\Delta t$ so we can then take the limit as $\Delta t \to 0$. Note that the expected values $x_i(t)$ are real numbers, despite $\widetilde{x_i}(t)$ being integer-valued.
	
	In section \ref{a:example1} below, we derive integral equations for Example \ref{ex:simple}. Integral equations like eq. \eqref{eq:example1x} should be thought of as the $\Delta t \to 0$ limit of a Reimann sum that gives the expected number of particles entering a given state (X) in each of $M$ small time intervals over [$0,t$] (where $M\,\Delta t=t$), multiplied by the expected proportion remaining in X at time $t$. More specifically, the expected rate of particles entering state X during time interval $[s_j,s_j+\Delta t]$ (where $\tau$ is some integer multiple  of $\Delta t$) is given by the instantaneous input rate $\mathcal{I}_X(t)$ times $\Delta t$. The expected proportion of a cohort that enters X during that time interval and remains in state X at time $t$ is given by the survival function for the dwell time distribution over [$\tau,t$], give or take small error on the order of $\Delta t$.  Other more systematic approaches exist, e.g., see \citet{Kurtz1970,Kurtz1971}.

	\subsection{Derivation of mean field equations in Example \ref{ex:simple}}\label{a:example1}
	
	Here we derive the mean field equations \eqref{eq:example1} from Example \ref{ex:simple} in \S \ref{sec:base} starting from an explicit stochastic model. To do this, we begin by describing a discrete time approximation (with arbitrarily small time step $\Delta t$) of a system of particles transitioning among the various states, then derive the corresponding discrete-time mean field model which then yields the desired mean field integral equations and ODEs by taking the limit as $\Delta t \to 0$.
	
	In addition to the assumptions spelled out in the text above eqs.\eqref{eq:example1}, assume there are $w_0$ particles in state W at time $t=0$ (where $w_0\gg1$)), which independently transition from state W to state X after an exponentially distributed duration of time with rate $a$, then (again, independently) from X to state Y after an Erlang($r,k$) distributed duration of time, and then finally to state Z after an exponentially distributed amount of time with rate $\mu$. Let $\widetilde w(t)$, $\widetilde x(t)$, $\widetilde y(t)$, and $\widetilde z(t)$ be the amount in each of the corresponding states at time $t\geq0$, with $w(0)=N_0$, and $x(0)=y(0)=z(0)=0$. 
	
	First, to derive the linear ODEs \eqref{eq:example1w} and \eqref{eq:example1z}, note that the number of particles that transition from state W to state X in a short time interval $(t,t+\Delta t)$ is binomially distributed: if we think of a transition from W to X as a ``success" then the number of ``trials" $n=\widetilde w(t)$ and the probability of success $p$ is given by the exponential CDF value $p=1-\exp(-a\Delta t)$. For sufficiently small $\Delta t$ this implies $p=a\Delta t+O(\Delta t^2)$. Let $w(t+s) = E(\widetilde w(t+s)|\widetilde w(t))$ for $s\geq0$. Since the expected value of a binomial random variable is $np$ it follows that 
	
	\begin{equation}\begin{split}\label{eq:dw}
	w(t+\Delta t) - w(t) \equiv& E(\widetilde w(t+\Delta t) - \widetilde w(t)|\widetilde w(t)) \\
	=& - a\,w(t)\,\Delta t + O(\Delta t^2).
	\end{split}
	\end{equation}
	
	Dividing both sides of \eqref{eq:dw} by $\Delta t$ and then taking the limit as $\Delta t \to 0$ yields 
	
	\begin{equation} \frac{d}{dt} w(t) = -a\,w(t). \end{equation}
	
	Similarly, define $z(t)$ in terms of $\widetilde z(t)$ then it follows that 
	
	\begin{equation} \frac{d}{dt}z(t)= \;\mu\,y(t). \end{equation} 
	
	Next, we derive the integral equations \eqref{eq:example1x} and \eqref{eq:example1y} by similarly deriving a discrete time mean field model and then taking its limit as bin width $\Delta t \to 0$ (i.e., as the number of bins $M\to\infty$). 
	
	\begin{figure}[b]
		%\centerline{\includegraphics[trim={0 355 0 360},clip,width=0.5\textwidth]{discrete-time.pdf}}
		\centerline{\includegraphics[trim={180 250 185 185},clip,width=0.6\textwidth]{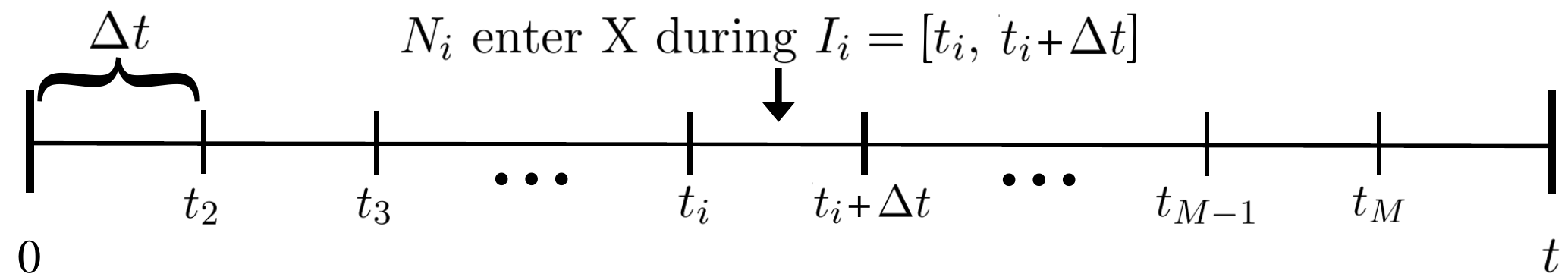}}
		\caption{Illustration of discretizing the time interval [$0,t$] into $M$ bins (each $\Delta t = t/M$ wide) in order to derive the integral equations \eqref{eq:example1x} and \eqref{eq:example1y} in Example \ref{ex:simple} (\S \ref{sec:base}) as detailed in Appendix \ref{a:example1}.}
		\label{fig:discretetime}
	\end{figure}
	
	Partition the interval $[0,t]$ into $M\gg1$ equally wide intervals of width $\Delta t\equiv t/M$ (see Fig. \ref{fig:discretetime}). Let $I_i=\big((i-1)\Delta t,i\Delta t\big]$ denote the $i^\text{th}$ such time interval ($i=1,\ldots,M$) and let $t_i=(i-1)\Delta t$ denote the start time of the $i^\text{th}$ interval. We may now account for the number transitioning into and out of state X during $I_i$, and sum across these values to compute $x(t)$.
	
	The number in state X at time $t$ ($\widetilde x(t)$) is the number that entered state X between time 0 and $t$, less the number that transitioned out of X before time $t$. A particle that enters state X at time $s\in (0,t)$ will still be in state X according to a Bernoulli random variable with $p=S_{r}^k(t-s)$ (the expected proportion under the given gamma distribution). Therefore, to compute $\widetilde x(t)$ we can sum over our $M$ intervals and add up the number that entered state X during interval $I_i$ and, from each of those $M$ cohorts, count how many remain in X at time $t$. Specifically, the number entering X during the $i^\text{th}$ interval [$t,\;t+\Delta t$] is given by $N_i\equiv \widetilde w(t_i+\Delta t) - \widetilde w(t_i)$ (see Fig. \ref{fig:discretetime}), and thus the number remaining in X at time $t$ is the sum of the number remaining at time $t$ from each such cohort (i.e., the sum over $i=1$ to $M$) where the number remaining in X at $t$ from each cohort follows a compound binomial distribution given by the sum of $N_i$ Bernoulli random variables $B_X$ each with probability $p=S_{r}^k(t-t_i)+O(\Delta t)$. This defines our stochastic state transition model, which yields a mean field model as follows.
	
	The expected amount entering X during [$t_i,t_i+\Delta t]$ is $E(N_i)=E(\widetilde w(t_i) - \widetilde w(t_i+\Delta t))=a\,w(t_i)\Delta t$, and the expected proportion of the $i^\text{th}$ cohort remaining at time $t$ is $E(B_X)=S_{r}^k(t-t_i) + O(\Delta t)$. Thus, the expected number from the $i^\text{th}$ cohort remaining in X at time $t$ is $E(N_i)E(B_X)=a\,w(t_i)\,S_{r}^k(t-t_i)\,\Delta t + O(\Delta t^2)$. Summing these expected values over all intervals yields
	\begin{equation}\begin{split}
	x(t) = E(\widetilde x(t)) = & \lim_{M\to \infty} \sum_{i=1}^M a\,w(t_i)\,S_{r}^k(t-t_i)\,\Delta t +O(\Delta t^2) \\  =&\; \int_0^t a\,w(s)\,S_{r}^k(t-s)\,ds.
	\end{split}\end{equation}
	To calculate $y(t)$, let $K_j$ be the number entering state Y during interval $I_j$ that are still in state Y at time $t$. As above, $K_j$ can be calculated by summing (over $I_1$ to $I_{j-1}$) the number in each cohort that entered state X during $I_i$ then transitioned to state Y during time interval $I_{j}$ and are still in state Y at $t$. Therefore $K_j$ can be written as the sum of $j-1$ compound distributions given by counting how many of the $N_i$ particles that entered state X during $I_i$ then transitioned to state Y during interval $I_j$ and then persisted until time $t$ without transitioning to Z.  To count these, notice that each such particle entering state X during $I_i$, state Y during $I_j$ and persisting in Y at time $t$ follows a Bernoulli random variable $B_{ij}$ with probability  
	
	\begin{equation} 
	p_{ij}=\underbrace{(S_r^k(t_j-t_i)-S_r^k(t_{j+1}-t_i)+O(\Delta t))}_\text{P(X$\to$Y in $I_j|$W$\to$X in $I_i$)}\underbrace{(S_\mu^1(t-t_j)+O(\Delta t))}_\text{P(still in Y at $t$)}.
	\end{equation} Therefore, the number of particles that entered Y at $I_j$ and remain in state Y at time $t$, $K_j$, can be written as a compound random variable $K_j= \sum_{i=1}^{j-1}  \sum_{k=1}^{N_i} B(p_{ij})$. The expected $K_j$ is thus 
	
	\begin{equation}\begin{split}
	E&(K_j) =	 \sum_{i=1}^{j-1}  E(N_i) E(B_{ij}) \\
	=& \sum_{i=1}^{j-1}  (a\,w(t_i)\Delta t) [(S_r^k(t_j-t_i)-S_r^k(t_{j+1}-t_i))S_\mu^1(t-t_j)+O(\Delta t)] \\
	%=& \sum_{i=1}^{j-1}  a\,w(t_i) [(S_r^k(t_j-t_i)-S_r^k(t_{j+1}-t_i))S_\mu^1(t-t_j)\Delta t+O(\Delta t^2)] \\
	=& \sum_{i=1}^{j-1}   a\,w(t_i) \frac{(S_r^k(t_j-t_i)-S_r^k(t_j-t_i+\Delta t))}{\Delta t}\,S_\mu^1(t-t_j)\,\Delta t^2 +O(\Delta t^2) \\
	=& \sum_{i=1}^{j-1}   a\,w(t_i) \frac{(G_r^k(t_j-t_i+\Delta t)-G_r^k(t_j-t_i))}{\Delta t}\,S_\mu^1(t-t_j)\,\Delta t^2 +O(\Delta t^2).
	\end{split}\end{equation} 
	
	Summing over all intervals and letting $\Delta t \to 0$ ($M \to \infty$) gives 
	
	\begin{equation} \begin{split}
	y(t)=& \lim_{M\to \infty}  E(\widetilde y(t)) = \lim_{M\to \infty}  E\bigg(\sum_{j=1}^M K_j\bigg) = \lim_{M\to \infty}  \sum_{j=1}^M E(K_j) \\
	=&\lim_{M\to \infty} \sum_{j=1}^M\sum_{i=1}^{j-1}  a\,w(t_i) \frac{(G_r^k(t_j-t_i+\Delta t)-G_r^k(t_j-t_i))}{\Delta t}\,\Delta t\\ 
	& \qquad \qquad \qquad \qquad \qquad \qquad \qquad \qquad \cdot S_\mu^1(t-t_j)\,\Delta t +O(\Delta t^2) \\
	=&\int_0^t\left(\int_0^\tau a\,w(s)\,g_r^k(\tau-s)\,ds\right) S_\mu^1(t-\tau)\,d\tau.
	\end{split}\end{equation}
	%\hfill $\qed$

	\section{Erlang mixture approximation of Gamma($\alpha,\beta$)}\label{sec:approxgamma}
	
	There is a growing body of literature on methods for approximating empirical and named distributions with mixtures of Erlang random variables or other phase-type distributions \citep{Asmussen1996,jPhase,Osogami2006,Thummler2006,Reinecke2012a,Mapfit,BuTools2}. Here we give a simple example of analytically approximating a gamma distribution with a mixture of two Erlang distributions by matching moments.
	
	Suppose random variable $T$ follows a gamma($\alpha,\beta$) distribution, which has mean $\mu=\beta/\alpha$ and variance $\sigma^2=\mu/\alpha$, and shape $\beta$ is not an integer. One can approximate this gamma distribution with a mixture of two Erlang distributions that yields the same mean and variance. 
	
	These Erlang distributions are T$_{\downarrow}\sim$Erlang($r_{\downarrow},k_{\downarrow}$) and T$_{\uparrow}\sim$Erlang($r_{\uparrow},k_{\uparrow}$) where the shape parameters are obtained by rounding $\beta$ down and up, respectively, to the nearest integer ($k_\downarrow \equiv \lfloor\beta\rfloor$ and $k_\uparrow \equiv \lceil\beta\rceil$) and the rate parameters are given by $r_\downarrow=\alpha\frac{\lfloor \beta \rfloor}{\beta}$ ($r_{\downarrow} = \frac{k_\downarrow}{\mu}$) and $r_\uparrow=\alpha\frac{\lceil \beta \rceil}{\beta}$ ($r_{\uparrow} \equiv \frac{k_\uparrow}{\mu}$).  This ensures that T$_{\downarrow}$ and T$_{\uparrow}$ have mean $\mu$. 
	
	To calculate their variance, let $p\equiv \lceil\beta\rceil-\beta$ and $q\equiv \beta-\lfloor\beta\rfloor$ (note $p+q=1$). By rounding shape $\beta$ down(up) the resulting Erlang distribution has higher(lower) variance, i.e.,  \begin{equation} \sigma^2_\downarrow\equiv\frac{\mu}{r_\downarrow}=\sigma^2 \bigg( 1+\frac{q}{k_\downarrow} \bigg) \quad \text{ and } \quad  \sigma^2_\uparrow\equiv\frac{\mu}{r_\uparrow}=\sigma^2 \bigg( 1-\frac{p}{k_\uparrow} \bigg). \end{equation} 
	
	To calculate the mixing proportion, let the mixture distribution $T_\rho = B_\rho\,T_\downarrow + (1-B_\rho)\,T_\uparrow$, where $B_\rho$ is a Bernoulli random variable with $P(B_p=1)=\rho$ and \begin{equation}\rho=\frac{p/k_\uparrow}{p/k_\uparrow+q/k_\downarrow}.
	\end{equation} This Erlang mixture has the desired mean $E(T_\rho)=\rho\,\mu + (1-\rho)\,\mu=\mu$ and variance $\sigma^2$, since \begin{equation} \begin{split}
	\text{Var}(T_\rho)=&\; E\big((B_\rho\,T_\downarrow + (1-B_\rho)\,T_\uparrow)^2\big))-\mu^2\\
	=&\;E\big((B_\rho\,T_\downarrow)^2\big)+E\big(((1-B_\rho)\,T_\uparrow)^2\big)-\mu^2\\
	=&\;E\big(B_\rho^2\big)E\big(T_\downarrow^2\big)+E\big((1-B_\rho)^2\big)E\big(T_\uparrow^2\big)-\mu^2\\
	=&\;E\big(B_\rho\big)E\big(T_\downarrow^2\big)+E\big(1-B_\rho\big)E\big(T_\uparrow^2\big)-\mu^2\\
	=&\;\rho\,E\big(T_\downarrow^2\big)+(1-\rho)\,E\big(T_\uparrow^2\big)-\mu^2\\
	=&\;\rho\,\sigma_\downarrow^2+(1-\rho)\,\sigma_\uparrow^2\\
	=&\;\rho\,\sigma^2\bigg(1+\frac{q}{k_\downarrow}\bigg)+(1-\rho)\,\sigma^2\bigg(1-\frac{p}{k_\uparrow}\bigg)\\
	=&\;\sigma^2 + \sigma^2\bigg(\frac{\rho\,q}{k_\downarrow}-\frac{(1-\rho)\,p}{k_\uparrow}\bigg)\\
	=&\;\sigma^2 + \sigma^2\bigg(\frac{p/k_\uparrow\cdot q/k_\downarrow}{p/k_\uparrow+q/k_\downarrow}-\frac{q/k_\downarrow\cdot p/k_\uparrow }{p/k_\uparrow+q/k_\downarrow}\bigg) = \;  \sigma^2.
	\end{split} \end{equation} Numerical comparisons suggest this mixture is a very good approximation of the target gamma($\alpha,\beta$) distribution for shape $\beta$ values larger than roughly 3 to 5, depending (to a lesser extent) on $\alpha$.
	
	Alternatively, to approximate a gamma($\alpha,\beta$) distribution with a mixture of Erlang distributions as described above, one could also select the mixing probabilities by, for example, using an alternative metric such as a distance in probability space \citep[e.g., see][]{Rachev1991}, e.g., the $L^\infty$-norm on their CDFs, or information-theoretic quantities such as KL or Jensen-Shannon divergence.
	
\end{appendices}

\clearpage
\bibliography{./bibreferences}

\begin{thebibliography}{68}
\providecommand{\natexlab}[1]{#1}
\providecommand{\url}[1]{{#1}}
\providecommand{\urlprefix}{URL }
\expandafter\ifx\csname urlstyle\endcsname\relax
  \providecommand{\doi}[1]{DOI~\discretionary{}{}{}#1}\else
  \providecommand{\doi}{DOI~\discretionary{}{}{}\begingroup
  \urlstyle{rm}\Url}\fi
\providecommand{\eprint}[2][]{\url{#2}}

\bibitem[{Allen(2017)}]{Allen2017}
Allen LJ (2017) {A primer on stochastic epidemic models: Formulation, numerical
  simulation, and analysis}. Infectious Disease Modelling 2(2):128 -- 142,
  \doi{10.1016/j.idm.2017.03.001}

\bibitem[{Allen(2010)}]{Allen2010}
Allen LJS (2010) {An Introduction to Stochastic Processes with Applications to
  Biology}, 2nd edn. Chapman and Hall/CRC

\bibitem[{Anderson and Watson(1980)}]{Anderson1980}
Anderson D, Watson R (1980) On the spread of a disease with gamma distributed
  latent and infectious periods. Biometrika 67(1):191--198,
  \doi{10.1093/biomet/67.1.191}

\bibitem[{Anderson and May(1992)}]{AndersonMay1992}
Anderson RM, May RM (1992) {Infectious Diseases of Humans: Dynamics and
  Control}. Oxford University Press

\bibitem[{Armbruster and Beck(2017)}]{Armbruster2017}
Armbruster B, Beck E (2017) {Elementary proof of convergence to the mean-field
  model for the SIR process}. Journal of Mathematical Biology 75(2):327--339,
  \doi{10.1007/s00285-016-1086-1}

\bibitem[{Asmussen et~al.(1996)Asmussen, Nerman, and Olsson}]{Asmussen1996}
Asmussen S, Nerman O, Olsson M (1996) {Fitting Phase-Type Distributions via the
  EM Algorithm}. Scandinavian Journal of Statistics 23(4):419--441

\bibitem[{Banks et~al.(2013)Banks, Catenacci, and Hu}]{Banks2013}
Banks HT, Catenacci J, Hu S (2013) A comparison of stochastic systems with
  different types of delays. Stochastic Analysis and Applications
  31(6):913--955, \doi{10.1080/07362994.2013.806217}

\bibitem[{Blythe et~al.(1984)Blythe, Nisbet, and Gurney}]{Blythe1984}
Blythe S, Nisbet R, Gurney W (1984) The dynamics of population models with
  distributed maturation periods. Theoretical Population Biology 25(3):289 --
  311, \doi{10.1016/0040-5809(84)90011-X}

\bibitem[{Boese(1989)}]{Boese1989}
Boese F (1989) {The stability chart for the linearized Cushing equation with a
  discrete delay and with gamma-distributed delays}. Journal of Mathematical
  Analysis and Applications 140(2):510 -- 536,
  \doi{10.1016/0022-247X(89)90081-4}

\bibitem[{Burton(2005)}]{Burton2005}
Burton TA (2005) Volterra Integral and Differential Equations, Mathematics in
  Science and Engineering, vol 202, 2nd edn. Elsevier

\bibitem[{C{\^a}mara De~Souza et~al.(2018)C{\^a}mara De~Souza, Craig, Cassidy,
  Li, Nekka, B{\'e}lair, and Humphries}]{CamaraDeSouza2018}
C{\^a}mara De~Souza D, Craig M, Cassidy T, Li J, Nekka F, B{\'e}lair J,
  Humphries AR (2018) Transit and lifespan in neutrophil production:
  implications for drug intervention. Journal of Pharmacokinetics and
  Pharmacodynamics 45(1):59--77, \doi{10.1007/s10928-017-9560-y}

\bibitem[{Campbell and Jessop(2009)}]{Campbell2009}
Campbell SA, Jessop R (2009) Approximating the stability region for a
  differential equation with a distributed delay. Math Model Nat Phenom
  4(2):1--27, \doi{10.1051/mmnp/20094201}

\bibitem[{Champredon et~al.(2018)Champredon, Dushoff, and
  Earn}]{Champredon2018}
Champredon D, Dushoff J, Earn D (2018) {Equivalence of the Erlang SEIR epidemic
  model and the renewal equation}. bioRxiv \doi{10.1101/319574}

\bibitem[{Ciaravino et~al.(2018)Ciaravino, García-Saenz, Cabras, Allepuz,
  Casal, García-Bocanegra, Koeijer, Gubbins, Sáez, Cano-Terriza, and
  Napp}]{Ciaravino2018}
Ciaravino G, García-Saenz A, Cabras S, Allepuz A, Casal J, García-Bocanegra
  I, Koeijer AD, Gubbins S, Sáez J, Cano-Terriza D, Napp S (2018) {Assessing
  the variability in transmission of bovine tuberculosis within Spanish cattle
  herds}. Epidemics 23:110 -- 120, \doi{10.1016/j.epidem.2018.01.003}

\bibitem[{Clapp and Levy(2015)}]{Clapp2015}
Clapp G, Levy D (2015) A review of mathematical models for leukemia and
  lymphoma. Drug Discovery Today: Disease Models 16:1 -- 6,
  \doi{10.1016/j.ddmod.2014.10.002}

\bibitem[{Cushing(1994)}]{Cushing1994}
Cushing JM (1994) The dynamics of hierarchical age-structured populations.
  Journal of Mathematical Biology 32(7):705--729, \doi{10.1007/BF00163023}

\bibitem[{Diekmann et~al.(2017)Diekmann, Gyllenberg, and Metz}]{Diekmann2017}
Diekmann O, Gyllenberg M, Metz JAJ (2017) Finite dimensional state
  representation of linear and nonlinear delay systems. Journal of Dynamics and
  Differential Equations \doi{10.1007/s10884-017-9611-5}

\bibitem[{Fargue(1973)}]{Fargue1973}
Fargue D (1973) R\'eductibilit\'e des syst\`emes h\'er\'editaires \`a des
  syst\`emes dynamiques (r\'egis par des \'equations diff\'erentielles ou aux
  d\'eriv\'ees partielles). C R Acad Sci Paris S\'er A-B 277:B471--B473

\bibitem[{Feng and Thieme(2000)}]{Feng2000}
Feng Z, Thieme H (2000) {Endemic Models with Arbitrarily Distributed Periods of
  Infection I: Fundamental Properties of the Model}. SIAM Journal on Applied
  Mathematics 61(3):803--833, \doi{10.1137/S0036139998347834}

\bibitem[{Feng et~al.(2007)Feng, Xu, and Zhao}]{Feng2007}
Feng Z, Xu D, Zhao H (2007) Epidemiological models with non-exponentially
  distributed disease stages and applications to disease control. Bulletin of
  Mathematical Biology 69(5):1511--1536, \doi{10.1007/s11538-006-9174-9}

\bibitem[{Feng et~al.(2016)Feng, Zheng, Hernandez-Ceron, Zhao, Glasser, and
  Hill}]{Feng2016}
Feng Z, Zheng Y, Hernandez-Ceron N, Zhao H, Glasser JW, Hill AN (2016)
  {Mathematical models of Ebola-Consequences of underlying assumptions}.
  Mathematical biosciences 277:89--107

\bibitem[{Fenton et~al.(2006)Fenton, Lello, and Bonsall}]{Fenton2006}
Fenton A, Lello J, Bonsall M (2006) Pathogen responses to host immunity: the
  impact of time delays and memory on the evolution of virulence. Proceedings
  Biological sciences 273(1597):2083--2090, \doi{10.1098/rspb.2006.3552}

\bibitem[{Goltser and Domoshnitsky(2013)}]{Goltser2013}
Goltser Y, Domoshnitsky A (2013) About reducing integro-differential equations
  with infinite limits of integration to systems of ordinary differential
  equations. Advances in Difference Equations 2013(1):187,
  \doi{10.1186/1687-1847-2013-187}

\bibitem[{Guan and Ling(2018)}]{Guan2018}
Guan ZH, Ling G (2018) Dynamic Analysis of Genetic Regulatory Networks with
  Delays, Springer Berlin Heidelberg, Berlin, Heidelberg, pp 285--309.
  \doi{10.1007/978-3-662-55663-4_14}

\bibitem[{Gyllenberg(2007)}]{Gyllenberg2007}
Gyllenberg M (2007) {Mathematical aspects of physiologically structured
  populations: the contributions of J. A. J. Metz}. Journal of Biological
  Dynamics 1(1):3--44, \doi{10.1080/17513750601032737}

\bibitem[{Hethcote and Tudor(1980)}]{Hethcote1980}
Hethcote HW, Tudor DW (1980) Integral equation models for endemic infectious
  diseases. Journal of Mathematical Biology 9(1):37--47,
  \doi{10.1007/BF00276034}

\bibitem[{Horv{\'a}th et~al.(2016)Horv{\'a}th, Scarpa, and Telek}]{Horvath2016}
Horv{\'a}th A, Scarpa M, Telek M (2016) Phase Type and Matrix Exponential
  Distributions in Stochastic Modeling, Springer International Publishing,
  Cham, pp 3--25. \doi{10.1007/978-3-319-30599-8_1}

\bibitem[{Horv{\'a}th and Telek(2017)}]{BuTools2}
Horv{\'a}th G, Telek M (2017) {BuTools 2: A rich toolbox for Markovian
  performance evaluation}. In: ValueTools 2016 - 10th EAI International
  Conference on Performance Evaluation Methodologies and Tools, Association for
  Computing Machinery, pp 137--142, \doi{10.4108/eai.25-10-2016.2266400}

\bibitem[{Jacquez and Simon(2002)}]{Jacquez2002}
Jacquez JA, Simon CP (2002) Qualitative theory of compartmental systems with
  lags. Mathematical Biosciences 180(1):329 -- 362,
  \doi{10.1016/S0025-5564(02)00131-1}

\bibitem[{Kermack and McKendrick(1927)}]{Kermack1927}
Kermack WO, McKendrick AG (1927) {A Contribution to the Mathematical Theory of
  Epidemics}. Proceedings of the Royal Society of London Series A, Containing
  Papers of a Mathematical and Physical Character 115(772):700--721

\bibitem[{Krylova and Earn(2013)}]{Krylova2013}
Krylova O, Earn DJD (2013) Effects of the infectious period distribution on
  predicted transitions in childhood disease dynamics. Journal of The Royal
  Society Interface 10(84), \doi{10.1098/rsif.2013.0098}

\bibitem[{Krzyzanski et~al.(2018)Krzyzanski, Hu, and Dunlavey}]{Krzyzanski2018}
Krzyzanski W, Hu S, Dunlavey M (2018) Evaluation of performance of distributed
  delay model for chemotherapy-induced myelosuppression. Journal of
  Pharmacokinetics and Pharmacodynamics 45(2):329--337,
  \doi{10.1007/s10928-018-9575-z}

\bibitem[{Kurtz(1970)}]{Kurtz1970}
Kurtz TG (1970) {Solutions of Ordinary Differential Equations as Limits of Pure
  Jump Markov Processes}. Journal of Applied Probability 7(1):49--58

\bibitem[{Kurtz(1971)}]{Kurtz1971}
Kurtz TG (1971) {Limit theorems for sequences of jump Markov processes
  approximating ordinary differential processes}. Journal of Applied
  Probability 8(2):344--356, \doi{10.2307/3211904}

\bibitem[{Lin et~al.(2018)Lin, Wang, and Wolkowicz}]{Lin2018}
Lin CJ, Wang L, Wolkowicz GSK (2018) An alternative formulation for a
  distributed delayed logistic equation. Bulletin of Mathematical Biology
  80(7):1713--1735, \doi{10.1007/s11538-018-0432-4}

\bibitem[{Lloyd(2001{\natexlab{a}})}]{Lloyd2001a}
Lloyd AL (2001{\natexlab{a}}) Destabilization of epidemic models with the
  inclusion of realistic distributions of infectious periods. Proceedings of
  the Royal Society of London B: Biological Sciences 268(1470):985--993,
  \doi{10.1098/rspb.2001.1599}

\bibitem[{Lloyd(2001{\natexlab{b}})}]{Lloyd2001b}
Lloyd AL (2001{\natexlab{b}}) {Realistic Distributions of Infectious Periods in
  Epidemic Models: Changing Patterns of Persistence and Dynamics}. Theoretical
  Population Biology 60(1):59 -- 71, \doi{10.1006/tpbi.2001.1525}

\bibitem[{Lloyd(2009)}]{Lloyd2009}
Lloyd AL (2009) Sensitivity of Model-Based Epidemiological Parameter Estimation
  to Model Assumptions, Springer Netherlands, Dordrecht, pp 123--141.
  \doi{10.1007/978-90-481-2313-1_6}

\bibitem[{Ma and Earn(2006)}]{Ma2006}
Ma J, Earn DJD (2006) Generality of the final size formula for an epidemic of a
  newly invading infectious disease. Bulletin of Mathematical Biology
  68(3):679--702, \doi{10.1007/s11538-005-9047-7}

\bibitem[{MacDonald(1978{\natexlab{a}})}]{MacDonald1978}
MacDonald N (1978{\natexlab{a}}) {Time Lags in Biological Models}, Lecture
  Notes in Biomathematics, vol~27. Springer-Verlag Berlin Heidelberg,
  \doi{10.1007/978-3-642-93107-9}

\bibitem[{MacDonald(1978{\natexlab{b}})}]{MacDonald1978ch2}
MacDonald N (1978{\natexlab{b}}) {Time Lags in Biological Models}, Lecture
  Notes in Biomathematics, vol~27, Springer-Verlag Berlin Heidelberg, chap
  Stability Analysis, pp 13--38. \doi{10.1007/978-3-642-93107-9_2}

\bibitem[{MacDonald(1989)}]{MacDonald1989}
MacDonald N (1989) Biological Delay Systems: Linear Stability Theory, Cambridge
  Studies in Mathematical Biology, vol~8. Cambridge University Press

\bibitem[{Makroglou et~al.(2006)Makroglou, Li, and Kuang}]{Makroglou2006}
Makroglou A, Li J, Kuang Y (2006) Mathematical models and software tools for
  the glucose-insulin regulatory system and diabetes: an overview. Applied
  Numerical Mathematics 56(3):559 -- 573, \doi{10.1016/j.apnum.2005.04.023}

\bibitem[{Metz and Diekmann(1991)}]{Metz1991}
Metz J, Diekmann O (1991) {Exact finite dimensional representations of models
  for physiologically structured populations. I: The abstract formulation of
  linear chain trickery}. In: Goldstein JA, Kappel F, Schappacher W (eds)
  Proceedings of Differential Equations With Applications in Biology, Physics,
  and Engineering 1989, vol 133, pp 269--289

\bibitem[{Metz and Diekmann(1986)}]{Metz1986}
Metz JAJ, Diekmann O (eds)  (1986) The Dynamics of Physiologically Structured
  Populations, Lecture Notes in Biomathematics, vol~68. Springer, Berlin,
  Heidelberg, \doi{10.1007/978-3-662-13159-6}

\bibitem[{Nisbet et~al.(1989)Nisbet, Gurney, and Metz}]{Nisbet1989}
Nisbet RM, Gurney WSC, Metz JAJ (1989) Stage Structure Models Applied in
  Evolutionary Ecology, Springer Berlin Heidelberg, Berlin, Heidelberg, pp
  428--449. \doi{10.1007/978-3-642-61317-3_18}

\bibitem[{Okamura and Dohi(2015)}]{Mapfit}
Okamura H, Dohi T (2015) {Mapfit: An R-Based Tool for PH/MAP Parameter
  Estimation}. In: Campos J, Haverkort BR (eds) Proceedings of the 12th
  International Conference on Quantitative Evaluation of Systems,
  Springer-Verlag New York, Inc., New York, NY, USA, QEST 2015, vol 9259, pp
  105--112, \doi{10.1007/978-3-319-22264-6_7}

\bibitem[{Osogami and Harchol-Balter(2006)}]{Osogami2006}
Osogami T, Harchol-Balter M (2006) {Closed form solutions for mapping general
  distributions to quasi-minimal PH distributions}. Performance Evaluation
  63(6):524 -- 552, \doi{10.1016/j.peva.2005.06.002}

\bibitem[{\"{O}zbay et~al.(2008)\"{O}zbay, Bonnet, and
  Clairambault}]{Ozbay2008}
\"{O}zbay H, Bonnet C, Clairambault J (2008) Stability analysis of systems with
  distributed delays and application to hematopoietic cell maturation dynamics.
  In: 2008 47th IEEE Conference on Decision and Control, pp 2050--2055,
  \doi{10.1109/CDC.2008.4738654}

\bibitem[{P{\'e}rez and Ria\~{n}o(2006)}]{jPhase}
P{\'e}rez JF, Ria\~{n}o G (2006) {jPhase: An Object-oriented Tool for Modeling
  Phase-type Distributions}. In: Proceeding from the 2006 Workshop on Tools for
  Solving Structured Markov Chains, ACM, New York, NY, USA, SMCtools '06,
  \doi{10.1145/1190366.1190370}

\bibitem[{Piotrowska and Bodnar(2018)}]{Piotrowska2018}
Piotrowska M, Bodnar M (2018) Influence of distributed delays on the dynamics
  of a generalized immune system cancerous cells interactions model.
  Communications in Nonlinear Science and Numerical Simulation 54:389 -- 415,
  \doi{10.1016/j.cnsns.2017.06.003}

\bibitem[{Ponosov et~al.(2002)Ponosov, Shindiapin, and Miguel}]{Ponosov2002}
Ponosov A, Shindiapin A, Miguel JJ (2002) {The W-transform links delay and
  ordinary differential equations}. Functional Differential Equations
  9(3-4):437--469

\bibitem[{Rachev(1991)}]{Rachev1991}
Rachev ST (1991) Probability metrics and the stability of stochastic models.
  Wiley Series in Probability and Mathematical Statistics, Wiley

\bibitem[{Reinecke et~al.(2012{\natexlab{a}})Reinecke, Bodrog, and
  Danilkina}]{Reinecke2012a}
Reinecke P, Bodrog L, Danilkina A (2012{\natexlab{a}}) {Phase-Type
  Distributions}, Springer Berlin Heidelberg, Berlin, Heidelberg, pp 85--113.
  \doi{10.1007/978-3-642-29032-9_5}

\bibitem[{Reinecke et~al.(2012{\natexlab{b}})Reinecke, Krauß, and
  Wolter}]{Reinecke2012b}
Reinecke P, Krauß T, Wolter K (2012{\natexlab{b}}) Cluster-based fitting of
  phase-type distributions to empirical data. Computers \& Mathematics with
  Applications 64(12):3840 -- 3851, \doi{10.1016/j.camwa.2012.03.016}

\bibitem[{Robertson et~al.(2018)Robertson, Henson, Robertson, and
  Cushing}]{Robertson2018}
Robertson SL, Henson SM, Robertson T, Cushing JM (2018) {A matter of maturity:
  To delay or not to delay? Continuous-time compartmental models of structured
  populations in the literature 2000-2016}. Natural Resource Modeling
  31(1):e12160, \doi{10.1111/nrm.12160}

\bibitem[{Roussel(1996)}]{Roussel1996}
Roussel MR (1996) The use of delay differential equations in chemical kinetics.
  The Journal of Physical Chemistry 100(20):8323--8330, \doi{10.1021/jp9600672}

\bibitem[{Smith(2010)}]{Smith2010}
Smith H (2010) An introduction to delay differential equations with
  applications to the life sciences, vol~57. Springer Science \& Business Media

\bibitem[{Smolen et~al.(2000)Smolen, Baxter, and Byrne}]{Smolen2000}
Smolen P, Baxter DA, Byrne JH (2000) Modeling transcriptional control in gene
  networks---methods, recent results, and future directions. Bulletin of
  Mathematical Biology 62(2):247--292, \doi{10.1006/bulm.1999.0155}

\bibitem[{Strogatz(2014)}]{Strogatz2014}
Strogatz SH (2014) {Nonlinear Dynamics and Chaos: With Applications to Physics,
  Biology, Chemistry, and Engineering}, 2nd edn. Studies in Nonlinearity,
  Westview Press

\bibitem[{Takashima et~al.(2011)Takashima, Ohtsuka, Gonz{\'a}lez, Miyachi, and
  Kageyama}]{Takashima2011}
Takashima Y, Ohtsuka T, Gonz{\'a}lez A, Miyachi H, Kageyama R (2011) Intronic
  delay is essential for oscillatory expression in the segmentation clock.
  Proceedings of the National Academy of Sciences 108(8):3300--3305,
  \doi{10.1073/pnas.1014418108}

\bibitem[{Thummler et~al.(2006)Thummler, Buchholz, and Telek}]{Thummler2006}
Thummler A, Buchholz P, Telek M (2006) {A Novel Approach for Phase-Type Fitting
  with the EM Algorithm}. IEEE Transactions on Dependable and Secure Computing
  3(3):245--258, \doi{10.1109/TDSC.2006.27}

\bibitem[{Vogel(1961)}]{Vogel1961}
Vogel T (1961) Syst{\`e}mes d{\'e}ferlants, syst{\`e}mes h{\'e}r{\'e}ditaires,
  syst{\`e}mes dynamiques. In: Proceedings of the international symposium
  nonlinear vibrations, IUTAM, Kiev, pp 123--130

\bibitem[{Vogel(1965)}]{Vogel1965}
Vogel T (1965) Th{\'e}orie des Syst{\`e}mes {\'E}volutifs. No.~22 in Trait\'{e}
  de physique th\'{e}orique et de physique math\'{e}matique, Gauthier-Villars,
  Paris

\bibitem[{Wang and Han(2016)}]{Wang2016}
Wang N, Han M (2016) {Slow-fast dynamics of Hopfield spruce-budworm model with
  memory effects}. Advances in Difference Equations 2016(1):73,
  \doi{10.1186/s13662-016-0804-8}

\bibitem[{Wearing et~al.(2005)Wearing, Rohani, and Keeling}]{Wearing2005}
Wearing HJ, Rohani P, Keeling MJ (2005) Appropriate models for the management
  of infectious diseases. PLOS Medicine 2(7),
  \doi{10.1371/journal.pmed.0020174}

\bibitem[{Wolkowicz et~al.(1997)Wolkowicz, Xia, and Ruan}]{Wolkowicz1997}
Wolkowicz G, Xia H, Ruan S (1997) {Competition in the Chemostat: A Distributed
  Delay Model and Its Global Asymptotic Behavior}. SIAM Journal on Applied
  Mathematics 57(5):1281--1310, \doi{10.1137/S0036139995289842}

\bibitem[{Yates et~al.(2017)Yates, Ford, and Mort}]{Yates2017}
Yates CA, Ford MJ, Mort RL (2017) {A Multi-stage Representation of Cell
  Proliferation as a Markov Process}. Bulletin of Mathematical Biology
  79(12):2905--2928, \doi{10.1007/s11538-017-0356-4}

\end{thebibliography}

\end{document}